\documentclass[11pt,a4paper,notitlepage]{article}

\usepackage{amsmath}        % extensions for typesetting of math
\usepackage{amsfonts}       % math fonts
\usepackage{amsthm}         % theorems, definitions, etc.
\usepackage{bbding}         % various symbols (squares, asterisks, scissors, ...)
\usepackage{bm}             % boldface symbols (\bm)
\usepackage{graphicx}       % embedding of pictures
\usepackage[square,sort,comma,numbers]{natbib}         % citation style AUTHOR (YEAR), or AUTHOR [NUMBER]
\usepackage[nottoc]{tocbibind} % makes sure that bibliography and the lists of figures/tables are included in the table of contents
\usepackage{paralist}       % improved enumerate and itemize
\usepackage[usenames]{xcolor}  % typesetting in color

\usepackage{xfrac}
\usepackage{tikz}
\usepackage{tikz-cd}

\usetikzlibrary{arrows,shapes,positioning}
\usetikzlibrary{decorations.markings}

\usetikzlibrary{arrows, matrix}

\usepackage{wrapfig}
\usepackage{caption}
\usepackage{titlesec}
\usepackage{filecontents}
\usepackage{mathtools}
\usepackage{mdwlist}
\usepackage{mathrsfs}
\usepackage[shortlabels]{enumitem}
\usepackage{amssymb}
\usepackage{textcomp}
\usepackage[toc,page]{appendix}

\usepackage{pgfplots}
\pgfplotsset{compat=1.7}

\usepackage[utf8]{inputenc} % To interpret Czech accents and other characters
\usepackage{authblk} % To write the affiliations

\usepackage[cal=boondox,calscaled=1.0]{mathalfa}

\usepackage{physics}

\numberwithin{equation}{section}
\numberwithin{figure}{section}
\theoremstyle{definition}

\newtheorem{observation}{Observation}[section]

\newtheorem{example}{Example}[section]

\usepackage{natbib}
\usepackage[nottoc]{tocbibind}

\title{Models of Discrete Linear Evolution for Quantum Systems}
\author[1]{Jakub Káninský}

\affil[1]{Charles University, Faculty of Mathematics and Physics, Institute of Theoretical Physics. E-mail address: jakubkaninsky@seznam.cz}

\date{\today}
\titleformat{\section}{\normalfont\scshape\large}{\thesection}{2em}{}
\titleformat{\subsection}{\normalfont\scshape\normalsize}{\thesubsection}{1em}{}
\titleformat{\subsubsection}{\itshape\normalsize}{\thesubsubsection}{1em}{}

\renewenvironment{abstract}
 {\small
  \begin{center}
  \end{center}
  \list{}{
    \setlength{\leftmargin}{15mm}%
    \setlength{\rightmargin}{\leftmargin}%
  }
  \item\relax}
 {\endlist}

\usepackage{tocloft}
\begin{document}

\maketitle

\begin{abstract}
Discrete canonical evolution is a key tool for understanding the dynamics in discrete models of spacetime, in particular those represented by a triangular Regge lattice. We consider a finite-dimensional system whose evolution is realized by a series of discrete-time evolution steps governed by Hamiltonian equations of motion that are linear in the canonical coordinates. The evolution is allowed to be irregular, which produces constraints as well as non-uniqueness of solutions. We provide two independent and fundamentally different approaches to canonical quantization of this system, including detailed description of the evolution maps, measurement and path integrals. It is argued that some irregular discrete systems may be most naturally described by a non-unitary quantum evolution. The formalism is then applied to a simple yet physically relevant model of massless scalar field on a two-dimensional spacetime lattice.
\end{abstract}

\vspace{\baselineskip}
\vspace{\baselineskip}
\vspace{\baselineskip}

\tableofcontents

\vspace{\baselineskip}
\vspace{\baselineskip}
\vspace{\baselineskip}

\section{Introduction}
Time and time again, physicists conjecture that the world we live in could be on microscopic scale discrete \cite{Henson2010, Hossenfelder2013, Crouse2016}. So far there is no compelling evidence in favor of that conjecture, neither is it expected to emerge soon---although there are some ideas around on how it could be acquired \cite{Brun2018}. In the meantime, discrete spacetime models of various kinds are subjected to intensive research in quantum gravity \cite{Oriti2009}, whether they are presented as instrumental \cite{Loll2019}, fundamental \cite{Surya2011} or even as candidates to a unified theory of physics \cite{Finster2018}. Many are predestined to fade away, but some have already proven useful in forming our views on critical aspects of the problem and raised valuable questions. The present work is motivated in particular by lattice spacetime models based on Regge calculus \cite{Regge1961}, which have a rich history of serving to numerous discrete approaches to quantum gravity and have given rise to some well established theories like Causal Dynamical Triangulations \cite{Ambjoern2001, Ambjoern2013, Loll2019, Mielczarek2017} and a few others \cite{Barrett2018}. In this paper, we limit our consideration to discrete evolution of quantum systems with configuration space isomorphic to $ \mathbb{R}^{q} $ and linear equations of motion. This is not the case of gravity but it typically applies e.g. to free quantum fields on a fixed spacetime lattice (or any other kinds of fields, given that their action is quadratic). Indeed, the introduced framework may be directly applied to this case and used to model dynamics of fields on fixed lattices (more on this below). Besides that, the linear model can serve as a basis for future generalizations or as a point of reference for other more general models of discrete quantum evolution like \cite{Hoehn2014a}. These are ultimately expected to describe the evolution of geometry in discrete approaches to quantum gravity \cite{Hamber2009, Xue2012, Dittrich2011, Hoehn2014a, Hoehn2014b, Mikovic2018}.

This work is based on the model of \textit{discrete canonical evolution} introduced by B. Dittrich and P. A. H\"{o}hn in a series of articles \cite{Dittrich2011, Dittrich2012, Dittrich2013}, also motivated by the application to simplicial gravity. Besides working in the discrete-time setting, one of the most important features of their treatment is that it allows for irregular evolution which does not always provide one-to-one correspondence of initial and final states. Due to this, the model demands special care in order to be correctly quantized. One approach towards quantization of discretely evolving systems has been presented in \cite{Hoehn2014a}, where the author builds upon the rich constraint classification from \cite{Dittrich2013} and follows the Dirac algorithm \cite{Dirac2001}. Within this approach, one begins with a kinematic quantization and imposes constraints on the quantum states and completes the resulting solution space to a Hilbert space. In practice, this amounts to introducing the so called physical Hilbert spaces and group averaging projectors in order to define the quantum evolution. It is shown how the resulting one-step evolution moves can be composed into an effective move across multiple time-steps with a matching procedure involving an integration over the variables at intermediate time-slices. The article shortly mentions the corresponding path integral. The formalism is also put into the context of simplicial quantum gravity, where the situation is still much more involved. Importantly for us, the quantization procedure has been subsequently applied specifically to systems with quadratic actions in section 10 of \cite{Hoehn2014}. Therein, prescriptions for (different versions of) the propagator and its measure are provided---we will come back to it when discussing our results.

In the present paper we offer another perspective on the problem. We take advantage of our preceding work \cite{Kaninsky2020a} concerning classical discrete linear systems, where we have supplied a simple yet efficient analysis of the symplectic structure with respect to the constraint surfaces. Based on that, we provided two special choices of the symplectic basis on phase spaces of two consecutive time-slices which trivialize the constraints as well as the one-step evolution. Global solutions were considered too. All these classical concepts shall find use in our present work. They are reviewed in section 2. In section 3 we supply a short description of the canonical quantization procedure which will be employed couple of times throughout the paper and add a couple of useful formulas concerning the symplectic transformation of observables. Eventually we will move on to present two standalone models of the treated system which differ both from each other and from the model described in \cite{Hoehn2014a, Hoehn2014}. The first model under our consideration, given in section 4, comes out rather straightforwardly by quantizing a suitably reduced space of solutions. Within this approach, the irregularity of the system is filtered out before quantization. The resulting model is well defined, with conserved symplectic structure and an inherent unitary quantum evolution induced by the classical solutions. As we shall see, this version of the quantum evolution resembles in effect the model introduced in \cite{Hoehn2014a, Hoehn2014} because it is unitary and depends on the chosen time interval, which results in breaking of causality. We argue that in some cases such behavior is not desirable. Besides that, the reduced model suffers from another more practical drawback, which is a non-standard form of the path integral. We are therefore encouraged to propose another model designated to surpass the stated problematic features and better suited for practical computations. We do so in the final section 5. Within the second model, the quantum evolution is found to be non-unitary. On the other hand, we show that this model produces an (almost) standard form of the one-step propagator, which, up to a necessary regularization term and small differences in the propagator measure, agrees with the propagator given in \cite{Hoehn2014}. Accordingly, it leads to the standard path integral in the original canonical variables, which is very favorable for application. The price for all this is the need for regularization of the evolved states. We provide two regularization schemes which are respectful of the constraints and result in the same normalization of states under one-step evolution. We further show that the wavefunctions need to be renormalized in order to recover the usual probabilistic interpretation of the theory, and describe the renormalization process.

In the final section 6 of the article, we apply our second non-unitary model to the example of quantum scalar field on a simple two-dimensional spacetime lattice. The benefit is two-fold. First, the application allows us to review and demonstrate in a clear way the behavior of the non-unitary evolution and explain some of its subtleties. Second, the model provides a relevant description of quantum field on spacetime lattice which can be seen as a discrete analogy of quantum field theory on (continuous) curved spacetime \cite{Wald1994} or as a generalization of the well established lattice quantum field theory \cite{Smit2003}. Similar models have been studied on numerous occasions for both fixed piecewise linear manifolds \cite{Brower2018, Foster2004} and those coupled to the field within some quantum gravity model \cite{Hamber1993, Hamber1993a, Paunkovic2016, Brower2016, Jha2018}. As we have already mentioned, the latter option is out of reach of the present model---gravity possesses nonlinear equations of motion and more complicated configuration spaces. However, the former option is implemented straightforwardly. For our purpose, we choose a simple two-dimensional setting with fixed edge lengths identical to the one introduced before in \cite{Kaninsky2020a} in order to demonstrate the formalism developed for classical evolution. The model can be easily adapted to account for different lattices or other terms in the action. We briefly review its general properties and consider several examples of illustrative one-step lattices, for which we express the evolution maps, propagators, (regularized) norms and inner products used for the computation of probabilities.

\section{Discrete Linear Evolution}

Within this section we shall briefly review the classical model which serves as our starting point. Further details and context can be found in our previous work \cite{Kaninsky2020a}.

Our system of interest evolves in discrete time-steps between individual time-slices indexed by $ n = 0, ..., t $. We assume that there is a natural number $ q $ such that the configuration space $ \mathcal{Q}_{n} $ of our system at any time-slice $ n $ from the chosen interval $ \{ 0, ..., t \} $ is isomorphic to $ \mathbb{R}^{q} $. This assumption implies that the system at hand is finitely-dimensional and its dimension $ q $ is not varying in time. As we will soon make clear, this invariance is only formal. The phase space $ \mathcal{P}_{n} $ is found to be isomorphic to $ \mathbb{R}^{2q} $. Thanks to its natural vector-space structure, one can identify the tangent space $ T_{y_{n}}\mathcal{P}_{n} $ at every point $ y_{n} \in \mathcal{P}_{n} $ with the phase space itself, which allows for the definition of the symplectic structure $ \omega_{n}: \mathcal{P}_{n} \times \mathcal{P}_{n} \rightarrow \mathbb{R} $. We denote the canonical basis of $ \mathcal{P}_{n} $ by $ \{ e_{n I} \}_{I = 1}^{2q} $ and the canonical coordinates of $ y_{n} \in \mathcal{P}_{n} $ by $ y_{nI} $ so that $ y_{n} = y_{n I} e_{n I} $ with $ I = 1, ..., 2q $. The first $ q $ components $ y_{nA} $ with $ A = 1, ..., q $ are referred to as coordinates, the second $ q $ components $ y_{nA+q} $ are referred to as momenta. We can then define the symplectic structure by the standard prescription
\begin{equation}\label{symplstruct}
\omega_{n}(y_{n},u_{n}) = y_{nA} u_{nA+q} - y_{nA+q} u_{nA}
\end{equation}
for all $ y_{n}, u_{n} \in \mathcal{P}_{n} $. Often we shall pass to the matrix form, in which we denote the coordinates $ y_{nI} $ of the point $ y_{n} \in \mathcal{P}_{n} $ by coordinate vectors $ y_{n} \in \mathbb{R}^{2q} $ with the usual abuse of notation. Analogically, the vector $ x_{n} \in \mathcal{Q}_{n} $ is expressed in basis as $ x_{n} = x_{n A} e_{n A} $ with $ A = 1, ..., q $ and we will write the coordinates $ x_{n A} $ compactly as a coordinate vector $ x_{n} \in \mathbb{R}^{q} $. This allows us to work efficiently in coordinates, using either $ x_{n} $ or $ y_{n} = \begin{pmatrix} x_{n} & p_{n} \end{pmatrix}^{T} $. The symplectic product \eqref{symplstruct} is equivalently expressed in matrix form as $ \omega_{n}(y_{n},u_{n}) = y_{n}^{T} \sigma u_{n} $ with the $ 2q \times 2q $ matrix
\begin{equation}\label{key}
\sigma = \begin{pmatrix}
\phantom{-} 0 & \mathbf{1} \\
-\mathbf{1} & 0
\end{pmatrix}
\end{equation}

We assume to be given a global action functional $ S_{0t} $ equal to the sum of contributions $ S_{n+1}(x_{n},x_{n+1}) $ of the individual time-steps,
\begin{equation}\label{action}
S_{0t} = \sum_{n=0}^{t-1} S_{n+1}(x_{n},x_{n+1})
\end{equation}
The action contribution $ S_{n+1} $ governs the discrete time evolution move in the time-step between $ n $ and $ n+1 $. Note that in our discrete-time setting, it plays the role of Lagrangian. Since we are in the business of linear discrete evolution, we assume that $ S_{n+1} $ is quadratic in the generalized coordinates $ x_{n} $. We can use it to pass to the canonical formulation by defining so-called \textit{pre-momenta} $ ^{-}p_{n} $ and \textit{post-momenta} $ ^{+}p_{n} $ via the equations
\begin{equation}\label{prepostmomentamatrix}
\begin{aligned}
^{-}p_{n} &= - \frac{\partial S_{n+1}}{\partial x_{n}} = L_{n} x_{n} + R_{n} x_{n+1} \\
^{+}p_{n+1} &= \frac{\partial S_{n+1}}{\partial x_{n+1}} =  \bar{L}_{n+1} x_{n} + \bar{R}_{n+1} x_{n+1}
\end{aligned}
\end{equation}
where $ ^{-}p_{n} $, $ ^{+}p_{n+1} $, $ x_{n} $, $ x_{n+1} $ are coordinate vectors of dimension $ q $ and $ L_{n}, R_{n}, \bar{L}_{n+1}, \bar{R}_{n+1} $ are real $ q \times q $ matrices. In this canonical picture, the equations of motion demonstrate themselves in the form of so-called \textit{momentum-matching}
\begin{equation}\label{mommatch}
^{-}p_{n} = ~^{+}p_{n+1}
\end{equation}
Therefore, if we consider a solution to the equations of motion, we may drop the $ - $ and $ + $ indices of the momenta, enforcing momentum matching. By virtue of the exchengability of partial derivatives, it holds
\begin{equation}\label{pardevsymm}
\bar{L}_{n+1} = - R_{n}^{T}
\end{equation}
\begin{equation}\label{pardevsymm2}
L_{n} = L_{n}^{T}, \qquad \bar{R}_{n+1} = \bar{R}_{n+1}^{T}
\end{equation}
We need to stress out that there are no other assumptions about $ L_{n}, R_{n} $ and $ \bar{R}_{n+1} $, in particular, they may be \textit{singular}. With the above definitions, the prescription for the time-step action contribution is rewritten as
\begin{equation}\label{Slag}
S_{n+1} = \textstyle - \frac{1}{2} \left(  x_{n}^{T} L_{n} x_{n} + 2 x_{n}^{T} R_{n} x_{n+1} - x_{n+1}^{T} \bar{R}_{n+1} x_{n+1} \right)
\end{equation}

\subsection{The One-Step Evolution}\label{constrsurf}
From the momentum matching \eqref{mommatch} one can deduce the general form of the one-step forward Hamiltonian evolution map between time-slices $ n $ and $ n+1 $. To express it, we use the singular value decomposition of the matrix $ R_{n} $, written as $ R_{n} = U \Sigma V^{T} $. Here, $ U $ and $ V $ are orthogonal $ q \times q $ matrices and $ \Sigma $ is the diagonal $ q \times q $ matrix of singular values. We denote $ r_{n} \equiv \rank (R_{n}) $, $ s_{n} \equiv q - r_{n} $ and divide the above matrices into blocks as
\begin{equation}\label{key}
U = \begin{pmatrix}
U_{1} & U_{2}
\end{pmatrix}, \qquad \Sigma = \begin{pmatrix}
\Sigma_{r} & 0 \\
0 & 0
\end{pmatrix}, \qquad V = \begin{pmatrix}
V_{1} & V_{2}
\end{pmatrix}
\end{equation}
with $ U_{1}, V_{1} \in \mathbb{R}^{q \times r_{n}} $ while $ U_{2}, V_{2} \in \mathbb{R}^{q \times s_{n}} $ and $ \Sigma_{r} \in \mathbb{R}^{r_{n} \times r_{n}} $. The above singular value decomposition has a narrowed version $ R_{n} = U_{1} \Sigma_{r} V_{1} $ in which one ignores the singular part of $ \Sigma $.

Due to the allowed irregularity of the system, the evolution in general does not work for all $ y_{n} \in \mathcal{P}_{n} $. Rather, one can only find solution of \eqref{prepostmomentamatrix} and \eqref{mommatch} for points $ y_{n} \in \mathcal{P}_{n} $ which satisfy the so-called \textit{pre-constraint}
\begin{equation}\label{preconst}
C_{n} y_n = 0, \qquad C_{n} = \begin{pmatrix}
- U_{2} U_{2}^{T} L_{n} & & U_{2} U_{2}^{T}
\end{pmatrix}
\end{equation}
This relation defines the \textit{pre-constraint surface} $ \mathcal{C}_{n}^{-} $. In our case, $ \mathcal{C}_{n}^{-} $ is simply a linear subspace of $ \mathcal{P}_{n} $. The one-step Hamiltonian evolution map $ \mathbb{H}_{n+1}(\lambda_{n+1}) : \mathcal{C}_{n}^{-} \rightarrow \mathcal{C}_{n+1}^{+} $ denoted also as $ y_{n+1} \equiv \mathbb{H}_{n+1}(\lambda_{n+1}) y_{n} $ is then found to act as
\begin{equation}\label{evol}
y_{n+1} = E_{n} y_{n} + F_{n+1} \lambda_{n+1}
\end{equation}
for all $ y_{n} \in \mathcal{C}_{n}^{-} $. Here
\begin{equation}\label{En}
E_{n} = \begin{pmatrix}
- R_{n}^{+} L_{n} & R_{n}^{+} \\
- R_{n}^{T} - \bar{R}_{n+1} R_{n}^{+} L_{n} & \bar{R}_{n+1} R_{n}^{+}
\end{pmatrix}
\end{equation}
and
\begin{equation}\label{Fn+1}
F_{n+1} = \begin{pmatrix}
V_{2} \\
\bar{R}_{n+1} V_{2}
\end{pmatrix}
\end{equation}
The matrix $ E_{n} $ has dimensions $ 2q \times 2q $ and the matrix $ F_{n+1} $ has dimensions $ 2q \times s_{n} $. Note that to specify the map uniquely, one needs to supply the coordinate vector $ \lambda_{n+1} \in \mathbb{R}^{s_{n}} $ of $ s_{n} $ free parameters. The vectors $ y_{n+1} $ are elements of the so-called \textit{post-constraint surface} $ \mathcal{C}_{n+1}^{+} \subset \mathcal{P}_{n+1} $. We further define the space $ \Lambda_{n+1} = \{ F_{n+1} \lambda_{n+1} ~ \vert ~ \lambda_{n+1} \in \mathbb{R}^{s_{n}} \} $ which is a subspace of the post-constraint surface and plays an interesting role in the evolution because the position in this space is determined by the free parameters $ \lambda_{n+1} $ rather than the initial data $ y_{n} $.

One can check that for solutions of the equations of motion, the symplectic structure \eqref{symplstruct} is conserved, i.e.,
\begin{equation}\label{symplcons}
\omega_{n+1}(y_{n+1}, u_{n+1}) = \omega_{n}(y_{n}, u_{n})
\end{equation}
Of course, this conservation property only relates to $ y_{n}, u_{n} \in \mathcal{C}_{n}^{-} $, because outside the pre-constraint surface there are no solutions $ y_{n+1}, u_{n+1} $ at all. One needs to acknowledge that when seen as a two-form, $ \omega_{n} $ is in general no longer symplectic on $ T_{y_{n}} \mathcal{C}_{n}^{-} $, since it may be degenerate there. The structure of the constraint surfaces $ \mathcal{C}_{n}^{-}, \mathcal{C}_{n+1}^{+} $ with respect to $ \omega_{n} $ is studied in detail in \cite{Kaninsky2020a}, where we also provided two special coordinate frames that considerably simplify the picture. They are given by the symplectic transforms
\begin{equation}\label{doteWe}
e_{nJ} = \dot{e}_{nI} \dot{W}_{nIJ}, \qquad e_{n+1J} = \ddot{e}_{n+1I} \ddot{W}_{n+1IJ}
\end{equation}
with
\begin{equation}\label{dotW}
\dot{W}_{n} = \begin{pmatrix}
-U^{T} L_{n} & U^{T} \\
-U^{T} & 0
\end{pmatrix}, \qquad \ddot{W}_{n+1} = \begin{pmatrix}
\bar{\Sigma} V^{T} & 0 \\
-\bar{\Sigma}^{-1} V^{T} \bar{R}_{n+1} & \bar{\Sigma}^{-1} V^{T}
\end{pmatrix}
\end{equation}
where we take advantage of the notation
\begin{equation}\label{key}
\bar{\Sigma} \equiv \begin{pmatrix}
\Sigma_{r} & 0 \\
0 & \mathbf{1}
\end{pmatrix}
\end{equation}
The new symplectic bases $ \{ \dot{e}_{nI} \}_{I = 1}^{2q} $ and $ \{ \ddot{e}_{n+1I} \}_{I = 1}^{2q} $ of $ \mathcal{P}_{n} $ and $ \mathcal{P}_{n+1} $, respectively, specify \textit{adapted coordinates}
\begin{equation}\label{key}
\dot{y}_{n} = \dot{W}_{n} y_{n}, \qquad \ddot{y}_{n+1} = \ddot{W}_{n+1} y_{n+1}
\end{equation}
on the corresponding phase spaces. In these coordinates the Hamiltonian evolution map $ \mathbb{H}_{n+1}(\lambda_{n+1}) $ takes a particularly simple form: the components $ \dot{y}_{nX} $ with $ X = 1,..., r_{n}, q+1, ..., q+r_{n} $ which correspond to propagating degrees of freedom are conserved,
\begin{equation}\label{uddotX}
\ddot{y}_{n+1X} = \dot{y}_{nX}
\end{equation}
while the zero components $ \dot{y}_{nM} $ with $ M = r_{n}+1, ..., q $ corresponding to the part of $ y_{n} $ outside the pre-constraint surface (we assume $ y_{n} \in \mathcal{C}_{n}^{-} $, which is equivalent to  $ \dot{y}_{nM} = 0 $) are updated with an arbitrary constant contribution generated by the free parameters,
\begin{equation}\label{uddotM}
\ddot{y}_{n+1M} = \lambda_{n+1 M-r_{n}}
\end{equation}
and the remaining components $ \dot{y}_{nM+q} $ are annihilated,
\begin{equation}\label{uddotMq}
\ddot{y}_{n+1M+q} = \dot{y}_{nM} = 0
\end{equation}

\subsection{Global Solutions}\label{globsol}
Once we know how the system evolves in a single time-step, we can chain these steps one after another to obtain global solutions. A \textit{solution} $ y \in \mathcal{S}_{0t} $ is a point in $ \mathcal{P}_{0t} = \bigotimes_{n = 0}^{t} \mathcal{P}_{n} $ which satisfies the equations of motion at each time-slice $ n $, including all the constraints. It may of course happen that no solution on the whole interval $ \{0, ..., t\} $ exists; we are however more interested in the case when the solutions are plentiful. We define the constraint surfaces
\begin{equation}\label{Dn}
\mathcal{D}_{n} = \{ y_{n} \in \mathcal{P}_{n} ~ \vert ~ \exists \text{ solution } y \in \mathcal{S}_{0t} \text{ such that } y_{n} = \mathbb{P}_{n} y \}
\end{equation}
where $ \mathbb{P}_{n} $ is the natural projection from $ \mathcal{P}_{0t} $ to $ \mathcal{P}_{n} $. We may also consider the total constraint space $ \mathcal{D}_{0t} = \bigotimes_{n = 0}^{t} \mathcal{D}_{n} $. By definition, $ \mathcal{S}_{0t} \subset \mathcal{D}_{0t} $. However, we must keep in mind that because of the free parameters of the Hamiltonian evolution map, $ y \in \mathcal{S}_{0t} $ is in general not uniquely defined by its time-slice data $ y_{n} \in \mathcal{D}_{n} $.

Thanks to \eqref{symplcons}, the symplectic structure naturally induces the \textit{product of solutions} $ \omega: \mathcal{S}_{0t} \times \mathcal{S}_{0t} \rightarrow \mathbb{R} $ given by $ \omega(y,z) = \omega_{n}( y_{n} , z_{n}) $ for an arbitrary $ n \in \{0, ..., t\} $. This product is \textit{not} generally symplectic. We say that two solutions $ y, \tilde{y} \in \mathcal{S}_{0t} $ are \textit{symplectically equivalent} if $ \omega(y,z) = \omega(\tilde{y},z) $ for all $ z \in \mathcal{S}_{0t} $, and write $ y \sim \tilde{y} $. Then we can render the equivalence classes $ [y] $ of all symplectically equivalent solutions $ [y] = \{ \tilde{y} ~ \vert ~ \tilde{y} \sim y \} $. The space of such equivalence classes shall be denoted by $ \tilde{\mathcal{S}}_{0t} $. There is a naturally induced product $ \omega : \tilde{\mathcal{S}}_{0t} \times \tilde{\mathcal{S}}_{0t} \rightarrow \mathbb{R} $, $ \omega( [y], [z] ) = \omega(y,z) $. It turns out that by definition, $ (\tilde{\mathcal{S}}_{0t}, \omega) $ is a symplectic vector space. In the same manner we get the symplectic space $ ( \tilde{\mathcal{D}}_{n}, \omega_{n} ) $ of equivalence classes $ [y_{n}] = \{ \tilde{y}_{n} \in \mathcal{D}_{n} ~ \vert ~ \tilde{y}_{n} \sim y_{n} \} $ with $ y_{n} \sim \tilde{y}_{n} $ defined by $ \omega_{n}(y_{n},z_{n}) = \omega_{n}(\tilde{y}_{n},z_{n}) $ for all $ z_{n} \in \mathcal{D}_{n} $.

For practical purposes, we can go one more step and represent each class of symplectically equivalent solutions $ [y] \in \tilde{\mathcal{S}}_{0t} $ by a single solution $ y \in [y] $. The space of these representative solutions shall be denoted by $ \dot{\mathcal{S}}_{0t} $.  We require that $ \dot{\mathcal{S}}_{0t} $ is a linear subspace of $ \mathcal{S}_{0t} $. Once it is chosen, it fixes uniquely the spaces $ \dot{\mathcal{D}}_{n} = \{ \mathbb{P}_{n}y ~ \vert ~ y \in \dot{\mathcal{S}}_{0t} \} $ of the corresponding initial data. We let $ \dot{\mathcal{S}}_{0t} $,  $ \dot{\mathcal{D}}_{n} $ inherit the symplectic structures $ \omega $, $ \omega_{n} $ of $ \tilde{\mathcal{S}}_{0t} $, $ \tilde{\mathcal{D}}_{n} $, respectively, then they turn out to be symplectic. Eventually we can make the following observation:

\begin{observation}\label{O3}
	The spaces $ (\dot{\mathcal{D}}_{n}, \omega_{n}) $ for each $ n \in \{ 0, ..., t \} $ and $ (\dot{\mathcal{S}}_{0t}, \omega) $ are all mutually symplectomorphic.
\end{observation}
\begin{proof}
See the discussion in section 5.9 of \cite{Kaninsky2020a}.
\end{proof}

Finally let us remark that the constraint surface $ \mathcal{D}_{n} $ naturally splits into the representative space $ \dot{\mathcal{D}}_{n} $ and the null space $ \mathcal{N}_{n} = \{ z_{n} \in \mathcal{D}_{n} ~ \vert ~ z_{n} \sim 0 \} $, in particular:
\begin{observation}\label{DdotDN}
	It holds $ \mathcal{D}_{n} = \dot{\mathcal{D}}_{n} \oplus \mathcal{N}_{n} $.
\end{observation}
\begin{proof}
	See Observation 5.4 of \cite{Kaninsky2020a}.
\end{proof}

\section{Quantization: a Quick Overview}\label{quant}
In this brief intermezzo we take the opportunity to lay out the basics of our canonical quantization procedure, recall the necessary formalism and establish our notation. We include some results which will be needed in the upcoming sections. For more details, we refer the reader to our previous work \cite{Kaninsky2020} and references therein.

\subsection{The Construction and Dirac Formalism}
Throughout this paper, it is assumed that the classical system in question has configuration space $ \mathcal{Q} $ of finite dimension $ q $ and natural vector space structure, which allows one to identify $ \mathcal{Q} = \mathbb{R}^{q} $. Then the phase space is $ \mathcal{P} = \mathbb{R}^{2q} $ with a symplectic form $ \omega : \mathcal{P} \times \mathcal{P} \rightarrow \mathbb{R} $. We will use a textbook construction of the quantum system, in which one defines the Hilbert space $ \mathcal{F} = L^{2}(\mathcal{Q}) $ to be the space of square-integrable complex functions on $ \mathcal{Q} $ with the usual inner product $ (\psi, \varphi) = \int_{\mathbb{R}^{q}} \overline{\psi(x)} \varphi(x) ~ d^{q} x $.

In order to fix a particular representation of the canonical commutation relations and define the coordinate and momentum operators, one needs to fix a basis in $ \mathcal{P} $. For start, we shall use the canonical symplectic basis $ \{ e_{I} \}_{I = 1}^{2q} $. A vector $ y \in \mathcal{P} $ can be written in this basis as $ y = y_{I} e_{I} $ with an implicit summation over $ I = 1, ..., 2q $. The configuration space $ \mathcal{Q} $ is identified as the space spanned by $ \{ e_{A} \}_{A = 1}^{q} $. The symplectic product of two vectors $ y, u \in \mathcal{P} $ has the form
\begin{equation}\label{key}
\omega(y, u) = y_{I} ~ \sigma_{IJ} ~ u_{J} = y_{A} u_{A+q} - y_{A+q} u_{A}
\end{equation}
Upon the Hilbert space $ \mathcal{F} $ one defines the coordinate and momentum operators $ \hat{y}_{A} : \mathcal{D}(\hat{y}_{A}) \rightarrow \mathcal{F} $ and $ \hat{y}_{A+q} : \mathcal{D}(\hat{y}_{A+q}) \rightarrow \mathcal{F} $, respectively, via
\begin{equation}\label{ymult}
\hat{y}_{A} \varphi(x) = x_{A} \varphi(x)
\end{equation}
\begin{equation}\label{yder}
\hat{y}_{A+q} \varphi(x) = - i \frac{\partial}{\partial x_{A}} \varphi(x)
\end{equation}
Note that their domains $ \mathcal{D}(\hat{y}_{A}) $, $ \mathcal{D}(\hat{y}_{A+q}) $ cannot cover the whole $ \mathcal{F} $. Instead, the operators can be defined at most densely in $ \mathcal{F} $. The domains are also not invariant under the action of the operators they belong to (nor other operators from the family), which makes it complicated to consider powers and combinations of the coordinate and momentum operators. One finds that both $ \hat{y}_{A} $ and $ \hat{y}_{A+q} $ are unbounded and \textit{have no eigenvectors} in $ \mathcal{F} $.

To work with the coordinate and momentum operators $ \hat{y}_{A} $, $ \hat{y}_{A+q} $ easily and safely, we take advantage of the formalism of Rigged Hilbert spaces \cite{Madrid2005}. A rigged Hilbert space is a triad of spaces $ \Phi \subset \mathcal{F} \subset \Phi^{\times} $ where $ \mathcal{F} $ is our Hilbert space, $ \Phi $ is a dense subset of $ \mathcal{F} $ containing well behaved test functions (roughly said, one requires $ \Phi $ to be invariant under the action of any combination of the coordinate and momentum operators), and $ \Phi^{\times} $ is the space of (suitably bounded) \textit{antilinear} functionals over $ \Phi $. This space accommodates the \textit{right states} $ \vert \varphi \rangle \in \Phi^{\times} $, a.k.a. ``kets''  of the Dirac formalism.

In our case, $ \Phi $ is found to be the Schwartz space of smooth rapidly decreasing functions on $ \mathbb{R}^{q} $. Thanks to its invariance under the action of any combination of $ \hat{y}_{A} $ and $ \hat{y}_{A+q} $, the expectation values $ (\varphi, \hat{y}_{I}^{m_{I}} \varphi) $ for $ \varphi \in \Phi $ are finite for any $ m_{I} \in \mathbb{N} $. Likewise, for coordinate and momentum operators narrowed to $ \Phi $, the canonical commutation relation can be written in the illustrious form
\begin{equation}\label{ccr}
[ \hat{y}_{I} , \hat{y}_{J} ] = i \sigma_{IJ} \hat{\mathbf{1}}
\end{equation}
as an identity on $ \Phi $.

In analogy to the above, one introduces a second rigged Hilbert space $ \Phi \subset \mathcal{F} \subset \Phi^{\overline{\times}} $, where $ \Phi^{\overline{\times}} $ is the space of (suitably bounded) \textit{linear} functionals over $ \Phi $, accommodating \textit{left states} $ \langle \varphi \vert \in \Phi^{\overline{\times}} $, a.k.a. ``bras''. The two described rigged Hilbert spaces are in one-to-one correspondence. With these in place, one can define left and right eigenstates of self-adjoint operators. It happens that although the coordinate and momentum operators (which are naturally self-adjoint) do not have eigenvectors in $ \mathcal{F} $, they do have eigenstates in $ \Phi^{\times} $ and $ \Phi^{\overline{\times}} $.

Thanks to the construction very roughly outlined above, one can take advantage of the powerful Dirac formalism. In particular, one gives good meaning to the equations
\begin{equation}\label{key}
\hat{y}_{A} \vert \alpha \rangle_{\mathtt{c}} = \alpha_{A} \vert \alpha \rangle_{\mathtt{c}}
\end{equation}
\begin{equation}\label{key}
\hat{y}_{A+q} \vert \beta \rangle_{\mathtt{m}} = \beta_{A} \vert \beta \rangle_{\mathtt{m}}
\end{equation}
which define the eigenvalues $ \alpha_{A}, \beta_{A} \in \mathbb{R} $ as well as the eigenstates $ \vert \alpha \rangle_{\mathtt{c}}, \vert \beta \rangle_{\mathtt{m}} \in \Phi^{\times} $ of the observables $ \hat{y}_{A}, \hat{y}_{A+q} $. Then one can rely on the basic rules known from quantum mechanics, most importantly the eigenstate wavefunctions and resolutions of identity. Upon using the symbolic summation
\begin{equation}\label{key}
\sum_{\alpha} \equiv \int_{\mathbb{R}} \prod_{A = 1}^{q} d \alpha_{A}
\end{equation}
these amount to
\begin{equation}\label{formfieldmom}
\begin{aligned}
_{\mathtt{c}} \langle \alpha \vert \gamma \rangle _{\mathtt{c}} &= \delta^{q}(\alpha - \gamma) \\
_{\mathtt{c}} \langle \alpha \vert \beta \rangle_{\mathtt{m}} &= (2\pi)^{-q/2} ~ e^{ i \beta_{A} \alpha_{A} } \\
\hat{\mathbf{1}} &= \sum_{\alpha} \vert \alpha \rangle_{\mathtt{c}} ~ _{\mathtt{c}} \langle \alpha \vert \\
\hat{\mathbf{1}} &= \sum_{\beta} \vert \beta \rangle _{\mathtt{m}} ~ _{\mathtt{m}} \langle \beta \vert
\end{aligned}
\end{equation}

\subsection{The Symplectic Transformation of Observables}

Let us consider the following problem. In addition to $ \hat{y}_{A} $ and $ \hat{y}_{A+q} $, one can introduce new observables $ \hat{\dot{y}}_{A} $, $ \hat{\dot{y}}_{A+q} $ with respect to a new symplectic basis $ \{ \dot{e}_{I} \}_{I = 1}^{2q} $ of $ \mathcal{P} $, and define the corresponding new eigenstates $ \vert \alpha \rangle_{\dot{\mathtt{c}}}, \vert \beta \rangle_{\dot{\mathtt{m}}} \in \Phi^{\times} $ by
\begin{equation}\label{key}
\begin{aligned}
\hat{\dot{y}}_{A} \vert \alpha \rangle_{\dot{\mathtt{c}}} &= \alpha_{A} \vert \alpha \rangle_{\dot{\mathtt{c}}}\\
\hat{\dot{y}}_{A+q} \vert \beta \rangle_{\dot{\mathtt{m}}} &= \beta_{A} \vert \beta \rangle_{\dot{\mathtt{m}}}
\end{aligned}
\end{equation}
In section 4 of \cite{Kaninsky2020}, we provided explicit formulas for expressing such eigenstates in the canonical coordinate eigenstate basis. For a transformation
\begin{equation}\label{key}
e_{J} = \dot{e}_{I} W_{IJ}
\end{equation}
given by a $ 2q \times 2q $ symplectic matrix
\begin{equation}\label{key}
W = \begin{pmatrix}
E & F\\
G & H
\end{pmatrix}
\end{equation}
it was found that the new eigenstates are given by
\begin{equation}\label{cdot}
\vert \alpha \rangle_{\dot{\mathtt{c}}} = \sum_{\xi} \vert V_{1} \xi + E^{+} U_{2} U_{2}^{T} \alpha \rangle_{\mathtt{c}} ~ e^{i C} e^{i \left(- \frac{1}{2} \xi^{T} V_{1}^{T} F^{+} E F^{T} F^{+T} V_{1} \xi + \xi^{T} V_{1}^{T} F^{+} \alpha \right) }
\end{equation}
\begin{equation}\label{mdot}
\vert \beta \rangle_{\dot{\mathtt{m}}} = \sum_{\chi} ~ \vert \acute{V}_{1} \chi + G^{+} \acute{U}_{2} \acute{U}_{2}^{T} \beta \rangle_{\mathtt{c}} ~ e^{i K} e^{i \left(- \frac{1}{2} \chi^{T} \acute{V}_{1}^{T} H^{+} G H^{T} H^{+T} \acute{V}_{1} \chi + \chi^{T} \acute{V}_{1}^{T} H^{+} \beta \right) }
\end{equation}
where $ \xi \in \mathbb{R}^{r} $ with $ r = \text{rank}(F) $ while $ \chi \in \mathbb{R}^{\acute{r}} $ with $ \acute{r} = \text{rank}(H) $. Also, $ V_{1} \equiv V_{1}(F) $ and $ U_{2} \equiv U_{2}(F) $ are matrices originating from the singular value decomposition of $ F $ while $ \acute{V}_{1} \equiv V_{1}(H) $ and $ \acute{U}_{2} \equiv U_{2}(H) $ are matrices originating from the singular value decomposition of $ H $. These results will be used later in a study of a suggested evolution scheme.

Finally, one can search section 4 of \cite{Kaninsky2020} to find the following formulas:
\begin{equation}\label{dotcdotc}
_{\dot{\mathtt{c}}}\langle \varrho \vert \omega \rangle_{\dot{\mathtt{c}}} = \delta^{s}( V_{2}^{T} E^{+} U_{2} U_{2}^{T} (\varrho - \omega) ) e^{-2 \Im C} (2\pi)^{r} \delta^{r}( V_{1}^{T} F^{+} ( \varrho - \omega) )
\end{equation}
\begin{equation}\label{dotmdotm}
_{\dot{\mathtt{m}}}\langle \varrho \vert \omega \rangle_{\dot{\mathtt{m}}} = \delta^{\acute{s}}( \acute{V}_{2}^{T} G^{+} \acute{U}_{2} \acute{U}_{2}^{T} (\varrho - \omega) ) e^{-2 \Im K} (2\pi)^{\acute{r}} \delta^{\acute{r}}( \acute{V}_{1}^{T} H^{+} ( \varrho - \omega) )
\end{equation}
where $ s = q - r $ and $ \acute{s} \equiv q - \acute{r} $. With the help of these, one can choose normalization constants $ e^{i C} $, $ e^{i K} $ such that one brings the wavefunctions to the standard form
\begin{equation}\label{key}
_{\dot{\mathtt{c}}}\langle \varrho \vert \omega \rangle_{\dot{\mathtt{c}}} = \delta^{q}(\varrho - \omega)
\end{equation}
\begin{equation}\label{varrhoomprodacute}
_{\dot{\mathtt{m}}}\langle \varrho \vert \omega \rangle_{\dot{\mathtt{m}}} = \delta^{q}(\varrho - \omega )
\end{equation}
Unless stated otherwise, we will assume this normalization for all eigenstates.

\section{A Unitary Evolution Model}
Having at our disposal the reduced constraint surfaces $ (\tilde{\mathcal{D}}_{n}, \omega_{n}) $ symplectomorphic to the reduced space of solutions $ (\tilde{\mathcal{S}}_{0t}, \omega) $ defined in section \ref{globsol}, it seems natural to deploy the standard canonical quantization procedure to obtain the corresponding quantum system. In this section, we will give its brief description and discuss some of its remarkable features as well as drawbacks.

We remind that $ \tilde{\mathcal{D}}_{n} $ is the space of equivalence classes of symplectically equivalent phase-space data at time-slice $ n $. Together with the induced form $ \omega_{n} $, it is a symplectic vector space. We shall denote its dimension by $ \dim(\tilde{\mathcal{D}}_{n}) \equiv 2b $. For greater technical ease, we will build the quantum theory on the corresponding representative space $ (\dot{\mathcal{D}}_{n}, \omega_{n}) $ fixed by the choice of $ \dot{\mathcal{S}}_{0t} $, whose dimension is also $ \dim(\dot{\mathcal{D}}_{n}) = 2b $. Altogether, the use of representative spaces makes no difference; one only needs to keep in mind that the choice of $ \dot{\mathcal{S}}_{0t} $ is arbitrary, and as such cannot bear any physical significance.

Because $ (\dot{\mathcal{S}}_{0t}, \omega ) $ is symplectic, we can choose a symplectic basis $ \{ d_{R} \}_{R = 1}^{2b} $ in it and write a representative solution $ u \in \dot{\mathcal{S}}_{0t} $ in this basis as $ u = u_{R} d_{R} $. Recall that in this basis, the symplectic product of representative solutions $ u,v \in \dot{\mathcal{S}}_{0t} $ has the form
\begin{equation}\label{symplunit}
\omega(u,v) = u_{R} ~ \sigma_{RS} ~ v_{S} = u_{C} v_{C+b} - u_{C+b} v_{C}
\end{equation}
with an implicit summation over $ C = 1, ..., b $ or over $ R,S = 1, ..., 2b $. The choice of symplectic basis also induces a splitting into coordinates $ u_{C} $ and momenta $ u_{C+b} $. Let us fix the corresponding basis-dependent configuration space $ \dot{\mathcal{D}}_{0t}^{\mathcal{Q}} $ to be the space spanned by $ \{ d_{C} \}_{C = 1}^{b} $. If we look at the individual time-slices, we see that by virtue of the time-slice projection $ \mathbb{P}_{n} $, our symplectic basis $ \{ d_{R} \}_{R = 1}^{2b} $ of the representative solutions space $ \dot{\mathcal{S}}_{0t} $ generates bases $ \{ d_{nR} \}_{R = 1}^{2b} $ of the constraint surfaces $ \dot{\mathcal{D}}_{n} $. Moreover, these bases are symplectic w.r.t. the time-slice symplectic forms $ \omega_{n} $. They fix the corresponding configuration spaces $ \dot{\mathcal{D}}_{n}^{\mathcal{Q}} $. A point $ u_{n} \in \dot{\mathcal{D}}_{n} $ is written in basis as $  u_{n} = u_{n R} d_{n R} $ and a point $ v_{n} \in \dot{\mathcal{D}}_{n}^{\mathcal{Q}} $ as $ v_{n} = v_{n C} d_{n C} $.

It is now possible to straightforwardly apply our quantization procedure outlined in the preceding section. At every time-slice $ n $, we introduce the Hilbert space $ \mathcal{H}_{n} \equiv L^{2}(\dot{\mathcal{D}}_{n}^{\mathcal{Q}}) $ and equip it with the coordinate and momentum operators $ \hat{u}_{n C} $ and $ \hat{u}_{n C+b} $, respectively. Physically speaking, these measure the values of coordinates $ u_{n C} $ and momenta $ u_{n C+b} $ in the above chosen symplectic basis $ \{ d_{n R} \}_{R = 1}^{2b} $ of $ \dot{\mathcal{D}}_{n} $. Next one constructs the rigged Hilbert spaces $ \Psi_{n} \subset \mathcal{H}_{n} \subset \Psi_{n}^{\times} $ and $ \Psi_{n} \subset \mathcal{H}_{n} \subset \Psi_{n}^{\overline{\times}} $, including the space $ \Psi_{n} $ of test functions and the spaces $ \Psi^{\times}_{n} $, $ \Psi^{\overline{\times}}_{n} $ of right states (antilinear functionals) and left states (linear functionals), respectively. Then one may define the eigenstates $ \vert \alpha_{n} \rangle_{\mathtt{c}}, \vert \beta_{n} \rangle_{\mathtt{m}} \in \Psi^{\times}_{n} $ by the equations
\begin{equation}\label{newgeneig}
\begin{aligned}
\hat{u}_{nC} \vert \alpha_{n} \rangle_{\mathtt{c}} &= \alpha_{nC} \vert \alpha_{n} \rangle_{\mathtt{c}} \\
\hat{u}_{nC+b} \vert \beta_{n} \rangle_{\mathtt{m}} &= \beta_{nC} \vert \beta_{n} \rangle_{\mathtt{m}}
\end{aligned}
\end{equation}
with $ b $-tuples of eigenvalues $ \alpha_{n}, \beta_{n} \in \mathbb{R}^{b} $. For convenience, we may not write $ \vert \alpha_{n} \rangle_{\mathtt{c}} $ but rather $ \vert \alpha \rangle_{n\mathtt{c}} $ for the eigenstates; the latter notation will be preferred throughout this section. Thanks to the above construction, one can take the full advantage of the Dirac formalism, especially the essential results \eqref{formfieldmom} providing wavefunctions and resolutions of identity.

The definition of the quantum system is almost complete. Our last task is to specify an evolutionary mapping between the rigged Hilbert spaces at individual time-slices. There is no need to make extra effort in order to do that, as the mapping is naturally induced by the classical solutions. According to Observation \ref{O3}, there is a symplectomorphism between $ (\dot{\mathcal{S}}_{0t}, \omega) $ and any $ (\dot{\mathcal{D}}_{n}, \omega_{n}) $, and thus also between any $ (\dot{\mathcal{D}}_{n}, \omega_{n}) $ and $ (\dot{\mathcal{D}}_{m}, \omega_{m}) $ for all $ n,m \in \{ 0, ..., t \} $. Observe that in the chosen coordinates, its form is extremely simple:
\begin{equation}\label{unitevol}
u_{R} = u_{n R} = u_{m R}
\end{equation}
In other words, all the time evolution has been captured in the evolution of the basis vectors, i.e. the solutions $ d_{R} $, and we give coordinates (and momenta) with respect to these. Thanks to the trivial prescription \eqref{unitevol}, we are able to straightforwardly define a fully analogical evolution prescription for the quantum system. For right states, it shall be $ \mathbb{U}_{mn}: \Psi_{n}^{\times} \rightarrow \Psi_{m}^{\times} $ fixed on the coordinate eigenstate basis by
\begin{equation}\label{Uunitdef}
\mathbb{U}_{mn} \vert \alpha \rangle_{n\mathtt{c}} = \vert \alpha \rangle_{m\mathtt{c}}
\end{equation}
for any $ \alpha \in \mathbb{R}^{b} $. We remark that the definition needs not be limited to the case $ m > n $, i.e., the evolution also works backwards. At this point we must add two technicalities. First we define an analogical evolution map for left states. It shall be $ \overline{\mathbb{U}_{mn}} : \Psi_{n}^{\overline{\times}} \rightarrow \Psi_{m}^{\overline{\times}}  $, whose action on $ \langle \psi_{n} \vert $ is denoted in the Dirac formalism as $ \langle \psi_{n} \vert \overline{\mathbb{U}_{mn}} $. It shall be given by $ \langle \psi_{n} \vert \overline{\mathbb{U}_{mn}} = \overline{\mathbb{U}_{mn} \vert \psi_{n} \rangle}  $ where the overline stands for the duality between $ \Psi_{n}^{\times} $ and $ \Psi_{n}^{\overline{\times}} $. In other words, if one denotes $ \vert \psi_{m} \rangle \equiv \mathbb{U}_{mn} \vert \psi_{n} \rangle $, one has $  \langle \psi_{n} \vert \overline{\mathbb{U}_{mn}} = \langle \psi_{m} \vert $. Second, we define the reversed evolution map $ \mathbb{U}_{mn}^{\ast} : \Psi_{m}^{\times} \rightarrow \Psi_{n}^{\times} $ by $ \langle \varphi_{n} \vert \mathbb{U}_{mn}^{\ast} \vert \psi_{m} \rangle = \langle \varphi_{n} \vert \overline{\mathbb{U}_{mn}} \vert \psi_{m} \rangle $ for all $ \vert \varphi_{n} \rangle \in \Psi_{n}^{\times} $ and $ \vert \psi_{m} \rangle \in \Psi_{m}^{\times} $. Now, these definitions together with \eqref{formfieldmom} imply
\begin{equation}\label{key}
_{\mathtt{c}n}\langle \gamma \vert \mathbb{U}_{mn}^{\ast} \vert \alpha \rangle_{m\mathtt{c}} = ~_{\mathtt{c}m}\langle \gamma \vert \alpha \rangle_{m\mathtt{c}} = \delta(\alpha - \gamma) = ~_{\mathtt{c}n}\langle \gamma \vert \alpha \rangle_{n\mathtt{c}}
\end{equation}
which means that
\begin{equation}\label{key}
\mathbb{U}_{mn}^{\ast} = \mathbb{U}_{nm}
\end{equation}
This in turn implies
\begin{equation}\label{key}
\mathbb{U}_{mn}^{\ast} \mathbb{U}_{mn} =  \mathbb{U}_{nm} \mathbb{U}_{mn} = \hat{\mathbf{1}}_{n}
\end{equation}
which can be interpreted as \textit{unitarity} of the evolution.

Eventually, thanks to the symplectomorphism between $ (\dot{\mathcal{S}}_{0t}, \omega) $ and $ (\dot{\mathcal{D}}_{n}, \omega_{n}) $, it makes good sense to move away from the idea of separate time-steps and instead view the system as a whole. Note that we could have just as well started with the Hilbert space $ \mathcal{H} \equiv L^{2}(\dot{\mathcal{D}}_{0t}^{\mathcal{Q}}) $. We can therefore erase all the time-slice subindices and forget about the evolution map completely, obtaining an equivalent quantum theory built on (a representative space of classes of symplectically equivalent) solutions.

\subsection{Path Integral}
Consider two states $ \vert \psi_{in} \rangle, \vert \psi_{fi} \rangle \in \Psi^{\times} $ which will represent initial and final state of the system across time interval $ n = 0, ..., t $. Since we are working in the representative solution space now, both states are defined on the whole interval. Nevertheless, we anticipate that $ \vert \psi_{in} \rangle $ shall be subjected to measurement on time-slice $ 0 $, while $ \vert \psi_{fi} \rangle $ shall be subjected to measurement on time-slice $ t $. The phase corresponding to finding the system in \textit{both} of these states, i.e., measuring $ \vert \psi_{in} \rangle $ at $ n = 0 $ and $ \vert \psi_{fi} \rangle $ at $ n = t $, is taken to be $ \langle \psi_{fi} \vert \psi_{in} \rangle $. We can now simulate measurements on every intermediate time-slice by entering copies of the coordinate and momentum resolutions of unity, and pass to path-integral formulation.

Let us abuse notation and write the states in the coordinate resolution of unity on $ \Psi^{\times} $ with subindex $ n $ as means of differentiating the integration variables between time-slices, i.e., $ \vert \alpha_{n} \rangle_{\mathtt{c}} \in \Psi^{\times} $ and similarly for momenta. Then, using \eqref{formfieldmom}, we can rewrite the wavefunction between the initial and final state as
\begin{equation}\label{pi}
\begin{aligned}
\langle \psi_{fi} \vert \psi_{in} \rangle &= \sum_{\alpha_{0}...\alpha_{t}} \sum_{\beta_{1}...\beta_{t}} \langle  \psi_{fi} \vert \alpha_{t} \rangle_{\mathtt{c}} ~_{\mathtt{c}}\langle \alpha_{t} \vert \beta_{t} \rangle_{\mathtt{c}} ~_{\mathtt{c}}\langle \beta_{t} \vert \alpha_{t-1} \rangle_{\mathtt{c}} ~_{\mathtt{c}}\langle \alpha_{t-1} \vert ~ ... \\ & \qquad \qquad ~ \vert \alpha_{1} \rangle_{\mathtt{c}} ~_{\mathtt{c}}\langle \alpha_{1} \vert \beta_{1} \rangle_{\mathtt{c}} ~_{\mathtt{c}}\langle \beta_{1} \vert \alpha_{0} \rangle_{\mathtt{c}} ~_{\mathtt{c}}\langle \alpha_{0} \vert \psi_{in} \rangle = \\
&= \sum_{\alpha_{0}...\alpha_{t}} \sum_{\beta_{1}...\beta_{t}} \langle  \psi_{fi} \vert \alpha_{t} \rangle_{\mathtt{c}} ~_{\mathtt{c}}\langle \alpha_{0} \vert \psi_{in} \rangle ~ (2 \pi)^{-tb} ~ e^{i \check{S}(\alpha_{0}, ..., \alpha_{t}, \beta_{1}, ..., \beta_{t})}
\end{aligned}
\end{equation}
where we have denoted
\begin{equation}\label{checkS}
\check{S}(\alpha_{0}, ..., \alpha_{t}, \beta_{1}, ..., \beta_{t}) = \sum_{n = 0}^{t-1} \beta_{n+1}^{T}( \alpha_{n+1} - \alpha_{n} )
\end{equation}
Of course, we have done nothing here except putting the wavefunction $ \langle \psi_{fi} \vert \psi_{in} \rangle $ into a more complicated form than it actually deserves. One can see that the role of action is played by the function $ \check{S}(\alpha_{0}, ..., \alpha_{t}, \beta_{1}, ..., \beta_{t}) $. One does not expect to get the original action; after all, the integration variables are different from the canonical ones. Besides, the integral \eqref{pi} sums only over initial data $ \dot{\mathcal{D}}_{n} $ of our representative solutions and ignores data from the null space, as well as data that do not give rise to solutions at all. This is of course a consequence of our manipulations with the phase space which enabled us to find a conserved symplectic structure and define the quantum theory canonically.

\subsection{Unimportance of the Representative Space}
In the context of path integral, one can check easily that the choice of the representative space of solutions $ \dot{\mathcal{S}}_{0t} $ has no relevance for physical results drawn from the theory. In this quick intermezzo, we show how to do it.

The eigenvalue vector $ \alpha_{n} \in \mathbb{R}^{b} $ can be interpreted as a $ b $-tuple of coordinates $ \alpha_{n C} $ of a point $ \alpha_{n} \in \dot{\mathcal{S}}_{0t} $ w.r.t. the basis $ \{ d_{R} \}_{R = 1}^{2b} $ (we abuse notation by neglecting to distinguish between the two). This point is of course not fixed, since we did not specify the momenta $ \alpha_{n C+b} $. We can do the same with $ \beta_{n} \in \mathbb{R}^{b} $, which is in turn interpreted as a $ b $-tuple of momenta $ \beta_{n C+b} $ of a point $ \beta_{n} \in \dot{\mathcal{S}}_{0t} $, whose coordinates $ \beta_{n C} $ are unspecified. This has some justification in the interpretation of eigenstates. Nonetheless, we do it only because it allows us to pull a little trick. From \eqref{symplunit} we know $ \omega(\alpha_{n}, \beta_{n}) = \alpha_{n C} \beta_{n C+b} - \alpha_{n C+b} \beta_{n C} $ so we can rewrite $ \beta_{n}^{T} \alpha_{n} = \beta_{n C+b} \alpha_{n C} = \omega(\alpha_{n}, \beta_{n}) + \alpha_{n C+b} \beta_{n C} $. The single time-step contribution to \eqref{checkS} can then be expressed as
\begin{equation}\label{key}
\check{S}_{n+1} \equiv \beta_{n+1}^{T}( \alpha_{n+1} - \alpha_{n} ) = \omega(\alpha_{n+1}, \beta_{n+1}) - \omega(\alpha_{n}, \beta_{n+1}) + \check{C}_{n+1}
\end{equation}
with $ \check{C}_{n+1} \equiv  \beta_{n+1 C} (\alpha_{n+1 C+b} - \alpha_{n C+b}) $. If we further denote $ \check{C}_{0t} \equiv \sum_{n = 0}^{t-1} \check{C}_{n+1} $, we can eventually express
\begin{equation}\label{key}
\begin{aligned}
\langle \psi_{fi} \vert \psi_{in} \rangle &= e^{- i \check{C}_{0t}} \sum_{\alpha_{0}...\alpha_{t}} \sum_{\beta_{1}...\beta_{t}} \langle  \psi_{fi} \vert \alpha_{t} \rangle_{\mathtt{c}} ~_{\mathtt{c}}\langle \alpha_{0} \vert \psi_{in} \rangle \\ & \qquad \qquad ~ (2 \pi)^{-tb} ~ e^{i \sum_{n = 0}^{t-1} \left[ \omega(\alpha_{n+1}, \beta_{n+1}) - \omega(\alpha_{n}, \beta_{n+1}) \right] }
\end{aligned}
\end{equation}
Since the overall phase $ e^{- i \check{C}_{0t}} $ is irrelevant and the integrand only contains symplectic products of solutions, this form explicitly shows that the result is irrespective of the particular solutions $ \alpha_{n}, \beta_{n} \in \dot{\mathcal{S}}_{0t} $ chosen as representatives of the classes $ [ \alpha_{n} ], [\beta_{n}] \in \tilde{\mathcal{S}}_{0t} $. The choice of $ \dot{\mathcal{S}}_{0t} $ therefore does not bear any physical meaning.

\subsection{Discussion}
Having set up our first version of the quantum theory of a discretely canonically evolved linear dynamical system, we will take a few lines to discuss its qualities. First of all, we want to draw the reader's attention to the fact that the whole work connected to the irregularity of the discrete evolution was carried out on the classical level. The model was rendered regular by finding an appropriate phase space in one-to-one correspondence to an appropriate solution space, and we made sure that both these spaces are symplectic. Thanks to this, the quantization procedure was straightforward, with a trivial (or alternatively none) prescription for evolution. The above described approach had two major consequences. First, the evolution turned out unitary. Second, all the structures within the model are based and thus dependent on classical solutions. This has serious repercussions. To actually draw any results from the theory, like for instance $ \langle \psi_{fi} \vert \psi_{in} \rangle $, one would have to first fully solve the classical system. In this account, the path integral formulation given by \eqref{pi} utterly loses its value since the dynamics is no longer contained in the integrand. Instead, an essential role is played by the integration variables themselves which have to be determined from the classical solutions.

Another questionable feature about the present model is that its functionality is dependent upon the time interval under consideration. Say that we have a system spanning between time-steps $ 0 $ and $ t $. If we choose $ m $ such that $ 0 < m < t $ and limit the system to time-steps between $ 0 $ and $ m $, we in general get a different quantum theory, because there may be (symplectically inequivalent) classical solutions between $ 0 $ and $ m $ which do not extend to $ t $. Consequently, the quantum evolution described in this manner is not causal, and as such cannot be considered physical, but rather effective. In other words, the procedure in which we render classes of symplectically equivalent solutions amounts to fixing the classical system so that we can easily define the quantum counterpart. We should not be surprised that we obtain a different system.

Let us compare our unitary quantum model to the model of \cite{Hoehn2014a}. Therein, the system is first quantized separately at each time-slice $ n $ which gives rise to the co-called \textit{kinematical Hilbert spaces} $ \mathcal{H}_{n}^{\text{kin}} $. From these one constructs the \textit{pre-physical} and \textit{post-physical Hilbert spaces} $ ^{-}\mathcal{H}_{n}^{\text{phys}} $ and $ ^{+}\mathcal{H}_{n}^{\text{phys}} $ with the employment of special improper projectors that implement the pre-constraints and the post-constraints, respectively. The physical Hilbert spaces therefore contain only those states which satisfy the appropriate quantum constraint. These are provided with the special \textit{physical inner product}. Then one can define a unitary map $ U_{n \rightarrow n+1} $ from $ ^{-}\mathcal{H}_{n}^{\text{phys}} $ to $ ^{+}\mathcal{H}_{n+1}^{\text{phys}} $, much like in the classical case. The details of this construction are given in section 4 of \cite{Hoehn2014a}. The constrained quantum evolution moves can be composed into an effective move across multiple time-steps as described in section 6 of that reference.

As it turns out, all of the previously discussed features of our unitary model can be traced in \cite{Hoehn2014a} too. Thanks to the introduction of physical Hilbert spaces, the model admits a unitary evolution map. Moreover, it also displays interval dependence. This feature is perhaps most clearly described in section 6.2, arguing that for $ 0 < m < t $, the number of constraints arising in the interval between $ 0 $ and $ t $ can be only equal to or higher than the number of constraints between $ 0 $ and $ m $ (this applies to both classical and quantum constraints). The effective quantum evolution move across multiple time-steps takes in account all these constraints and the resulting initial and final physical Hilbert spaces are fixed accordingly. It follows that the unitary isomorphism $ U_{n \rightarrow m} $ which maps between the physical Hilbert spaces at question will depend on the chosen interval as well. This is exactly what we observe in our model. Finally let us mention the path integral provided in section 6.4 of \cite{Hoehn2014a}. It is given in terms of formal regularized measures $ d\xi_{l}(x_{l}) $ defined on each phase space $ \mathcal{P}_{l} $. Again, the measures cannot be determined locally since they respond to constraints across the whole time interval of the effective move. Besides the general treatment, this can be seen also in the application to linear systems treated in \cite{Hoehn2014}. In section 10.4 of the reference, the author gives and explicit prescription for the propagator $ K_{0\rightarrow 2} $ of the move between time-slices $ 0 $ and $ 2 $---see equation (10.25) and the subsequent discussion---and explains that the constraints propagate not only forward but also backwards, thereby limiting the number of propagating degrees of freedom. In this lies a clear parallel to our prescription \eqref{pi} which includes only the integration over propagating degrees of freedom. Meanwhile, one would arguably like to have a theory whose evolution map is local in time so that the system does not experience back-propagation of constraints and therefore allows for causal evolution. Moreover, it would be preferable to have a path integral in terms of the canonical coordinates with the dynamics captured fully in the integrand.

\section{A Non-Unitary Evolution Model}
In view of the previous section, we set out to search for a more natural and practical alternative for the quantum evolution. In doing so, we decide to quantize the system separately at each time-step $ n $ without any prior treatment, as it is done in \cite{Hoehn2014, Hoehn2014a}. However, we will not take the path of reducing the Hilbert spaces. On the contrary: we shall look for an evolution scheme which maps every element of the initial state space to the state space one time-step later. This is not to say the scheme will be onto, it will of course not. Consequently, it will also turn out non-unitary. We will implement the pre-constraint in effectively the same way it was done in \cite{Hoehn2014, Hoehn2014a}, but it will actually not yield any divergences for us to worry about. More trouble will be caused by the null space $ \Lambda_{n+1} \subset \mathcal{C}_{n+1}^{+} $ generated by the free parameters $ \lambda_{n+1} $ of the classical Hamiltonian evolution (in the terminology of \cite{Hoehn2014, Hoehn2014a}, it is called the \textit{post-orbit}) but this trouble can be overcome with a suitable choice of regularization. Besides avoiding the questionable features described in the previous section, the present model has another important advantage: it is characterized by a single relatively simple evolution map. Moreover, we spend effort to build this map in a natural way as a direct quantum analogy of the classical evolution. In doing so, we benefit heavily from the classical framework introduced in \cite{Kaninsky2020a} and the ability to perform arbitrary symplectic transformations of quantum-mechanical observables established in \cite{Kaninsky2020}. Thanks to these, we will be able to derive an expression for the propagator without any outer input.

Let us proceed with the construction. Once again we employ the definitions from section \ref{quant}. We start by introducing the Hilbert spaces $ \mathcal{F}_{n} \equiv L^{2}(\mathcal{Q}_{n}) $ and the coordinate and momentum operators $ \hat{y}_{nA} $ and $ \hat{y}_{nA+q} $, respectively, where $ A = 1, ..., q $. These measure the values of coordinates $ y_{nA} $ and momenta $ y_{nA+q} $ in the canonical basis $ \{ e_{n I} \}_{I = 1}^{2q} $ of $ \mathcal{P}_{n} $. Next we consider the suitable rigged Hilbert spaces giving rise to the useful spaces of functionals (states) $ \Phi^{\times}_{n} $ and $ \Phi^{\overline{\times}}_{n} $. Then we may define the eigenstates $ \vert \alpha_{n} \rangle_{\mathtt{c}}, \vert \beta_{n} \rangle_{\mathtt{m}} \in \Phi^{\times}_{n} $ by the equations
\begin{equation}\label{newgeneig}
\begin{aligned}
\hat{y}_{nA} \vert \alpha_{n} \rangle_{\mathtt{c}} &= \alpha_{nA} \vert \alpha_{n} \rangle_{\mathtt{c}}\\
\hat{y}_{nA+q} \vert \beta_{n} \rangle_{\mathtt{m}} &= \beta_{nA} \vert \beta_{n} \rangle_{\mathtt{m}}
\end{aligned}
\end{equation}
This construction allow to take advantage of the Dirac formalism, especially the essential formulas \eqref{formfieldmom} providing wavefunctions and resolutions of identity, and all results derived from these.

\subsection{Eigenstates of the Adapted Observables}

In this section we give explicit expressions for the eigenstates of all observables (i.e., coordinates and momenta) measured in the adapted coordinates described in section \ref{constrsurf}. Therein we have provided new symplectic bases $ \{ \dot{e}_{nI} \}_{I = 1}^{2q} $ and $ \{ \ddot{e}_{n+1I} \}_{I = 1}^{2q} $ of $ \mathcal{P}_{n} $ and $ \mathcal{P}_{n+1} $, respectively, defined by the symplectic matrices
\begin{equation}\label{dotW2}
\dot{W}_{n} = \begin{pmatrix}
-U^{T}L_{n} & U^{T} \\
-U^{T} & 0
\end{pmatrix}
\end{equation}
and
\begin{equation}\label{ddotW2}
\ddot{W}_{n+1} = \begin{pmatrix}
\bar{\Sigma} V^{T} & 0 \\
- \bar{\Sigma}^{-1} V^{T} \bar{R}_{n+1} & \bar{\Sigma}^{-1} V^{T}
\end{pmatrix}
\end{equation}
Recall that the matrices $ U $, $ V $ and $ \Sigma_{r} $ originate in the singular value decomposition of $ R_{n} $ and $ \bar{\Sigma} $ is just a version of $ \Sigma_{r} $ extended to a $ q \times q $ matrix by inserting the $ s_{n} \times s_{n} $ unit matrix at the diagonal. According to \eqref{uddotX}, \eqref{uddotM} and \eqref{uddotMq}, the one-step Hamiltonian evolution $ u_{n+1} = \mathbb{H}_{n+1}(\lambda_{n+1}) u_{n} \in \mathcal{C}_{n+1}^{+} $ of a vector $ u_{n} \in \mathcal{C}_{n}^{-} $ is given in the adapted coordinates by
\begin{equation}\label{dotuddotu}
\begin{aligned}
\ddot{u}_{n+1E} &= \dot{u}_{nE}\\
\ddot{u}_{n+1E+q} &= \dot{u}_{nE+q}\\
\ddot{u}_{n+1M} &= \lambda_{n+1 M-r_{n}}\\
\ddot{u}_{n+1M+q} &= \dot{u}_{nM} = 0
\end{aligned}
\end{equation}
for all $ E = 1,..., r_{n} $ and all $ M = r_{n}+1, ..., q $. We remark that we chose to switch the coordinates and momenta in the irregular sector in order to have a simpler transformation matrix \eqref{ddotW2}.

We shall now introduce new observables $ \hat{\dot{y}}_{nA} $, $ \hat{\dot{y}}_{nA+q} $ measuring coordinates and momenta w.r.t. the basis $ \{ \dot{e}_{nI} \}_{I = 1}^{2q} $ of $ \mathcal{P}_{n} $ and consider the corresponding new eigenstates $ \vert \alpha_{n} \rangle_{\dot{\mathtt{c}}}, \vert \beta_{n} \rangle_{\dot{\mathtt{m}}} \in \Phi^{\times}_{n} $ defined by the equations $ \hat{\dot{y}}_{nA} \vert \alpha_{n} \rangle_{\dot{\mathtt{c}}} = \alpha_{nA} \vert \alpha_{n} \rangle_{\dot{\mathtt{c}}} $ and $ \hat{\dot{y}}_{nA+q} \vert \beta_{n} \rangle_{\dot{\mathtt{m}}} = \beta_{nA} \vert \beta_{n} \rangle_{\dot{\mathtt{m}}} $. By plugging from the corresponding symplectic matrix \eqref{dotW2} into the formula \eqref{cdot}, we get
\begin{equation}\label{alphandot}
\vert \alpha_{n} \rangle_{\dot{\mathtt{c}}} = \sum_{\xi_{n}} \vert \xi_{n} \rangle_{\mathtt{c}} ~ e^{i C} e^{ i \left( \frac{1}{2} \xi_{n}^{T} L_{n} \xi_{n} + \xi_{n}^{T} U \alpha_{n} \right) }
\end{equation}
where $ \xi_{n} \in \mathbb{R}^{q} $ and we used the trivial singular value decomposition $ U(R_{n})^{T} = U(R_{n})^{T} \mathbf{1} \mathbf{1} $. Before we continue, we are going to find the appropriate normalization stemming from our requirement $ ~_{\dot{\mathtt{c}}} \langle \gamma_{n} \vert \alpha_{n} \rangle_{\dot{\mathtt{c}}} = \delta( \gamma_{n} - \alpha_{n} ) $. Using \eqref{dotcdotc}, we find
\begin{equation}\label{key}
~_{\dot{\mathtt{c}}} \langle \gamma_{n} \vert \alpha_{n} \rangle_{\dot{\mathtt{c}}} = e^{-2 \Im C} (2\pi)^{q} \delta^{q} \left( \gamma_{n} - \alpha_{n} \right)
\end{equation}
where we recalled that $ \delta( O \theta) = \delta(\theta) $ for $ O $ orthogonal. If we fix the arbitrary real part of $ C $ to zero, we get the normalization
\begin{equation}\label{key}
e^{iC} = (2\pi)^{-q/2}
\end{equation}
and we end up with
\begin{equation}\label{alphandot2}
\vert \alpha_{n} \rangle_{\dot{\mathtt{c}}} = \sum_{\xi_{n}} \vert \xi_{n} \rangle_{\mathtt{c}} ~ (2\pi)^{-q/2} ~ e^{ i \left( \frac{1}{2} \xi_{n}^{T} L_{n} \xi_{n} + \xi_{n}^{T} U \alpha_{n} \right) }
\end{equation}
Let us continue by expressing the momentum eigenstates. By plugging from the matrix \eqref{dotW2} into \eqref{mdot}, we immediately obtain
\begin{equation}\label{betandot}
\vert \beta_{n} \rangle_{\dot{\mathtt{m}}} = \vert - U \beta_{n} \rangle_{\mathtt{c}} ~ e^{i K}
\end{equation}
The appropriate normalization given by the requirement $ ~_{\dot{\mathtt{m}}} \langle \gamma_{n} \vert \beta_{n} \rangle_{\dot{\mathtt{m}}} = \delta( \gamma_{n} - \beta_{n} ) $ is easily found to be $ e^{iK} = 1 $, so we have
\begin{equation}\label{betandot2}
\vert \beta_{n} \rangle_{\dot{\mathtt{m}}} = \vert - U \beta_{n} \rangle_{\mathtt{c}}
\end{equation}
Eventually, we shall give the momentum wavefunction $ _{\dot{\mathtt{c}}} \langle \alpha_{n} \vert \beta_{n} \rangle_{\dot{\mathtt{m}}} $. After plugging in from \eqref{alphandot2} and \eqref{betandot2}, one gets
\begin{equation}\label{dotcm}
_{\dot{\mathtt{c}}} \langle \alpha_{n} \vert \beta_{n} \rangle_{\dot{\mathtt{m}}} = (2\pi)^{-q/2} ~ e^{i \alpha_{n}^{T} \beta_{n}} ~ e^{- i \frac{1}{2} \beta_{n}^{T} U^{T} L_{n} U \beta_{n}} 
\end{equation}
We must appreciate how simple the wavefunction is. Also, it contains the canonical momentum wavefunction $ _{\mathtt{c}} \langle \alpha_{n} \vert \beta_{n} \rangle_{\mathtt{m}} = (2\pi)^{-q/2} ~ e^{i \alpha_{n}^{T} \beta_{n}} $ which is a good sign.

\vspace{\baselineskip}

Now we shall move to the next time-slice. Analogically to the previous treatment, we introduce the observables $ \hat{\ddot{y}}_{n+1A} $, $ \hat{\ddot{y}}_{n+1A+q} $ with respect to the new symplectic basis $ \{ \ddot{e}_{n+1 I} \}_{I = 1}^{2q} $ of $ \mathcal{P}_{n+1} $ and consider the corresponding new generalized eigenvectors $ \vert \alpha_{n+1} \rangle_{\ddot{\mathtt{c}}}, \vert \beta_{n+1} \rangle_{\ddot{\mathtt{m}}} \in \Phi^{\times}_{n+1} $ defined by $ \hat{\ddot{y}}_{n+1A} \vert \alpha_{n+1} \rangle_{\ddot{\mathtt{c}}} = \alpha_{n+1A} \vert \alpha_{n+1} \rangle_{\ddot{\mathtt{c}}} $ and $ \hat{\ddot{y}}_{n+1A+q} \vert \beta_{n+1} \rangle_{\ddot{\mathtt{m}}} = \beta_{n+1A} \vert \beta_{n+1} \rangle_{\ddot{\mathtt{m}}} $. We plug from \eqref{ddotW2} into the formula \eqref{cdot}, obtaining
\begin{equation}\label{key}
\vert \alpha_{n+1} \rangle_{\ddot{\mathtt{c}}} = \vert V \bar{\Sigma}^{-1} \alpha_{n+1} \rangle_{\mathtt{c}} ~ e^{i C}
\end{equation}
The normalization given by $ ~_{\ddot{\mathtt{c}}} \langle \gamma_{n+1} \vert \alpha_{n+1} \rangle_{\ddot{\mathtt{c}}} = \delta( \gamma_{n+1} - \alpha_{n+1} ) $ is easily found from
\begin{equation}\label{key}
\begin{aligned}
~_{\ddot{\mathtt{c}}} \langle \gamma_{n+1} \vert \alpha_{n+1} \rangle_{\ddot{\mathtt{c}}} &= e^{-2 \Im C} ~ \delta^{q}\left( V \bar{\Sigma}^{-1} (\alpha_{n+1} - \gamma_{n+1}) \right) = \\ &= \abs{\det \Sigma_{r} } ~ e^{-2 \Im C} ~ \delta^{q}\left( \gamma_{n+1} - \alpha_{n+1} \right)
\end{aligned}
\end{equation}
to be
\begin{equation}\label{key}
e^{iC} = \abs{\det \Sigma_{r} }^{-1/2}
\end{equation}
We thus have
\begin{equation}\label{ddotc}
\vert \alpha_{n+1} \rangle_{\ddot{\mathtt{c}}} = \abs{\det \Sigma_{r} }^{-1/2} ~ \vert V \bar{\Sigma}^{-1} \alpha_{n+1} \rangle_{\mathtt{c}}
\end{equation}
Let us look at momenta. By plugging from \eqref{dotW2} into \eqref{mdot}, we find
\begin{equation}\label{ddotm}
\vert \beta_{n+1} \rangle_{\ddot{\mathtt{m}}} = \sum_{\chi_{n+1}} ~ \vert V \chi_{n+1} \rangle_{\mathtt{c}} ~ e^{i K} ~ e^{i \left( \frac{1}{2} \chi_{n+1}^{T} V^{T} \bar{R}_{n+1} V \chi_{n+1} + \chi_{n+1}^{T} \bar{\Sigma} \beta_{n+1} \right) }
\end{equation}
and using \eqref{dotmdotm}, we obtain the normalization
\begin{equation}\label{ddotnorm}
e^{iK} = (2\pi)^{-q/2} \abs{\det \Sigma_{r} }^{1/2}
\end{equation}
Thus the final expression is
\begin{equation}\label{ddotm2}
\vert \beta_{n+1} \rangle_{\ddot{\mathtt{m}}} = \sum_{\chi_{n+1}} ~ \vert V \chi_{n+1} \rangle_{\mathtt{c}} ~ (2\pi)^{-q/2} ~ \abs{\det \Sigma_{r} }^{1/2} ~ e^{i \left( \frac{1}{2} \chi_{n+1}^{T} V^{T} \bar{R}_{n+1} V \chi_{n+1} + \chi_{n+1}^{T} \bar{\Sigma} \beta_{n+1} \right) }
\end{equation}
We close by giving the momentum wavefunction $ _{\ddot{\mathtt{c}}} \langle \alpha_{n+1} \vert \beta_{n+1} \rangle_{\ddot{\mathtt{m}}} $. Plugging in from the equations \eqref{ddotc} and \eqref{ddotm2}, we get
\begin{equation}\label{ddotcm}
~_{\ddot{\mathtt{c}}} \langle \alpha_{n+1} \vert \beta_{n+1} \rangle_{\ddot{\mathtt{m}}} = (2\pi)^{-q/2} ~ e^{i \alpha_{n+1}^{T} \beta_{n+1}} ~ e^{i\frac{1}{2} \alpha_{n+1}^{T} \bar{\Sigma}^{-1} V^{T} \bar{R}_{n+1} V \bar{\Sigma}^{-1} \alpha_{n+1}}
\end{equation}
Again, the wavefunction is outstandingly simple and contains the canonical momentum wavefunction $ _{\mathtt{c}} \langle \alpha_{n+1} \vert \beta_{n+1} \rangle_{\mathtt{m}} = (2\pi)^{-q/2} ~ e^{i \alpha_{n+1}^{T} \beta_{n+1}} $.

\subsection{The One-Step Evolution Scheme}

We eventually have the necessary structures to introduce the promised evolution scheme. It shall be realized by the map $ \mathbb{U}_{n+1}: \Phi^{\times}_{n} \rightarrow \Phi^{\times}_{n+1} $ between the generalized time-slice state spaces of our quantum system. We will define it in three steps. In the first step, we introduce a partial evolution map $ \bar{\mathbb{U}}_{n+1}: \Phi^{\times}_{n} \rightarrow \Phi^{\times}_{n+1} $ defined only for the adapted coordinate eigenstates $ \vert \alpha_{n} \rangle_{\dot{\mathtt{c}}} $ satisfying the pre-constraint (i.e., $ \alpha_{n} \in \mathbb{R}^{q} $ such that $ \alpha_{n M} = 0 $ for all $ M = r_{n}+1, ..., q $). In the second step, we consider a suitable generalization to all coordinate eigenstates. In the third step, we check the action of the mapping on the momentum eigenstates, and fix it so that they are correctly transported too.

We suggest that $ \bar{\mathbb{U}}_{n+1} $ should be defined as follows. Assume that
\begin{equation}\label{key}
\alpha_{n} = \begin{pmatrix}
\varrho_{n}\\
0
\end{pmatrix}
\end{equation}
Then we fix
\begin{equation}\label{Un+1bar}
\bar{\mathbb{U}}_{n+1} \vert \alpha_{n} \rangle_{\dot{\mathtt{c}}} = \sum_{\lambda_{n+1}} \mathcal{l}(\lambda_{n+1}) ~ \vert \alpha_{n+1} \rangle_{\ddot{\mathtt{c}}} ~ f(\alpha_{n+1})
\end{equation}
with
\begin{equation}\label{key}
\alpha_{n+1} \equiv \begin{pmatrix}
\varrho_{n}\\
\lambda_{n+1}
\end{pmatrix}
\end{equation}
and a complex locally integrable function $ f(\alpha_{n+1}) = f(\varrho_{n}, \lambda_{n+1}) $ of the vector $ \alpha_{n+1} $ which will be specified later. Note that we have built the image of $ \bar{\mathbb{U}}_{n+1} $ in such a way that its regular part $ \varrho_{n} \in \mathbb{R}^{r_{n}} $ is copied from the previous time-slice in fashion of the classical evolution, while the singular part is made of a linear combination of the classical free parameters $ \lambda_{n+1} \in \mathbb{R}^{s_{n}} $ parametrizing the null space; cf. \eqref{dotuddotu}. We have further included another complex locally integrable function $ \mathcal{l}(\lambda_{n+1}) $ which is supposed to serve as a regularization of the integral over $ \lambda_{n+1} $. We would of course prefer the linear combination over $ \lambda_{n+1} $ to be uniform ($ \mathcal{l} = const. $), but that will soon turn out problematic. The regularization procedure, including a special choice of $ \mathcal{l} $, will be provided in the next section.

It remains a question what happens to those coordinate eigenstates which are not of the form assumed in \eqref{Un+1bar}. We answer by introducing the map $ \Xi_{n} : \Phi^{\times}_{n} \rightarrow \Phi^{\times}_{n} $ satisfying
\begin{equation}\label{Chi}
\Xi_{n} \sum_{\iota_{n}} \vert \alpha_{n} \rangle_{\dot{\mathtt{c}}} ~ h(\alpha_{n}) = \vert P_{r_{n}} \alpha_{n} \rangle_{\dot{\mathtt{c}}} ~ h(P_{r_{n}} \alpha_{n})
\end{equation}
where
\begin{equation}\label{key}
\alpha_{n} \equiv \begin{pmatrix}
\varrho_{n} \\
\iota_{n}
\end{pmatrix} \qq{and} P_{r_{n}} \equiv \begin{pmatrix}
\mathbf{1} & 0 \\
0 & 0
\end{pmatrix} \qq{so that} P_{r_{n}} \alpha_{n} = \begin{pmatrix}
\varrho_{n} \\
0
\end{pmatrix} 
\end{equation}
with $ \varrho_{n} \in \mathbb{R}^{r_{n}} $ and $ \iota_{n} \in \mathbb{R}^{s_{n}} $. Above, $ h : \mathbb{R}^{q} \rightarrow \mathbb{R} $ is any locally integrable function. In a more relaxed notation, a physicist could write $ \Xi_{n} \vert \alpha_{n} \rangle_{\dot{\mathtt{c}}} = \vert \alpha_{n} \rangle_{\dot{\mathtt{c}}} ~ \delta^{s_{n}}( J_{s_{n}}^{T} \alpha_{n}) $ with the $ q \times s_{n} $ matrix
\begin{equation}\label{key}
J_{s_{n}} \equiv \begin{pmatrix}
0 \\
\mathbf{1}
\end{pmatrix}
\end{equation}
Simply put, the map $ \Xi_{n} $ effectively picks only those eigenstates $ \vert \alpha_{n} \rangle_{\dot{\mathtt{c}}} $ which satisfy $ \alpha_{n M} = 0 $ and discards the rest. Note that essentially the same procedure is used in section 4.2 of \cite{Hoehn2014} to map vectors from the kinematical Hilbert space $ \mathcal{H}_{0}^{\text{kin}} $ to the pre-physical Hilbert space $ ^{-}\mathcal{H}_{0}^{\text{phys}} $ (the corresponding map is denoted $ ^{-}\mathbb{P}_{0} $ and referred to as the \textit{pre-projector}). We can now define
\begin{equation}\label{Un+1}
\mathbb{U}_{n+1} = \bar{\mathbb{U}}_{n+1} \Xi_{n}
\end{equation}
Thanks to this, the domain of $ \mathbb{U}_{n+1} $ is the whole $ \Phi^{\times}_{n} $. Employing together all the definitions \eqref{Un+1bar}, \eqref{Chi} and \eqref{Un+1}, we may write in physical notation
\begin{equation}\label{Un+1c}
\mathbb{U}_{n+1} \vert \alpha_{n} \rangle_{\dot{\mathtt{c}}} = \bar{\mathbb{U}}_{n+1} \Xi_{n} \vert \alpha_{n} \rangle_{\dot{\mathtt{c}}} = \sum_{\lambda_{n+1}} \mathcal{l}~ \vert \alpha_{n+1} \rangle_{\ddot{\mathtt{c}}} ~ \delta^{s_{n}}( \iota_{n}) ~ f(\alpha_{n+1})
\end{equation}
where we use once again the short-hands
\begin{equation}\label{key}
\alpha_{n} \equiv \begin{pmatrix}
\varrho_{n} \\
\iota_{n}
\end{pmatrix}, \qquad \alpha_{n+1} \equiv \begin{pmatrix}
\varrho_{n}\\
\lambda_{n+1}
\end{pmatrix}
\end{equation}
In the above expression one can clearly see how the outside-pre-constraint part $ \iota_{n} $ of $ \alpha_{n} $ is discarded and replaced by a combination of null space parts $ \lambda_{n+1} $ in the image $ \alpha_{n+1} $, in a parallel to the classical evolution.

It is time to consider the evolution of momentum eigenstates. We suggest that they should be evolved as
\begin{equation}\label{Un+1m}
\mathbb{U}_{n+1} \vert \beta_{n} \rangle_{\dot{\mathtt{m}}} = \vert P_{r_{n}} \beta_{n} \rangle_{\ddot{\mathtt{m}}} ~ g(\beta_{n})
\end{equation}
with a complex integrable function $ g(\beta_{n}) $ of the vector $ \beta_{n} $. We take advantage of the notation
\begin{equation}\label{key}
\beta_{n} \equiv \begin{pmatrix}
\theta_{n}\\
\nu_{n}
\end{pmatrix}, \qquad \beta_{n+1} \equiv P_{r_{n}} \beta_{n} = \begin{pmatrix}
\theta_{n}\\
0
\end{pmatrix}
\end{equation}
That $ \beta_{n+1} $ should have zero singular part is of course given by the classical evolution \eqref{dotuddotu}, which sets the irregular momentum components $ \ddot{u}_{n+1M+q} $ to zero as a realization of the post-constraint. On the other hand, there are no free parameters in the momentum sector and thus no additional integration. Note that while the singular part $ \nu_{n} $ of the vector $ \beta_{n} $ is annihilated when entering the eigenstate $ \vert \beta_{n+1} \rangle_{\ddot{\mathtt{m}}} $, we allow its presence in the normalization $ g(\beta_{n}) = g(\theta_{n}, \nu_{n})  $. As it will turn out in a moment, this is indeed necessary.

Eventually, let us see whether the suggested definitions \eqref{Un+1c} and \eqref{Un+1m} are compatible and whether they fix the so far unspecified functions $ f,g $. From \eqref{Un+1c} we compute
\begin{equation}\label{key}
\begin{aligned}
\mathbb{U}_{n+1} \vert \beta_{n} \rangle_{\dot{\mathtt{m}}} &= \mathbb{U}_{n+1} \sum_{\alpha_{n}} \vert \alpha_{n} \rangle_{\dot{\mathtt{c}}} ~_{\dot{\mathtt{c}}} \langle \alpha_{n} \vert \beta_{n} \rangle_{\dot{\mathtt{m}}} = \\
&= \sum_{\alpha_{n}} \sum_{\lambda_{n+1}} \mathcal{l}~ \vert \alpha_{n+1} \rangle_{\ddot{\mathtt{c}}} ~ \delta^{s_{n}}( \iota_{n}) ~ f(\alpha_{n+1}) ~_{\dot{\mathtt{c}}} \langle \alpha_{n} \vert \beta_{n} \rangle_{\dot{\mathtt{m}}} =\\
&= \sum_{\alpha_{n+1}} \mathcal{l}~ \vert \alpha_{n+1} \rangle_{\ddot{\mathtt{c}}} ~ f(\alpha_{n+1}) ~_{\dot{\mathtt{c}}} \langle P_{r_{n}} \alpha_{n+1} \vert \beta_{n} \rangle_{\dot{\mathtt{m}}}
\end{aligned}
\end{equation}
Now recall \eqref{dotcm} and \eqref{ddotcm}. From these one can (rather arbitrarily but correctly) express
\begin{equation}\label{key}
_{\dot{\mathtt{c}}} \langle P_{r_{n}} \alpha_{n+1} \vert \beta_{n} \rangle_{\dot{\mathtt{m}}} = ~_{\ddot{\mathtt{c}}} \langle \alpha_{n+1}  \vert P_{r_{n}} \beta_{n} \rangle_{\ddot{\mathtt{m}}} ~ e^{-i\frac{1}{2} \beta_{n}^{T} U^{T} L_{n} U \beta_{n}} ~ e^{-i\frac{1}{2} \alpha_{n+1}^{T} \bar{\Sigma}^{-1} V^{T} \bar{R}_{n+1} V \bar{\Sigma}^{-1} \alpha_{n+1}}
\end{equation}
Then one gets
\begin{equation}\label{key}
\begin{aligned}
\mathbb{U}_{n+1} \vert \beta_{n} \rangle_{\dot{\mathtt{m}}} &= \sum_{\alpha_{n+1}} \mathcal{l}~ \vert \alpha_{n+1} \rangle_{\ddot{\mathtt{c}}} ~ f(\alpha_{n+1}) ~_{\ddot{\mathtt{c}}} \langle \alpha_{n+1}  \vert P_{r_{n}} \beta_{n} \rangle_{\ddot{\mathtt{m}}} \\ & \qquad \qquad ~ e^{-i\frac{1}{2} \beta_{n}^{T} U^{T} L_{n} U \beta_{n}} ~ e^{-i\frac{1}{2} \alpha_{n+1}^{T} \bar{\Sigma}^{-1} V^{T} \bar{R}_{n+1} V \bar{\Sigma}^{-1} \alpha_{n+1}}
\end{aligned}
\end{equation}
For this moment, let us imagine that $ \mathcal{l} = const. $. Then we could take it out of the integral and obtain the desired result. By comparison with \eqref{Un+1m}, we would find
\begin{equation}\label{key}
f(\alpha_{n+1}) = e^{i D} ~ e^{i\frac{1}{2} \alpha_{n+1}^{T} \bar{\Sigma}^{-1} V^{T} \bar{R}_{n+1} V \bar{\Sigma}^{-1} \alpha_{n+1}}
\end{equation}
\begin{equation}\label{key}
g(\beta_{n}) = e^{i D} ~ e^{-i\frac{1}{2} \beta_{n}^{T} U^{T} L_{n} U \beta_{n}}
\end{equation}
We will later need to find a slightly more sophisticated prescription for $ \mathcal{l} $, but the reader will see that it will not prevent us from keeping this exact choice of $ f,g $.

The complex constant $ D $ remains unspecified. To fix it, we shall define the mapping $ \overline{\mathbb{U}_{n+1}} : \Phi_{n}^{\overline{\times}} \rightarrow \Phi_{n+1}^{\overline{\times}}  $ whose action on a left state $ \langle \psi_{n} \vert $ shall be denoted by $ \langle \psi_{n} \vert \overline{\mathbb{U}_{n+1}} $ and which is defined by denoting $ \vert \psi_{n+1} \rangle \equiv \mathbb{U}_{n+1} \vert \psi_{n} \rangle $ and putting $  \langle \psi_{n} \vert \overline{\mathbb{U}_{n+1}} = \langle \psi_{n+1} \vert $. In other words, $ \langle \psi_{n} \vert \overline{\mathbb{U}_{n+1}} = \overline{\mathbb{U}_{n+1} \vert \psi_{n} \rangle}  $ where the overline stands for the duality between $ \Phi_{n}^{\times} $ and $ \Phi_{n}^{\overline{\times}} $.

Now consider the wavefunction $ ~_{\dot{\mathtt{c}}} \langle \gamma_{n} \vert \overline{\mathbb{U}_{n+1}} \mathbb{U}_{n+1} \vert \alpha_{n} \rangle_{\dot{\mathtt{c}}} $ between two evolved states of the adapted basis and compare it to their original wavefunction
\begin{equation}\label{key}
~_{\dot{\mathtt{c}}} \langle \gamma_{n} \vert \alpha_{n} \rangle_{\dot{\mathtt{c}}} = \delta(\gamma_{n} - \alpha_{n})
\end{equation}
If we want the normalization of states to be conserved throughout the evolution, these wavefunctions should be equal. Let us therefore denote
\begin{equation}\label{key}
\gamma_{n} \equiv \begin{pmatrix}
\tau_{n} \\
\varepsilon_{n}
\end{pmatrix}, \qquad \gamma_{n+1} \equiv \begin{pmatrix}
\tau_{n}\\
\eta_{n+1}
\end{pmatrix}
\end{equation}
and compute
\begin{equation}\label{UoverlineU}
\begin{aligned}
~_{\dot{\mathtt{c}}} \langle \gamma_{n} \vert \overline{\mathbb{U}_{n+1}} \mathbb{U}_{n+1} \vert \alpha_{n} \rangle_{\dot{\mathtt{c}}}
&= \sum_{\lambda_{n+1}} \sum_{\eta_{n+1}} \overline{\mathcal{l}(\eta_{n+1})}  \mathcal{l}(\lambda_{n+1})  ~_{\ddot{\mathtt{c}}} \langle \gamma_{n+1} \vert \alpha_{n+1} \rangle_{\ddot{\mathtt{c}}} \\ & \qquad \qquad ~ \delta^{s_{n}}( \varepsilon_{n}) ~ \delta^{s_{n}}( \iota_{n}) ~ \overline{f(\gamma_{n+1})} ~ f(\alpha_{n+1}) =\\
&= \sum_{\lambda_{n+1}} \sum_{\eta_{n+1}} \overline{\mathcal{l}(\eta_{n+1})}\mathcal{l}(\lambda_{n+1}) ~ \delta^{q}(\gamma_{n+1} - \alpha_{n+1})  \\ & \qquad \qquad ~ \delta^{s_{n}}( \varepsilon_{n}) ~ \delta^{s_{n}}( \iota_{n}) ~ e^{-2 \Im D} =\\
&= \sum_{\lambda_{n+1}} \overline{\mathcal{l}(\lambda_{n+1})}\mathcal{l}(\lambda_{n+1}) ~ \delta^{r_{n}}(\tau_{n} - \varrho_{n})  \\ & \qquad \qquad ~ \delta^{s_{n}}( \varepsilon_{n} - \iota_{n}) ~ \delta^{s_{n}}( \iota_{n}) ~ e^{-2 \Im D} =\\
&= \delta(\gamma_{n} - \alpha_{n}) ~ \delta^{s_{n}}( \iota_{n}) ~ e^{-2 \Im D} \sum_{\lambda_{n+1}} \overline{\mathcal{l}(\lambda_{n+1})}\mathcal{l}(\lambda_{n+1})
\end{aligned}
\end{equation}
Here we see that up to the delta function $ \delta^{s_{n}}( \iota_{n}) $ and the expression $ \sum_{\lambda_{n+1}}\overline{\mathcal{l}}\mathcal{l} $, which are residuals of the irregular nature of the evolution, the evolved states are normalized standardly. Thus we decide to put $ D = 0 $. We will look closer at the issues of normalization in the next section.

The last goal of this section is to express the evolution mapping in the canonical coordinate bases. Upon employing \eqref{Un+1c}, we obtain
\begin{equation}\label{key}
\begin{aligned}
\mathbb{U}_{n+1} \vert \gamma_{n} \rangle_{\mathtt{c}} &= \mathbb{U}_{n+1} \sum_{\alpha_{n}} \vert \alpha_{n} \rangle_{\dot{\mathtt{c}}} ~_{\dot{\mathtt{c}}} \langle \alpha_{n} \vert \gamma_{n} \rangle_{\mathtt{c}} = \\
&= \sum_{\alpha_{n}} \sum_{\lambda_{n+1}} \mathcal{l} ~ \vert \alpha_{n+1} \rangle_{\ddot{\mathtt{c}}} ~ \delta^{s_{n}}(\iota_{n}) ~ e^{i\frac{1}{2} \alpha_{n+1}^{T} \bar{\Sigma}^{-1} V^{T} \bar{R}_{n+1} V \bar{\Sigma}^{-1} \alpha_{n+1}} ~_{\dot{\mathtt{c}}} \langle \alpha_{n} \vert \gamma_{n} \rangle_{\mathtt{c}} = \\
&= \sum_{\alpha_{n+1}} \mathcal{l} ~ \vert \alpha_{n+1} \rangle_{\ddot{\mathtt{c}}} ~ e^{i\frac{1}{2} \alpha_{n+1}^{T} \bar{\Sigma}^{-1} V^{T} \bar{R}_{n+1} V \bar{\Sigma}^{-1} \alpha_{n+1}} ~_{\dot{\mathtt{c}}} \langle P_{r_{n}} \alpha_{n+1} \vert \gamma_{n} \rangle_{\mathtt{c}}
\end{aligned}
\end{equation}
From \eqref{alphandot2} we know that
\begin{equation}\label{key}
~_{\dot{\mathtt{c}}} \langle P_{r_{n}} \alpha_{n+1} \vert \gamma_{n} \rangle_{\mathtt{c}} = (2\pi)^{-q/2} ~ e^{-i \frac{1}{2} \gamma_{n}^{T} L_{n} \gamma_{n} } ~ e^{-i \gamma_{n}^{T} U P_{r_{n}} \alpha_{n+1} }
\end{equation}
and we can also directly plug in for $ \vert \alpha_{n+1} \rangle_{\ddot{\mathtt{c}}} $ from \eqref{ddotc}. Altogether, it is
\begin{equation}\label{Un+1fincanon}
\begin{aligned}
\mathbb{U}_{n+1} \vert \gamma_{n} \rangle_{\mathtt{c}} &= \sum_{\alpha_{n+1}} \mathcal{l} ~ \vert V \bar{\Sigma}^{-1} \alpha_{n+1} \rangle_{\mathtt{c}} ~ \abs{\det \Sigma_{r} }^{-1/2} ~ (2\pi)^{-q/2} \\ & \qquad \qquad ~ e^{-i \frac{1}{2} \gamma_{n}^{T} L_{n} \gamma_{n} } ~ e^{-i \gamma_{n}^{T} U P_{r_{n}} \alpha_{n+1} } ~ e^{i\frac{1}{2} \alpha_{n+1}^{T} \bar{\Sigma}^{-1} V^{T} \bar{R}_{n+1} V \bar{\Sigma}^{-1} \alpha_{n+1}}
\end{aligned}
\end{equation}

The corresponding coordinate wavefunction is the following,
\begin{equation}\label{cUn+1c}
\begin{aligned}
_{\mathtt{c}} \langle \beta_{n+1} \vert \mathbb{U}_{n+1} \vert \gamma_{n} \rangle_{\mathtt{c}} &= \sum_{\alpha_{n+1}} \mathcal{l} ~ \delta \left( V \bar{\Sigma}^{-1} \alpha_{n+1} - \beta_{n+1} \right) ~ \abs{\det \Sigma_{r} }^{-1/2} ~ (2\pi)^{-q/2} \\ & \qquad \qquad ~ e^{-i \frac{1}{2} \gamma_{n}^{T} L_{n} \gamma_{n} } ~ e^{-i \gamma_{n}^{T} U P_{r_{n}} \alpha_{n+1} } ~ e^{i\frac{1}{2} \alpha_{n+1}^{T} \bar{\Sigma}^{-1} V^{T} \bar{R}_{n+1} V \bar{\Sigma}^{-1} \alpha_{n+1}} =\\
&= \mathcal{l}(V_{2}^{T} \beta_{n+1}) ~ \abs{\det \Sigma_{r} }^{1/2} ~ (2\pi)^{-q/2} \\ & \qquad \qquad  ~ e^{-i \frac{1}{2} \gamma_{n}^{T} L_{n} \gamma_{n} } ~ e^{-i \gamma_{n}^{T} R_{n} \beta_{n+1} } ~ e^{i\frac{1}{2} \beta_{n+1}^{T} \bar{R}_{n+1} \beta_{n+1}}
\end{aligned}
\end{equation}
where we restored $ U P_{r_{n}} \bar{\Sigma} V^{T} = U \Sigma_{r} V^{T} = R_{n} $ and computed $ \lambda_{n+1} = J_{s_{n}}^{T} \alpha_{n+1} = J_{s_{n}}^{T} \bar{\Sigma} V^{T} \beta_{n+1} = V_{2}^{T} \beta_{n+1} $. Comparing this result to \eqref{Slag}, which gives the one-step action in the form
\begin{equation}\label{timestepactionrepeated}
S_{n+1}(x_{n}, x_{n+1}) = \textstyle - \frac{1}{2} \left(  x_{n}^{T} L_{n} x_{n} + 2 x_{n}^{T} R_{n} x_{n+1} - x_{n+1}^{T} \bar{R}_{n+1} x_{n+1} \right)
\end{equation}
one immediately observes that
\begin{equation}\label{prop}
_{\mathtt{c}} \langle \beta_{n+1} \vert \mathbb{U}_{n+1} \vert \gamma_{n} \rangle_{\mathtt{c}} = \mathcal{l}(V_{2}^{T} \beta_{n+1}) ~ \abs{\det \Sigma_{r} }^{1/2} ~ (2\pi)^{-q/2} ~ e^{iS_{n+1}(\gamma_{n},\beta_{n+1})}
\end{equation}
i.e., up to the term $ \mathcal{l}(V_{2}^{T} \beta_{n+1}) $, we obtained the standard propagator proportional to the complex exponential of the classical action. This is satisfying as well as remarkable, since the procedure taken in defining our one-step evolution scheme was far from standard. We had to tackle the problem of irregularity of the Hamiltonian evolution, which we did by a suitable choice of the adapted coordinates. We did not use any Hamiltonian operator, instead we decided to rely on the classical evolution and find its suitable quantum analogy. We resolved the problem of non-existence of solutions by throwing away all states that do not satisfy the pre-constraint. Then we resolved the problem of non-uniqueness by evolving each quantum state into a linear combination of its possible images on the next time-slice. This approach resulted in the formula \eqref{Un+1c} for the quantum one-step evolution mapping, or \eqref{Un+1fincanon} in terms of the canonical coordinate bases. Eventually we learned that it is exactly this map which gives the (almost) standard propagator \eqref{prop}.

Eventually, let us comment on the result \eqref{prop} in context of the approach adopted in \cite{Hoehn2014, Hoehn2014a}. In these works, the (unfixed) propagator between time-slices 0 and 1 is assumed to be of the form
\begin{equation}\label{key}
K_{0\rightarrow 1} (x_{0}, x_{1}) = M_{0\rightarrow 1} ~ e^{i S_{1}(x_{0}, x_{1})/\hbar}
\end{equation}
where $ M_{0\rightarrow 1} $ is a propagator measure; see section 3.2 of \cite{Hoehn2014} and section 10.2 of \cite{Hoehn2014a}. The assumption is made on the basis of a similar result
\begin{equation}\label{Kcont}
K_{0 \rightarrow 1}(x_{0}, x_{1}) = \sqrt{\left(\frac{i}{\pi \hbar}\right)^{N} \det(\frac{\partial^{2} S(x_{0}, x_{1})}{\partial x_{0} \partial x_{1}})} ~ e^{i S(x_{0}, x_{1})/\hbar}
\end{equation}
known from the continuum \cite{Miller1974, Miller1970} which is exact for quadratic actions; see section 6.1.3 of \cite{Klauder2011}. In \cite{Hoehn2014}, the measure $ M_{0\rightarrow 1} $ is subsequently fixed via the invertibility (or to say, unitarity) conditions, yielding
\begin{equation}\label{K01}
K_{0\rightarrow 1} (x_{0}, x_{1}) = \sqrt{(-2\pi i \hbar)^{-N^{0}_{A}} ~ \abs{\det(T_{0}^{-1}) \det(c^{1}_{AB})\det(T_{1}^{-1})}} ~ e^{i S_{1}(x_{0}, x_{1})/\hbar}
\end{equation}
The details of the computation can be found in section 10.3 of \cite{Hoehn2014}. Here, $ N^{0}_{A} $ is the number of degrees of freedom propagating from time-slice 0 to time-slice 1, which corresponds to our $ r_{0} $. The $ q \times q $ regular matrix $ T_{n} $ is defined implicitly (and not uniquely) by means of a classification of vectors on the configuration space; see section 6 of \cite{Hoehn2014a}. The matrix elements $ c^{1}_{AB} $ are defined as $ c^{1}_{AB} = (T_{0})_{A}^{~~j} c_{ji}^{1} (T_{1}^{T})^{i}_{~B} $ (mind that $ A,B $ index columns of $ T_{n} $, contrary to the usual convention) where $ c_{ji}^{1} $ is our $ - (R_{0})_{ji} $. They form an $ N^{0}_{A} \times N^{0}_{A} $ regular matrix. At first sight, the determinants present in \eqref{prop} and \eqref{K01} do not seem to be equivalent. Recall that our $  \Sigma_{r} $ is the diagonal matrix of nonzero singular values of $ R_{n} $, which makes $ \det(\Sigma_{r}) $ simply the product of nonzero singular values of $ R_{n} $. Meanwhile, $ T_{n} $ is constructed from $ c^{n}, c^{n+1} $ and the Hessian $ h^{n ~ n+1} = b^{n} + a^{n+1} $ (these correspond to our $ L_{n} \equiv - a^{n+1} $, $ R_{n} \equiv - c^{n+1} $ and $ \bar{R}_{n+1} \equiv b^{n} $). Nevertheless, considering that $ T_{n} $ is not unique, it is reasonable to assume that the product of determinants in \eqref{K01} should actually simplify so that it would not be dependent on the particular choice of $ T_{n} $. Chances are that after this simplification, it might turn out equivalent to our $ \det(\Sigma_{r}) $. After all, the latter factor very well corresponds to \eqref{Kcont} which suggests that the propagator should be proportional to
\begin{equation}\label{Lag2formmatrix}
\det(\frac{\partial^{2} S_{n+1}(x_{0}, x_{1})}{\partial x_{n} \partial x_{n+1}})^{1/2}
\end{equation}
where we recognize the canonical coordinate matrix associated to (minus) the Lagrange two-form
\begin{equation}\label{Lagra}
\Omega_{n+1} = - \frac{\partial^{2} S_{n+1}(x_{n}, x_{n+1})}{\partial x_{n A} \partial x_{n+1 B}} dx_{n A} \wedge dx_{n+1 B}
\end{equation}
known from discrete canonical formalism \cite{Dittrich2013}. In our case, the matrix associated to \eqref{Lagra} is nothing but $ R_{n} $, and \eqref{Lag2formmatrix} is therefore equivalent to $ i^{q} \det(\Sigma)^{1/2} $. Our result \eqref{prop} uses only the regular part $ \Sigma_{r} $ of $ \Sigma $ to compute the determinant, which is appropriate---the propagator would otherwise completely vanish for degenerate actions.

The rest of the expressions in \eqref{prop} and \eqref{K01} can be matched more easily. The different powers of their numerical prefactors are a logical consequence of the fact that in \eqref{K01} the phase space has been effectively reduced to $ N^{0}_{A} $ propagating degrees of freedom---thus the factor $ (2\pi)^{-N^{0}_{A}} $---, while \eqref{prop} does not assume any reduction of the phase space---thus the full factor $ (2\pi)^{-q/2} $. Let us note that we use the unit convention in which $ \hbar = 1 $. Any extra $ \pm i $ amounts to a complex phase. The last notable difference is of course our regularization term $ \mathcal{l}(V_{2}^{T} \beta_{n+1}) $ which serves as a technical support to the evolution scheme. The advantage of \eqref{K01} is that it needs no such regularization.

\subsection{The Non-Unitarity}
Let us recall the standard physical interpretation of wavefunction: its absolute value squared should equal the probability of finding the system in the state it describes. Let us define our probability measure on $ m $-dimensional locally integrable complex functions by the $ m $-fold Lebesque integral
\begin{equation}\label{key}
\sum_{\zeta} \equiv \int_{-\infty}^{\infty} d\zeta_{1} ... d\zeta_{m}
\end{equation}
Assume that the functional $ \vert \varphi_{n} \rangle \in \Phi^{\times}_{n} $ is an initial state which corresponds to some wavefunction $ \varphi_{n} \in \mathcal{F}_{n} $ via
\begin{equation}\label{functionfunctional}
\langle \psi_{n} \vert \varphi_{n} \rangle = \int_{\mathbb{R}^{q}} \overline{\psi_{n}(x)} \varphi_{n}(x) ~ dx
\end{equation}
for all test functions $ \psi_{n} \in \Phi_{n} $, where the expression $ \langle \psi_{n} \vert \varphi_{n} \rangle $ stands for the action of the functional $ \vert \varphi_{n} \rangle $ on the test function $ \psi_{n} $. Then $ _{\mathtt{c}} \langle \alpha_{n} \vert \varphi_{n} \rangle = \varphi_{n}(\alpha_{n}) $ is the value of the wavefunction $ \varphi_{n} $ at a point $ \alpha_{n} \in \mathcal{Q}_{n} $ and 
\begin{equation}\label{key}
\sum_{\alpha_{n}} ~ \langle \varphi_{n} \vert \alpha_{n} \rangle_{\mathtt{c}} ~_{\mathtt{c}} \langle \alpha_{n} \vert \varphi_{n} \rangle = \langle \varphi_{n} \vert \varphi_{n} \rangle \in \mathbb{R}
\end{equation}
is the (square) norm of $ \varphi_{n} $. Note that the expression $ \langle \varphi_{n} \vert \varphi_{n} \rangle $ should not be understood in the sense of \eqref{functionfunctional} (because $ \varphi_{n} $ need not belong to $ \Phi_{n} $) but rather simply as an inner product on $ \mathcal{F}_{n} $. According to the physical interpretation, if this norm is nonzero, then the expression
\begin{equation}\label{key}
\langle \varphi_{n} \vert \varphi_{n} \rangle^{-1} \langle \varphi_{n} \vert \overline{\mathbb{U}_{n+1}} \vert \gamma_{n+1} \rangle_{\ddot{\mathtt{c}}} ~_{\ddot{\mathtt{c}}} \langle \gamma_{n+1} \vert \mathbb{U}_{n+1} \vert \varphi_{n} \rangle
\end{equation}
should equal the probability of finding the system one time-step later in state $ \vert \gamma_{n+1} \rangle_{\ddot{\mathtt{c}}} $. Usually one expects the normalization to be conserved throughout the evolution, which would be the case if
\begin{equation}\label{varphilaternorm}
\langle \varphi_{n} \vert \varphi_{n} \rangle^{-1} \sum_{\gamma_{n+1}} ~ \langle \varphi_{n} \vert \overline{\mathbb{U}_{n+1}} \vert \gamma_{n+1} \rangle_{\ddot{\mathtt{c}}} ~_{\ddot{\mathtt{c}}} \langle \gamma_{n+1} \vert \mathbb{U}_{n+1} \vert \varphi_{n} \rangle = \langle \varphi_{n} \vert \varphi_{n} \rangle^{-1} \langle \varphi_{n} \vert \overline{\mathbb{U}_{n+1}} \mathbb{U}_{n+1} \vert \varphi_{n} \rangle
\end{equation}
was equal to one. However, it turns out this is not the case, because we have
\begin{equation}\label{varphiUoverlineUvarphi}
\begin{aligned}
\langle \varphi_{n} \vert \overline{\mathbb{U}_{n+1}} \mathbb{U}_{n+1} \vert \varphi_{n} \rangle &= \sum_{\alpha_{n}} \sum_{\gamma_{n}} \langle \varphi_{n} \vert \gamma_{n} \rangle_{\dot{\mathtt{c}}} ~_{\dot{\mathtt{c}}} \langle \gamma_{n} \vert \overline{\mathbb{U}_{n+1}} \mathbb{U}_{n+1} \vert \alpha_{n} \rangle_{\dot{\mathtt{c}}} ~_{\dot{\mathtt{c}}} \langle \alpha_{n} \vert \varphi_{n} \rangle = \\
&= \sum_{\alpha_{n}} \sum_{\gamma_{n}} \langle \varphi_{n} \vert \gamma_{n} \rangle_{\dot{\mathtt{c}}} ~_{\dot{\mathtt{c}}} \langle \alpha_{n} \vert \varphi_{n} \rangle ~ \delta(\gamma_{n} - \alpha_{n}) ~ \delta^{s_{n}}( \iota_{n}) \sum_{\lambda_{n+1}} \overline{\mathcal{l}}\mathcal{l} =\\
&= \sum_{\alpha_{n}} \langle \varphi_{n} \vert \alpha_{n} \rangle_{\dot{\mathtt{c}}} ~_{\dot{\mathtt{c}}} \langle \alpha_{n} \vert \varphi_{n} \rangle ~ \delta^{s_{n}}( \iota_{n}) \sum_{\lambda_{n+1}} \overline{\mathcal{l}}\mathcal{l} =\\
&= \sum_{\varrho_{n}} \langle \varphi_{n} \vert J_{r_{n}} \varrho_{n} \rangle_{\dot{\mathtt{c}}} ~_{\dot{\mathtt{c}}} \langle J_{r_{n}} \varrho_{n} \vert \varphi_{n} \rangle ~ \sum_{\lambda_{n+1}} \overline{\mathcal{l}}\mathcal{l}
\end{aligned}
\end{equation}
where we recalled \eqref{UoverlineU} and took advantage of the $ q \times r_{n} $ matrix
\begin{equation}\label{Jrn}
J_{r_{n}} = \begin{pmatrix}
\mathbf{1}_{r_{n}}\\
0
\end{pmatrix}
\end{equation}

Let us define the operator $ \mathbb{U}_{n+1}^{\ast} : \Phi_{n+1}^{\times} \rightarrow \Phi_{n}^{\times} $ (note that it evolves states backwards) by $ \langle \varphi_{n} \vert \mathbb{U}_{n+1}^{\ast} \vert \psi_{n+1} \rangle = \langle \varphi_{n} \vert \overline{\mathbb{U}_{n+1}} \vert \psi_{n+1} \rangle $ for all $ \vert \varphi_{n} \rangle \in \Phi_{n}^{\times} $ and $ \vert \psi_{n+1} \rangle \in \Phi_{n+1}^{\times} $. Then we can investigate the behavior of the operator product $ \mathbb{U}_{n+1}^{\ast} \mathbb{U}_{n+1} : \Phi_{n}^{\times} \rightarrow \Phi_{n}^{\times} $. Clearly, the above discussed requirement that the normalization is conserved is equivalent to the requirement that $ \mathbb{U}_{n+1}^{\ast} \mathbb{U}_{n+1} = \hat{\mathbf{1}} $, which is a kind of \textit{unitarity} of the one-step evolution mapping. The equation \eqref{varphiUoverlineUvarphi} tells us that this requirement is in general not satisfied since we formally obtain
\begin{equation}\label{UastU}
\mathbb{U}_{n+1}^{\ast} \mathbb{U}_{n+1} = \sum_{\varrho_{n}} ~ \vert J_{r_{n}} \varrho_{n} \rangle_{\dot{\mathtt{c}}} ~_{\dot{\mathtt{c}}} \langle J_{r_{n}} \varrho_{n} \vert ~ \sum_{\lambda_{n+1}} \overline{\mathcal{l}}\mathcal{l}
\end{equation}
This is the reason for us to say that the evolution is \textit{non-unitary}. Indeed, it must be this way: it is based on classical evolution which looses information and consumes new information along the way. It is simply an inconvenient feature of the present approach that, as it stands now, the standard physical interpretation is broken. If we want to draw results from the model, it is our task to repair it. First of all, we shall regularize the evolution map to avoid divergences arising from the integration over free parameters $ \lambda_{n+1} $. Second, if necessary, we will renormalize the initial states to recover the probabilistic interpretation of the theory. In the following two sections, we perform this series of tasks for two particular regularization schemes which we consider especially simple and useful.

\subsection{Functional Regularization}

The regularization which we want to propose is rather straightforward. The result \eqref{varphiUoverlineUvarphi} tells us that for $ \mathcal{l} = const. $, which would be our preferred choice, the above defined evolution scheme produces a divergence of the order $ \sum_{\lambda_{n+1}} $ in the normalization of evolved wavefunctions. We suggest to fix this by introducing the regularized time-step evolution map $ \mathbb{U}_{n+1}^{\varepsilon} : \Phi^{\times}_{n} \rightarrow \Phi^{\times}_{n+1} $ given by the same formula \eqref{Un+1bar} as $ \mathbb{U}_{n+1} $ and
\begin{equation}\label{el}
\mathcal{l} = k(\varepsilon) ~ e^{-\frac{1}{2} \varepsilon \lambda_{n+1}^{T}\lambda_{n+1} }
\end{equation}
with a regularization parameter $ \varepsilon > 0 $. We include a real normalization constant $ k(\varepsilon) $ which depends only on $ \varepsilon $. Upon employing the formula
\begin{equation}\label{key}
\int_{\mathbb{R}^{n}} e^{-\frac{1}{2} x^{T} A x + i J^{T} x} ~ d^{n}x = \sqrt{\frac{(2 \pi)^{n}}{\det A}} ~ e^{-\frac{1}{2} J^{T} A^{-1} J}
\end{equation}
(the matrix $ A $ must be positive definite here) which is a multidimensional version of the Hubbard–Stratonovich transformation \cite{Hubbard1959}, one finds
\begin{equation}\label{lalt}
\mathcal{l} = \sum_{\mu_{n+1}} k ~ (\varepsilon 2 \pi)^{-s_{n}/2} ~ e^{-\frac{1}{2} \varepsilon^{-1} \mu_{n+1}^{T} \mu_{n+1}} ~ e^{i \lambda_{n+1}^{T} \mu_{n+1}}
\end{equation}
With this, one can compute
\begin{equation}\label{key}
\begin{aligned}
\mathbb{U}_{n+1}^{\varepsilon} \vert \beta_{n} \rangle_{\dot{\mathtt{m}}} &= \sum_{\alpha_{n+1}} \mathcal{l} ~ \vert \alpha_{n+1} \rangle_{\ddot{\mathtt{c}}} ~ f(\alpha_{n+1}) ~_{\ddot{\mathtt{c}}} \langle \alpha_{n+1}  \vert P_{r_{n}} \beta_{n} \rangle_{\ddot{\mathtt{m}}} \\ & \qquad \qquad ~ e^{-i\frac{1}{2} \beta_{n}^{T} U^{T} L_{n} U \beta_{n}} ~ e^{-i\frac{1}{2} \alpha_{n+1}^{T} \bar{\Sigma}^{-1} V^{T} \bar{R}_{n+1} V \bar{\Sigma}^{-1} \alpha_{n+1}} =\\
&= \sum_{\alpha_{n+1}} \mathcal{l} ~ \vert \alpha_{n+1} \rangle_{\ddot{\mathtt{c}}} ~_{\ddot{\mathtt{c}}} \langle \alpha_{n+1}  \vert P_{r_{n}} \beta_{n} \rangle_{\ddot{\mathtt{m}}} ~ g(\beta_{n}) =\\
&= \sum_{\alpha_{n+1}} \sum_{\mu_{n+1}} k ~ (2 \pi \varepsilon)^{-s_{n}/2} e^{-\frac{1}{2} \varepsilon^{-1} \mu_{n+1}^{T} \mu_{n+1}} \vert \alpha_{n+1} \rangle_{\ddot{\mathtt{c}}} ~_{\ddot{\mathtt{c}}} \langle \alpha_{n+1}  \vert \beta_{n+1} \rangle_{\ddot{\mathtt{m}}} ~ g(\beta_{n}) =\\
&= \sum_{\mu_{n+1}} k ~ (2 \pi \varepsilon)^{-s_{n}/2} e^{-\frac{1}{2} \varepsilon^{-1} \mu_{n+1}^{T} \mu_{n+1}} \vert \beta_{n+1} \rangle_{\ddot{\mathtt{m}}} ~ g(\beta_{n})
\end{aligned}
\end{equation}
where we recalled \eqref{ddotcm} and denoted
\begin{equation}\label{key}
\beta_{n} \equiv \begin{pmatrix}
\theta_{n}\\
\nu_{n}
\end{pmatrix}, \qquad \beta_{n+1} \equiv \begin{pmatrix}
\theta_{n}\\
\mu_{n+1}
\end{pmatrix}
\end{equation}
This is nothing but normal distribution in terms of $ \mu_{n+1} $ centered around zero. We were able to obtain this simple result only thanks to the special form of \eqref{el}.

In summary, we have defined the evolution map $ \mathbb{U}_{n+1}^{\varepsilon} : \Phi^{\times}_{n} \rightarrow \Phi^{\times}_{n+1} $ acting as
\begin{equation}\label{Un+1creg}
\mathbb{U}_{n+1}^{\varepsilon} \vert \alpha_{n} \rangle_{\dot{\mathtt{c}}} = \sum_{\lambda_{n+1}} \mathcal{l}(\lambda_{n+1}) ~ \vert \alpha_{n+1} \rangle_{\ddot{\mathtt{c}}} ~ \delta^{s_{n}}( \iota_{n}) ~ f(\alpha_{n+1})
\end{equation}
\begin{equation}\label{key}
\alpha_{n} \equiv \begin{pmatrix}
\varrho_{n} \\
\iota_{n}
\end{pmatrix}, \qquad \alpha_{n+1} \equiv \begin{pmatrix}
\varrho_{n}\\
\lambda_{n+1}
\end{pmatrix}
\end{equation}
\begin{equation}\label{key}
f(\alpha_{n+1}) \equiv e^{i\frac{1}{2} \alpha_{n+1}^{T} \bar{\Sigma}^{-1} V^{T} \bar{R}_{n+1} V \bar{\Sigma}^{-1} \alpha_{n+1}}
\end{equation}
\begin{equation}\label{key}
\mathcal{l}(\lambda_{n+1}) \equiv k ~ e^{-\frac{1}{2} \varepsilon \lambda_{n+1}^{T}\lambda_{n+1} }
\end{equation}
or equivalently
\begin{equation}\label{Un+1mreg}
\mathbb{U}_{n+1}^{\varepsilon} \vert \beta_{n} \rangle_{\dot{\mathtt{m}}} = \sum_{\mu_{n+1}} \tilde{\mathcal{l}}(\mu_{n+1}) ~ \vert \beta_{n+1} \rangle_{\ddot{\mathtt{m}}} ~ g(\beta_{n})
\end{equation}
\begin{equation}\label{key}
\beta_{n} \equiv \begin{pmatrix}
\theta_{n}\\
\nu_{n}
\end{pmatrix}, \qquad \beta_{n+1} \equiv \begin{pmatrix}
\theta_{n}\\
\mu_{n+1}
\end{pmatrix}
\end{equation}
\begin{equation}\label{key}
g(\beta_{n}) \equiv e^{-i\frac{1}{2} \beta_{n}^{T} U^{T} L_{n} U \beta_{n}}
\end{equation}
\begin{equation}\label{key}
\tilde{\mathcal{l}}(\mu_{n+1}) \equiv k ~ (2 \pi \varepsilon)^{-s_{n}/2} e^{-\frac{1}{2} \varepsilon^{-1} \mu_{n+1}^{T} \mu_{n+1}}
\end{equation}

We can now turn to the normalization of evolved states. Since $ \mathcal{l} $ is square integrable, it is possible to choose such $ k $ that
\begin{equation}\label{knorm}
\sum_{\lambda_{n+1}}  \mathcal{l}^{2} = 1
\end{equation}
This requirement corresponds to
\begin{equation}\label{kepsilon}
k = \pi^{-s_{n}/4} ~ \varepsilon^{s_{n}/4}
\end{equation}
Indeed, when  $ \varepsilon \rightarrow 0 $, we shall have $ k \rightarrow 0 $ so that the assumed constant normalization \eqref{knorm} is sustained. We remark that while $ \mathcal{l}^{2} $ approaches a (zero) constant function in this limit, $ \tilde{\mathcal{l}}^{2} $ in fact approaches the delta function. It means that the singular part $ \mu_{n+1} $ of the evolved momentum state $ \vert \beta_{n+1} \rangle_{\ddot{\mathtt{m}}} $ in \eqref{Un+1mreg} ultimately gets localized around value $ \mu_{n+1} = 0 $, which clearly complies to the post-constraint. Note that $ \tilde{\mathcal{l}} $ does not diverge fast enough to render the integral nonzero, so upon performing the limit, one gets zero final state, just as in the case of the coordinate eigenstate. Nevertheless, for $ \varepsilon > 0 $, all the expressions are finite and the evolved states are nonzero. With \eqref{knorm} in place, \eqref{UoverlineU} and \eqref{varphiUoverlineUvarphi} become
\begin{equation}\label{UoverlineUreg}
~_{\dot{\mathtt{c}}} \langle \gamma_{n} \vert \overline{\mathbb{U}_{n+1}^{\varepsilon}} \mathbb{U}_{n+1}^{\varepsilon} \vert \alpha_{n} \rangle_{\dot{\mathtt{c}}}
= \delta(\gamma_{n} - \alpha_{n}) ~ \delta^{s_{n}}( \iota_{n})
\end{equation}
\begin{equation}\label{vphiUvepsisq}
\langle \varphi_{n} \vert \overline{\mathbb{U}_{n+1}^{\varepsilon}} \mathbb{U}_{n+1}^{\varepsilon} \vert \varphi_{n} \rangle
= \sum_{\varrho_{n}} \langle \varphi_{n} \vert J_{r_{n}} \varrho_{n} \rangle_{\dot{\mathtt{c}}} ~_{\dot{\mathtt{c}}} \langle J_{r_{n}} \varrho_{n} \vert \varphi_{n} \rangle
\end{equation}

The possible renormalization of probability offered in \eqref{vphiUvepsisq} is now obvious. For any initial state $ \vert \varphi_{n} \rangle \in \Phi^{\times}_{n} $ corresponding to a normalizable wavefunction $ \varphi_{n} \in \mathcal{F}_{n} $ as in \eqref{functionfunctional}, we introduce the constant
\begin{equation}\label{Cvarphi}
C_{\vert \varphi_{n} \rangle} \equiv \langle \varphi_{n} \vert \overline{\mathbb{U}_{n+1}^{\varepsilon}} \mathbb{U}_{n+1}^{\varepsilon} \vert \varphi_{n} \rangle
\end{equation}
which, if finite, can be interpreted as the square norm of $ \vert \varphi_{n} \rangle $ in the subspace of $ \Phi^{\times}_{n} $ corresponding to the pre-constraint surface. In particular, if $ C_{\vert \varphi_{n} \rangle} = 0 $, it means that $ \varphi_{n} $ is, up to domain subsets of zero measure, supported only outside the pre-constraint surface. In that case, we will not talk about probability of finding the system in an evolved state, just as we do not talk about such probability for an initial state $ \vert \varphi_{n} \rangle $ when $ \langle \varphi_{n} \vert \varphi_{n} \rangle = 0 $. On the other hand, if $ C_{\vert \varphi_{n} \rangle} $ is finite and nonzero, we will say that the probability of finding the system one time-step later in state $ \vert \psi_{n+1} \rangle $ is
\begin{equation}\label{prob1}
p_{\vert \varphi_{n} \rangle \rightarrow \vert \psi_{n+1} \rangle} =  C_{\vert \varphi_{n} \rangle}^{-1} ~  \lim_{\varepsilon \rightarrow 0} ~ \langle \varphi_{n} \vert \overline{\mathbb{U}_{n+1}^{\varepsilon}} \vert \psi_{n+1} \rangle \langle \psi_{n+1} \vert \mathbb{U}_{n+1}^{\varepsilon} \vert \varphi_{n} \rangle
\end{equation}
rather than $ \langle \varphi_{n} \vert \varphi_{n} \rangle^{-1} \langle \varphi_{n} \vert \overline{\mathbb{U}_{n+1}} \vert \psi_{n+1} \rangle \langle \psi_{n+1} \vert \mathbb{U}_{n+1} \vert \varphi_{n} \rangle $. This finally results in the correct normalization
\begin{equation}\label{key}
\sum_{\gamma_{n+1}} p_{\vert \varphi_{n} \rangle \rightarrow  \vert \gamma_{n+1} \rangle_{\ddot{\mathtt{c}}} } = 1
\end{equation}

Let us remark that one can also look at the problem of normalization differently. If we use the functional regularization for which \eqref{knorm} holds true, we know that the change of the norms of wavefunctions during the evolution is caused solely by the mapping $ \Xi_{n} $ introduced in \eqref{Chi}, which implements the pre-constraint. It does so in a rather radical way: the map effectively filters out every part of the wavefunction $ _{\mathtt{c}} \langle \alpha_{n} \vert \varphi_{n} \rangle $ laying at points $ \alpha_{n} $ which do not satisfy the pre-constraint. In effect, the norm decreases or at most stays the same (if the wavefunction is supported only inside the pre-constraint surface). Then, if one has a system whose irregularity bears some kind of physical meaning, the decreasing norm can be viewed less as a problem and more as a reflection of the physical reality. For instance, imagine a system describing quantum field on a discrete model of spacetime, which at some point splits into two disjoint regions. One of the regions continues without any irregularity, the other diminishes and ends. We fix the quantum field by initial data on some spacelike hypersurface before the splitting and let it evolve. For a family of observers living in the regular region, the field which entered the other region is forever lost. It thus makes perfect sense if the norm of the corresponding state in the regular region decreases, reflecting on that loss. This example illustrates that in some cases, we actually might want to keep the normalization as it is. The reduced norm of the resulting final state, and even the corresponding transition probability, will then simply reflect the dissipative nature of our system.

\subsection{Cylindrical Regularization}
An alternative and very useful approach to regularization is adopted in \cite{Hoehn2014, Hoehn2014a}. We dedicate this section to its short introduction in context of our work. The main point is that instead of regularizing the wavefunction, one can actually leave it in its natural form and instead change the inner product so that the norm exists and is finite. It is argued that the wavefunctions describing states of an irregular quantum system can be seen as \textit{cylindrical functions}---see section 5.2 of \cite{Hoehn2014a}---which are constant w.r.t. some parameters from a (linear) subspace of their domain. These of course do not belong to the original Hilbert space so there is no norm for them. However, the inner product for cylindrical functions can be defined with the help of so-called \textit{cylindrical measure}.

Within our formalism, the procedure would be the following. We consider the evolution map $ \mathbb{U}_{n+1}^{0} $ still given by \eqref{Un+1bar} in which we lay
\begin{equation}\label{el0}
\mathcal{l} = 1
\end{equation}
In consequence, the evolved state is simply
\begin{equation}\label{U0}
\mathbb{U}_{n+1}^{0} \vert \alpha_{n} \rangle_{\dot{\mathtt{c}}} = \sum_{\lambda_{n+1}} ~ \vert \alpha_{n+1} \rangle_{\ddot{\mathtt{c}}} ~ \delta^{s_{n}}( \iota_{n}) ~ f(\alpha_{n+1})
\end{equation}
with
\begin{equation}\label{key}
\alpha_{n} \equiv \begin{pmatrix}
\varrho_{n} \\
\iota_{n}
\end{pmatrix}, \qquad \alpha_{n+1} \equiv \begin{pmatrix}
\varrho_{n}\\
\lambda_{n+1}
\end{pmatrix}
\end{equation}
\begin{equation}\label{key}
f(\alpha_{n+1}) \equiv e^{i\frac{1}{2} \alpha_{n+1}^{T} \bar{\Sigma}^{-1} V^{T} \bar{R}_{n+1} V \bar{\Sigma}^{-1} \alpha_{n+1}}
\end{equation}
just as before. Note that \eqref{U0} still belongs to $ \Phi^{\times}_{n+1} $ so the definition is formally correct. However, for an initial $ \varphi_{n} \in \mathcal{F}_{n} $, it results in the state
\begin{equation}\label{U0varphi}
\begin{aligned}
\mathbb{U}_{n+1}^{0} \vert \varphi_{n} \rangle &= \sum_{\alpha_{n}} \mathbb{U}_{n+1}^{0} \vert \alpha_{n} \rangle_{\dot{\mathtt{c}}} ~_{\dot{\mathtt{c}}} \langle \alpha_{n} \vert \varphi_{n} \rangle =\\
&= \sum_{\alpha_{n}} \sum_{\lambda_{n+1}} ~ \vert \alpha_{n+1} \rangle_{\ddot{\mathtt{c}}} ~ \delta^{s_{n}}( \iota_{n}) ~ f(\alpha_{n+1}) ~_{\dot{\mathtt{c}}} \langle \alpha_{n} \vert \varphi_{n} \rangle =\\
&= \sum_{\alpha_{n+1}} ~ \vert \alpha_{n+1} \rangle_{\ddot{\mathtt{c}}} ~ f(\alpha_{n+1}) ~_{\dot{\mathtt{c}}} \langle P_{r_{n}} \alpha_{n+1} \vert \varphi_{n} \rangle
\end{aligned}
\end{equation}
with $ P_{r_{n}} $ as in \eqref{Chi}. We see that the dependence of the integrand on $ \lambda_{n+1} $ is limited to the factor $ f(\alpha_{n+1}) $ and the eigenstate $ \vert \alpha_{n+1} \rangle_{\ddot{\mathtt{c}}} $. As we know, this is certainly not enough to render $ \mathbb{U}_{n+1}^{0} \vert \varphi_{n} \rangle $ normalizable. The corresponding wavefunction amounts to
\begin{equation}\label{gammaU0varphi}
\begin{aligned}
~_{\ddot{\mathtt{c}}} \langle \gamma_{n+1} \vert \mathbb{U}_{n+1}^{0} \vert \varphi_{n} \rangle
&= \sum_{\alpha_{n+1}} ~_{\ddot{\mathtt{c}}} \langle \gamma_{n+1} \vert  \alpha_{n+1} \rangle_{\ddot{\mathtt{c}}} ~ f(\alpha_{n+1}) ~_{\dot{\mathtt{c}}} \langle P_{r_{n}} \alpha_{n+1} \vert \varphi_{n} \rangle = \\
&= f(\gamma_{n+1}) ~_{\dot{\mathtt{c}}} \langle P_{r_{n}} \gamma_{n+1} \vert \varphi_{n} \rangle
\end{aligned}
\end{equation}
Indeed, \eqref{gammaU0varphi} clearly does not belong to $ \mathcal{F}_{n+1} $. This is manifested by an infinite expression for what would otherwise be the norm,
\begin{equation}\label{nonorm}
\begin{aligned}
\langle \varphi_{n} \vert \overline{\mathbb{U}^{0}_{n+1}} \mathbb{U}^{0}_{n+1} \vert \varphi_{n} \rangle = \sum_{\varrho_{n}} \langle \varphi_{n} \vert J_{r_{n}} \varrho_{n} \rangle_{\dot{\mathtt{c}}} ~_{\dot{\mathtt{c}}} \langle J_{r_{n}} \varrho_{n} \vert \varphi_{n} \rangle ~ \sum_{\lambda_{n+1}} = \infty
\end{aligned}
\end{equation}
cf. \eqref{varphiUoverlineUvarphi}. We therefore need to introduce an alternative inner product for wavefunctions of the form \eqref{gammaU0varphi}. This can be done relatively easily. The configuration space splits as $ \mathcal{Q}_{n+1} = \Lambda^{\mathcal{Q}}_{n+1} \oplus \dot{\mathcal{C}}^{+\mathcal{Q}}_{n+1} $ into the space $ \Lambda^{\mathcal{Q}}_{n+1} = \{ V_{2} \lambda_{n+1} ~ \vert ~ \lambda_{n+1} \in \mathbb{R}^{s_{n}} \} $ which consists of coordinate parts of vectors in $ \Lambda_{n+1} $ (see section 2.1) and the space $ \dot{\mathcal{C}}^{+\mathcal{Q}}_{n+1} = \{ V_{1} \varrho_{n+1} ~ \vert ~ \varrho_{n+1} \in \mathbb{R}^{r_{n}} \} $ which consists of coordinate parts of vectors in the symplectic space $ \dot{\mathcal{C}}^{+}_{n+1} \equiv \{ E_{n} y_{n} ~ \vert ~ y_{n} \in \mathcal{C}^{-}_{n} \} $ familiar from section 5.3 of \cite{Kaninsky2020a}. We can then introduce a new \textit{physical} Hilbert space $ \mathcal{F}_{n+1}^{\mathrm{ph}} = L^{2}(\dot{\mathcal{C}}^{+\mathcal{Q}}_{n+1}) $ with the usual inner product denoted by $ \langle ~ \rangle^{\mathrm{ph}} $. Recall that the same trick is used in \cite{Hoehn2014, Hoehn2014a}---as we have already discussed before, the reduction of kinematical Hilbert spaces to physical ones is a key point of these works. For us, the situation is very much simplified thanks to our choice of basis: the adapted coordinates $ \gamma_{n+1} \in \mathbb{R}^{q} $ w.r.t. the double-dotted symplectic basis $ \{ \ddot{e}_{n+1 I} \}_{I = 1}^{q} $ of $ \mathcal{P}_{n+1} $ split trivially into
\begin{equation}\label{gammataueta}
\gamma_{n+1} \equiv \begin{pmatrix}
\tau_{n+1}\\
\eta_{n+1}
\end{pmatrix}
\end{equation}
where $ \tau_{n+1} \in \mathbb{R}^{r_{n}} $ and $ \eta_{n+1} \in \mathbb{R}^{s_{n}} $. The former are coordinates on $ \dot{\mathcal{C}}^{+\mathcal{Q}}_{n+1} $, the latter are coordinates on $ \Lambda^{\mathcal{Q}}_{n+1} $. The physical inner product of two functions $ \psi_{n+1}, \phi_{n+1} \in \mathcal{F}_{n+1}^{\mathrm{ph}} $ which correspond to the functions $ \psi_{n+1}, \phi_{n+1} \in \mathcal{F}_{n+1} $ can be written as
\begin{equation}\label{1spip}
\langle \psi_{n+1} \vert \phi_{n+1} \rangle^{\mathrm{ph}} = \sum_{\tau_{n+1}} ~ \langle \psi_{n+1} \vert \gamma_{n+1} \rangle_{\ddot{\mathtt{c}}} ~_{\ddot{\mathtt{c}}} \langle \gamma_{n+1} \vert \phi_{n+1} \rangle
\end{equation}
where the coordinates $ \eta_{n+1} $ play the role of additional parameters. The functions are expressed in the adapted coordinates here. If we have them given in terms of the canonical coordinates (which is typically the case), we need to transform them with
\begin{equation}\label{ddotcgammaphi}
~_{\ddot{\mathtt{c}}} \langle \gamma_{n+1} \vert \phi_{n+1} \rangle = \abs{\det \Sigma_{r} }^{-1/2} ~_{\mathtt{c}}\langle V \bar{\Sigma}^{-1} \gamma_{n+1} \vert \phi_{n+1} \rangle
\end{equation}
where we used \eqref{ddotc}. Then we can plug \eqref{ddotcgammaphi} into \eqref{1spip}, getting
\begin{equation}\label{pipcanonical}
\langle \psi_{n+1} \vert \phi_{n+1} \rangle^{\mathrm{ph}} = \abs{\det \Sigma_{r} }^{-1} \sum_{\tau_{n+1}} ~ \langle \psi_{n+1} \vert V \bar{\Sigma}^{-1} \gamma_{n+1} \rangle_{\mathtt{c}} ~_{\mathtt{c}}\langle V \bar{\Sigma}^{-1} \gamma_{n+1} \vert \phi_{n+1} \rangle
\end{equation}
The same can be done for the final state wavefunction. Again, we shall view the $ \eta_{n+1} $ as parameters of \eqref{gammaU0varphi} rather than its variables. Then we can easily convince ourselves that while \eqref{gammaU0varphi} does not belong to $ \mathcal{F}_{n+1} $, it actually belongs to $ \mathcal{F}_{n+1}^{\mathrm{ph}} $ for any choice of these parameters $ \eta_{n+1} $, because the expression
\begin{equation}\label{phnorm}
\langle \varphi_{n} \vert \overline{\mathbb{U}^{0}_{n+1}} \mathbb{U}^{0}_{n+1} \vert \varphi_{n} \rangle^{\mathrm{ph}} = \sum_{\varrho_{n}} \langle \varphi_{n} \vert J_{r_{n}} \varrho_{n} \rangle_{\dot{\mathtt{c}}} ~_{\dot{\mathtt{c}}} \langle J_{r_{n}} \varrho_{n} \vert \varphi_{n} \rangle
\end{equation}
is---unlike \eqref{nonorm}---finite. This is of course the norm of $ ~_{\ddot{\mathtt{c}}} \langle \gamma_{n+1} \vert \mathbb{U}_{n+1}^{0} \vert \varphi_{n} \rangle $ on $ \mathcal{F}_{n+1}^{\mathrm{ph}} $. Note that it does not depend on the free parameters at all.

The regularization of the one-step evolution map is now complete. We can of course define
\begin{equation}\label{Cphnorm}
C_{\vert \varphi_{n} \rangle} \equiv \langle \varphi_{n} \vert \overline{\mathbb{U}^{0}_{n+1}} \mathbb{U}^{0}_{n+1} \vert \varphi_{n} \rangle^{\mathrm{ph}}
\end{equation}
which yields the exact same result as \eqref{Cvarphi}. Then we may perform renormalization in complete analogy with the previous section. The prescription for the renormalized transition probability is given similarly to \eqref{prob1}, only with use of the alternative inner product and with no need for a limiting procedure:
\begin{equation}\label{prob1}
p_{\vert \varphi_{n} \rangle \rightarrow \vert \psi_{n+1} \rangle} =  C_{\vert \varphi_{n} \rangle}^{-1} ~ \langle \varphi_{n} \vert \overline{\mathbb{U}_{n+1}^{0}} \vert \psi_{n+1}^{\mathrm{ph}} \rangle^{\mathrm{ph}} \langle \psi_{n+1}^{\mathrm{ph}} \vert \mathbb{U}_{n+1}^{0} \vert \varphi_{n} \rangle^{\mathrm{ph}}
\end{equation}
One has to plug in the state $ \vert \psi_{n+1}^{\mathrm{ph}} \rangle \in \Phi_{n+1}^{\mathrm{ph} \times} $ whose state space $ \Phi_{n+1}^{\mathrm{ph} \times} $ is built from $ \mathcal{F}_{n+1}^{\mathrm{ph}} $. The state $ \vert \psi_{n+1}^{\mathrm{ph}} \rangle $ has a wavefunction $ \psi_{n+1}^{\mathrm{ph}}(\tau_{n+1}) \in \mathcal{F}_{n+1}^{\mathrm{ph}} $. It corresponds to the wavefunction $ \psi_{n+1}(\gamma_{n+1}) \in \mathcal{F}_{n+1} $ which is constant w.r.t. the variables $ \eta_{n+1} $. This fixes a corresponding state $ \vert \psi_{n+1} \rangle \in \Phi_{n+1}^{\times} $. For one-step evolution, the parameters $ \eta_{n+1} $ can be actually considered unphysical, since the evolution map does not make any difference between them. To have a final state wavefunction constant in these variables is therefore natural. Eventually, the resolution of identity on the physical Hilbert space yields the correct normalization of probability, as before.

\subsection{Density Operators}

We can easily rewrite the theory into the language of density operators \cite{Blum2012}. The general form of a \textit{density operator} $ \hat{A}_{n} : \Phi^{\times}_{n} \rightarrow \Phi^{\times}_{n} $ is
\begin{equation}\label{key}
\hat{A}_{n} = \sum_{\Omega_{n}} A_{\vert \varphi_{n} \rangle} ~ \vert \varphi_{n} \rangle \langle \varphi_{n} \vert
\end{equation}
where $ \sum_{\Omega_{n}} $ stands for a summation (or an integral of some kind) over a set $ \Omega_{n} $ of states $ \vert \varphi_{n} \rangle \in \Omega_{n} $ and $ A_{\vert \varphi_{n} \rangle} $ is a real \textit{probability} (or \textit{probability density}) of the state $ \vert \varphi_{n} \rangle $. In the latter case we allow $ A_{\vert \varphi_{n} \rangle} $ to include delta functions, defined e.g. as a measure w.r.t. the given integral.

The renormalization of the transition probability can be performed in complete parallel with the previous sections. Let us do it e.g. with the functional regularization. The renormalization constant analogical to \eqref{Cvarphi} for a state described by $ \hat{A}_{n} $ shall be
\begin{equation}\label{CA}
\begin{aligned}
C_{\hat{A}_{n}} &\equiv \Tr \left[ \mathbb{U}_{n+1}^{\varepsilon} \hat{A}_{n} \mathbb{U}_{n+1}^{\varepsilon \ast} \right]
\end{aligned}
\end{equation}
The trace of an operator $ \hat{D}_{n} $ is defined as
\begin{equation}\label{key}
\Tr \left[ \hat{D}_{n} \right] \equiv \sum_{\delta_{n}} ~_{\mathtt{c}} \langle \delta_{n} \vert \hat{D}_{n} \vert \delta_{n} \rangle_{\mathtt{c}}
\end{equation}
and it works the same in the adapted basis: the reader can easily check that by writing $ \sum_{\delta_{n}} ~_{\dot{\mathtt{c}}} \langle \delta_{n} \vert \hat{D}_{n} \vert \delta_{n} \rangle_{\dot{\mathtt{c}}} $ and directly plugging in \eqref{alphandot2}. In the special case of a pure state $ \hat{A}_{n} = \vert \varphi_{n} \rangle \langle \varphi_{n} \vert $, the constant is equal to \eqref{Cvarphi}, i.e.,
\begin{equation}\label{key}
C_{\vert \varphi_{n} \rangle \langle \varphi_{n} \vert} = C_{\vert \varphi_{n} \rangle}
\end{equation}
Let us assume that the initial state of the system is described by a density operator $ \hat{A}_{n} $ such that $ C_{\hat{A}_{n}} $ is finite and nonzero. Then we state that the probability of a single time-step transition from the state described by $ \hat{A}_{n} $ to another state described by $ \hat{B}_{n+1} $ is expressed as
\begin{equation}\label{prob1dens}
p_{\hat{A}_{n} \rightarrow \hat{B}_{n+1}} = \lim_{\varepsilon \rightarrow 0} ~ C_{\hat{A}_{n}}^{-1} ~ \Tr \left[  \hat{B}_{n+1} \mathbb{U}_{n+1}^{\varepsilon} \hat{A}_{n} \mathbb{U}_{n+1}^{\varepsilon \ast} \right] 
\end{equation}
And again, for pure states this probability is in correspondence with \eqref{prob1}. By fixing $ \hat{B}_{n+1} = \hat{\mathbf{1}}_{n+1} $ (which means \textit{all states}), one finds that the normalization is correct,
\begin{equation}\label{key}
p_{\hat{A}_{n} \rightarrow \hat{\mathbf{1}}_{n+1}} = 1
\end{equation}

In case of cylindrical regularization, one considers the evolution map $ \mathbb{U}_{n+1}^{0} $ instead of $ \mathbb{U}_{n+1}^{\varepsilon} $ and employs the \textit{physical trace}
\begin{equation}\label{phystrace}
\begin{aligned}
\Tr \left[ \hat{D}_{n+1} \right]^{\mathrm{ph}} &\equiv \sum_{\delta_{n+1}} ~_{\mathtt{c}} \langle \delta_{n+1} \vert \hat{D}_{n+1} \vert \delta_{n+1} \rangle_{\mathtt{c}}^{\mathrm{ph}} \\
&= \abs{\det \Sigma_{r} }^{-1} \sum_{\tau_{n+1}} ~_{\mathtt{c}}\langle V \bar{\Sigma}^{-1} \gamma_{n+1} \vert \hat{D}_{n+1} \vert V \bar{\Sigma}^{-1} \gamma_{n+1} \rangle_{\mathtt{c}} 
\end{aligned}
\end{equation}
for operators $ \hat{D}_{n+1} $ on $ \Phi^{\times}_{n+1} $. The second line of \eqref{phystrace} follows from \eqref{pipcanonical} and uses the convention \eqref{gammataueta}. Similarly to the physical inner product, the physical trace leaves out the integration over the free parameters, and therefore turns out finite even for evolved density operators $ \mathbb{U}_{n+1}^{0} \hat{A}_{n} \mathbb{U}_{n+1}^{0 \ast} $ whose normal trace diverges.

\subsection{Global Evolution and the Path Integral}\label{sec:globevol}

Now that we have in our hand the (regularized) one-step evolution scheme, we can extend our analysis to multiple time-steps. Assume that $ m, n $ are two whole numbers such that $ m > n $. Then it is natural to define the multi-step evolution mapping $ \mathbb{U}_{mn}: \Phi^{\times}_{n} \rightarrow \Phi^{\times}_{m} $ which evolves states from a time-slice $ n $ to a later time-slice $ m $ by $ \mathbb{U}_{mn} = \mathbb{U}_{m} ... \mathbb{U}_{n+1} $. We shall put $ n = 0 $, $ m = t $ and seek the propagator wavefunction $ \langle \psi_{t} \vert \mathbb{U}_{t0} \vert \varphi_{0} \rangle $ given by two arbitrary states $ \vert \varphi_{0} \rangle \in \Phi^{\times}_{0} $ and $ \vert \psi_{t} \rangle \in \Phi^{\times}_{t} $. It is a matter of applying \eqref{prop} repeatedly to obtain the path-integral form for the propagator:
\begin{equation}\label{longprop}
\begin{aligned}
\langle \psi_{t} \vert \mathbb{U}_{t0} \vert \varphi_{0} \rangle &= \sum_{\alpha_{0}...\alpha_{t}} \langle  \psi_{t} \vert \alpha_{t} \rangle_{\mathtt{c}} ~_{\mathtt{c}}\langle \alpha_{t} \vert \mathbb{U}_{t} \vert \alpha_{t-1} \rangle_{\mathtt{c}} ~_{\mathtt{c}}\langle \alpha_{t-1} \vert ~ ... \\ & \qquad \qquad ~ \vert \alpha_{1} \rangle_{\mathtt{c}} ~_{\mathtt{c}}\langle \alpha_{1} \vert \mathbb{U}_{1} \vert \alpha_{0} \rangle_{\mathtt{c}} ~_{\mathtt{c}}\langle \alpha_{0} \vert \varphi_{0} \rangle = \\
&= \sum_{\alpha_{0}...\alpha_{t}} (2\pi)^{-tq/2} ~ \abs{ \prod_{m = 1}^{t} \det \Sigma_{r}(R_{m}) }^{1/2} \left( \prod_{n = 0}^{t-1} \mathcal{l}(V_{2}(R_{n})^{T} \alpha_{n+1}) \right) \\ & \qquad \qquad ~ e^{i S(\alpha_{0}, ..., \alpha_{t})} ~ \langle  \psi_{t} \vert \alpha_{t} \rangle_{\mathtt{c}}  ~_{\mathtt{c}}\langle \alpha_{0} \vert \varphi_{0} \rangle
\end{aligned}
\end{equation}
with the overall action
\begin{equation}\label{key}
S(\alpha_{0}, ..., \alpha_{t}) = \sum_{n = 0}^{t-1} S_{n+1}(\alpha_{n}, \alpha_{n+1})
\end{equation}
The propagator \eqref{longprop} inherits all the normalization problems of \eqref{prop}, so it is again suitable (at least for applications in which the irregularity is not really physical) to establish a renormalization procedure in order to compute probabilities. Note that this effectively makes the constant factor $ \abs{ \prod_{m = 1}^{t} \det \Sigma_{r}(R_{m}) }^{1/2} $ irrelevant. For clarity, we will keep it in place. If we use the functional regularization for our evolution map, we get the propagator
\begin{equation}\label{propreg}
\begin{aligned}
\langle \psi_{t} \vert \mathbb{U}_{t0}^{\varepsilon} \vert \varphi_{0} \rangle &= \sum_{\alpha_{0}...\alpha_{t}} (2\pi)^{-tq/2} ~ \abs{ \prod_{m = 1}^{t} \det \Sigma_{r}(R_{m}) }^{1/2} ~ \langle  \psi_{t} \vert \alpha_{t} \rangle_{\mathtt{c}}  ~_{\mathtt{c}}\langle \alpha_{0} \vert \varphi_{0} \rangle \\ & \qquad \qquad ~ \pi^{-ts_{n}/4} ~ \varepsilon^{ts_{n}/4} ~ e^{- \frac{1}{2} \varepsilon M( V_{2}(R_{0})^{T} \alpha_{1} , ..., V_{2}(R_{t-1})^{T} \alpha_{t})} ~ e^{i S(\alpha_{0}, ..., \alpha_{t})} 
\end{aligned}
\end{equation}
with
\begin{equation}\label{key}
M(\lambda_{1}, ..., \lambda_{t}) \equiv \sum_{n = 0}^{t-1}  \lambda_{n+1}^{T} \lambda_{n+1}
\end{equation}
Here we may introduce the complex regularized action
\begin{equation}\label{key}
S^{\varepsilon}(\alpha_{0}, ..., \alpha_{t}) \equiv \sum_{n = 0}^{t-1} S^{\varepsilon}_{n+1}(\alpha_{n},\alpha_{n+1})
\end{equation}
composed of the individual time-step contributions
\begin{equation}\label{key}
S^{\varepsilon}_{n+1}(\alpha_{n},\alpha_{n+1}) \equiv S_{n+1}(\alpha_{n}, \alpha_{n+1}) + \textstyle i \frac{1}{2} \varepsilon ~ \alpha_{n+1}^{T} V_{2}(R_{n})  V_{2}(R_{n})^{T} \alpha_{n+1} - i \frac{1}{4} s_{n} \ln(\frac{\varepsilon}{\pi})
\end{equation}
Then we can rewrite the regularized path integral in the form
\begin{equation}\label{key}
\begin{aligned}
\langle \psi_{t} \vert \mathbb{U}_{t0}^{\varepsilon} \vert \varphi_{0} \rangle &= \sum_{\alpha_{0}...\alpha_{t}} (2\pi)^{-tq/2} ~ \abs{ \prod_{m = 1}^{t} \det \Sigma_{r}(R_{m}) }^{1/2} ~ \langle  \psi_{t} \vert \alpha_{t} \rangle_{\mathtt{c}}  ~_{\mathtt{c}}\langle \alpha_{0} \vert \varphi_{0} \rangle ~ e^{i S^{\varepsilon}(\alpha_{0}, ..., \alpha_{t})} 
\end{aligned}
\end{equation}
Observe that the imaginary part of the regularized action supplies an additional dependence on the null space parts of $ \alpha_{n+1} $.  When $ \varepsilon \rightarrow 0 $, the logarithm takes over and it happens that $ \Im S^{\varepsilon}(\alpha_{0}, ..., \alpha_{t}) \rightarrow + \infty $. Nevertheless, observe that we cannot discard the term $ i \frac{1}{2} \varepsilon ~ \alpha_{n+1}^{T} V_{2}(R_{n})  V_{2}(R_{n})^{T} \alpha_{n+1} $ in the action, as this would cause the integral to diverge. In other words, we must keep the order of the integral $ \sum_{\alpha_{0}...\alpha_{t}} $ (which acts first) and the $ \lim_{\varepsilon \rightarrow 0} $ (which acts second).

Now let us express
\begin{equation}\label{key}
\begin{aligned}
& \langle \varphi_{0} \vert \overline{\mathbb{U}_{t0}^{\varepsilon}}  \mathbb{U}_{t0}^{\varepsilon} \vert \varphi_{0} \rangle = \\
&= \sum_{\alpha_{0}...\alpha_{t}} \sum_{\gamma_{0}...\gamma_{t}} ~ (2\pi)^{-tq} ~ \abs{ \prod_{m = 1}^{t} \det \Sigma_{r}(R_{m}) } ~ \langle \varphi_{0} \vert \gamma_{0} \rangle_{\mathtt{c}}  ~_{\mathtt{c}}\langle \alpha_{0} \vert \varphi_{0} \rangle ~_{\mathtt{c}}\langle \gamma_{t} \vert \alpha_{t} \rangle_{\mathtt{c}} \\
& \qquad \qquad ~ e^{-i \overline{S^{\varepsilon}(\gamma_{0}, ..., \gamma_{t})}} ~ e^{i S^{\varepsilon}(\alpha_{0}, ..., \alpha_{t})} =\\
&= \sum_{\alpha_{0}...\alpha_{t}} \sum_{\gamma_{0}...\gamma_{t-1}} ~ (2\pi)^{-tq} ~ \abs{ \prod_{m = 1}^{t} \det \Sigma_{r}(R_{m}) } ~ \langle \varphi_{0} \vert \gamma_{0} \rangle_{\mathtt{c}}  ~_{\mathtt{c}}\langle \alpha_{0} \vert \varphi_{0} \rangle \\
& \qquad \qquad ~ e^{-i \overline{S^{\varepsilon}(\gamma_{0}, ..., \gamma_{t-1}, \alpha_{t})}} ~ e^{i S^{\varepsilon}(\alpha_{0}, ..., \alpha_{t})}
\end{aligned}
\end{equation}
We could in principle renormalize the multi-step transition probability the same way we did it before for a single time-step: by defining the constant $ C_{\vert \varphi_{0} \rangle}^{\varepsilon t} \equiv \langle \varphi_{0} \vert \overline{\mathbb{U}_{t0}^{\varepsilon}}  \mathbb{U}_{t0}^{\varepsilon} \vert \varphi_{0} \rangle $ and stating that
\begin{equation}\label{key}
p_{\vert \varphi_{0} \rangle \rightarrow \vert \psi_{t} \rangle} = \lim_{\varepsilon \rightarrow 0} ~ C_{\vert \varphi_{0} \rangle}^{\varepsilon t ~ -1} ~ \langle \varphi_{0} \vert \overline{\mathbb{U}_{t0}^{\varepsilon}} \vert \psi_{t} \rangle \langle \psi_{t} \vert \mathbb{U}_{t0}^{\varepsilon} \vert \varphi_{0} \rangle
\end{equation}
This procedure does not seem very practical though. One can imagine it rather demanding to compute the integrals \textit{and} the limit. Therefore, if one has to renormalize, one should think of ways to simplify the expressions or come up with an approximation to make the computation more affordable.

Next let us look into the cylindrical regularization of multi-step evolution. The final state computed under cylindrical regularization will of course appear simpler, because it will not contain the functional regularization term:
\begin{equation}\label{cylfinalstate}
\begin{aligned}
\mathbb{U}_{t0}^{0} \vert \varphi_{0} \rangle &= \sum_{\alpha_{0}...\alpha_{t}} (2\pi)^{-tq/2} ~ \abs{ \prod_{m = 1}^{t} \det \Sigma_{r}(R_{m}) }^{1/2} ~ \vert \alpha_{t} \rangle_{\mathtt{c}}  ~_{\mathtt{c}}\langle \alpha_{0} \vert \varphi_{0} \rangle ~ e^{i S(\alpha_{0}, ..., \alpha_{t})} 
\end{aligned}
\end{equation}
However, we have to be careful when passing from \eqref{cylfinalstate} to a (propagator) wavefunction, since we in general cannot use the inner product of $ \mathcal{F}_{t} $. Instead, we should use the inner product of some suitable $ \mathcal{F}_{t}^{\mathrm{ph}} $. This physical Hilbert space will be determined by the evolution across the \textit{whole interval} between $ 0 $ and $ t $ so it will be interval-dependent. In other words, the cylindrical regularization is inherently global.

How to find $ \mathcal{F}_{t}^{\mathrm{ph}} $? First it is useful to think of the classical evolution. The classical one-step Hamiltonian evolution map comes with a pre-constraint and a set of free parameters. When we line up a number of these maps one after another, some of the free parameters produced at earlier time-steps are fixed by pre-constraints at later time-steps; other remain free. To understand the resulting structure of $ \mathcal{P}_{t} $, one may turn back to section 2.2. There we have defined the space $ \mathcal{D}_{t} \subset \mathcal{P}_{t} $ of data on the phase space which have a corresponding solution on the whole interval, i.e., the final post-constraint surface. According to Observation \ref{DdotDN}, it split as $ \mathcal{D}_{t} = \dot{\mathcal{D}}_{t} \oplus \mathcal{N}_{t} $ into a representative space of classes of symplectically equivalent data and the null space, respectively. We claim that it is exactly the (unique) null space $ \mathcal{N}_{t} $ which is generated by the ultimately free parameters of the evolution. This is because all the physical information about the system is carried by the symplectic structure which vanishes on $ \mathcal{N}_{t} $. Therefore, we can fix all the free parameters arbitrarily (thereby choosing a point in $ \mathcal{N}_{t} $) without affecting the physical results. This fixing of free parameters is exactly what happens when we pass from $ \mathcal{F}_{t} $ to $ \mathcal{F}_{t}^{\mathrm{ph}} $. Let us denote the space of coordinate parts of vectors from $ \mathcal{N}_{t} $ by $ \mathcal{N}_{t}^{\mathcal{Q}} \subset \mathcal{Q}_{n} $, the space of coordinate parts of vectors from $ \dot{\mathcal{D}}_{t} $ will be denoted by $ \dot{\mathcal{D}}_{t}^{\mathcal{Q}}  \subset \mathcal{Q}_{n} $. These conventions imply that $ \mathcal{Q}_{t} = \mathcal{N}_{t}^{\mathcal{Q}} \oplus \dot{\mathcal{D}}_{t}^{\mathcal{Q}} $. We may eventually fix $ \mathcal{F}_{t}^{\mathrm{ph}} = L^{2}(\dot{\mathcal{D}}_{t}^{\mathcal{Q}}) $. This is reminiscent of section 4, where the exact same space was used to construct the Hilbert space $ \mathcal{H}_{t} $.

So far the theory. In principle, it is of course possible to find $ \mathcal{D}_{t} $, make some suitable choice of the representative space $ \dot{\mathcal{D}}_{t} $ (recall that it is not a priori unique) and from it obtain the physical Hilbert space $ \mathcal{F}_{t}^{\mathrm{ph}} $ with the inner product $ \langle ~ \rangle^{\mathrm{ph}} $. The problem is that in order to do so, one would typically need to solve the classical equations of motion on the whole interval. Nevertheless, once it is done, \eqref{cylfinalstate} can be used to compute the regularized inner product $ \langle \psi_{t}^{\mathrm{ph}} \vert \mathbb{U}_{t0}^{0} \vert \varphi_{0} \rangle^{\mathrm{ph}} $ with a final state given by the wavefunction $ \psi_{t}^{\mathrm{ph}} \in \mathcal{F}_{t}^{\mathrm{ph}} $. It is also possible to consider $ \langle \varphi_{0} \vert \overline{\mathbb{U}_{t0}^{0}}  \mathbb{U}_{t0}^{0} \vert \varphi_{0} \rangle^{\mathrm{ph}} $ which is the norm of the wavefunction $ _{\mathtt{c}} \langle \gamma_{t} \vert \mathbb{U}_{t0}^{0} \vert  \varphi_{0} \rangle $ when seen as a function on $ \mathcal{F}_{t}^{\mathrm{ph}} $. The renormalization would be performed in the same way as before by denoting $ C_{\vert \varphi_{0} \rangle}^{0 t} \equiv \langle \varphi_{0} \vert \overline{\mathbb{U}_{t0}^{0}}  \mathbb{U}_{t0}^{0} \vert \varphi_{0} \rangle^{\mathrm{ph}} $ and stating that
\begin{equation}\label{key}
p_{\vert \varphi_{0} \rangle \rightarrow \vert \psi_{t} \rangle} = C_{\vert \varphi_{0} \rangle}^{0 t ~ -1} ~ \langle \psi_{t}^{\mathrm{ph}} \vert \mathbb{U}_{t0}^{0} \vert \varphi_{0} \rangle^{\mathrm{ph}}
\end{equation}
This approach can be potentially simpler than the one with functional regularization since there are no extra terms in the integrand and---more importantly---no limit. On the other hand, it requires us to locate the space $ \mathcal{D}_{t} $ which is typically a non-trivial task. In the functional regularization scheme, this is done implicitly by means of including the regularization terms in the integrand and performing the limit. Either approach can be favorable, depending on a particular application. For instance, in a numerical computation, on may settle for the functional regularization with some small parameter $ \varepsilon $, leading to an approximate result.

Our final note is the following. In section 4.3 we argued that some features shared by the unitary model and the model of \cite{Hoehn2014, Hoehn2014a} may be inconvenient for certain applications. In particular, we have in mind the back-propagation of constraints, which appears to be an inevitable feature of the unitary evolution but is at the same time unphysical and breaks causality. In the present section, we spent effort to make the non-unitary model interval-independent. For this reason we did not a priori reduce the state spaces; instead we quantized the system separately at each time-slice and introduced a one-step evolution map which naturally implements the pre-constraint. Consequently, the map is non-unitary and requires regularization in order to obtain well-defined wavefunctions. The two offered regularization schemes allow for that, along with the possibility to renormalize probabilities. The resulting path integrals feature the standard propagator and an integration over the canonical variables, which is a great advantage over the unitary model. Under functional regularization, the path integral keeps the formal interval-independence. Moreover, it necessitates no additional analysis of the structure of classical solutions, which is something we usually expect from a path integral. On the other hand, the cylindrical regularization is interval-dependent and does require an analysis of the classical solution space, but in turn makes the expression for the path integral simpler and perhaps better accessible for computation.

\section{Massless Scalar Field on a 2D Spacetime Lattice}\label{toy}
Finally we present an exemplary application of the above introduced formalism to the case of massless scalar field on a fixed two-dimensional Regge triangulation.  The setup will be exactly the same as in section 6 of \cite{Kaninsky2020a}, where it was used to illustrate an analysis of the classical evolution. The only difference is that here the field is quantized. We will summarize the properties of the model only briefly, for more details see the given reference. The lattice is assumed to be composed of a finite number of spacelike slices indexed by $ n \in \{0, ..., t\} $ such that every slice includes a finite number of vertices (at most $ q $) and every vertex is a member of exactly one slice. The triangulation is then provided with additional \textit{virtual} vertices so that the number of vertices in each time-slice is exactly $ q $; these have no physical meaning and are not connected by any edges. After adding the virtual vertices, the total number of vertices in the triangulation is $ N = q(t+1) $. To every vertex $ i $ we associate a field value $ \varphi_{i} \in \mathbb{R} $. The corresponding scalar field action is taken to be
\begin{equation}\label{sfaction}
S_{0t} = \frac{1}{2} \sum_{\text{edges } ij} w_{ij} ~ (\varphi_{i}-\varphi_{j})^{2}
\end{equation}
with the weight
\begin{equation}\label{key}
w_{ij} = \frac{\mathcal{V}_{ij}}{\mathcal{l}_{ij}^{2}}
\end{equation}
associated to each edge $ ij $, where $ \mathcal{V}_{ij} $ is the dual edge volume and $ \mathcal{l}_{ij}^{2} $ the squared edge length. For simplicity, we assume there are only two kinds of edges in the triangulation: \textit{spacelike} edges with squared edge lengths equal to $ 1 $ and \textit{timelike} edges with squared edge lengths equal to $ -1/2 $. The triangulation is built up from two types of triangles denoted by (1,2) and (2,1). Both types have equal area $ A $. It follows that $ \mathcal{V}_{ij} = m A $ where $ m $ is the number of triangles which contain the given edge. Edges with $ m = 1 $ are called \textit{boundary edges}, edges with $ m = 2 $ are called \textit{interior edges}. We will assume that there is no timelike boundary, therefore all timelike edges will be interior edges. Meanwhile, in order to describe the evolution we split the triangulation into individual time-steps, therefore all spacelike edges will turn into boundary edges. For simplicity, we divide the action by the overall constant $ 2A $. In result, we shall have
\begin{equation}\label{wijint2}
w_{ij} = \begin{cases}
-2 & \text{if } ij \text{ is an interior timelike edge} \\
1/2 & \text{if } ij \text{ is a boundary spacelike edge}
\end{cases}
\end{equation}

The action \eqref{sfaction} can be written compactly as
\begin{equation}\label{sfactioncondensed}
	S_{0t} = \textstyle \frac{1}{2} ~ \varphi^{T} K \varphi
\end{equation}
where $ \varphi \in \mathbb{R}^{N} $ and $ K $ is a real, symmetric $ N \times N $ matrix which we call the \textit{dynamical matrix}. By comparison of \eqref{sfaction} and \eqref{sfactioncondensed}, one finds that
\begin{equation}\label{Kij}
	K_{ij} = \delta_{ij} \sum_{k} \delta^{e}_{ik} w_{ik} - \delta^{e}_{ij} w_{ij}
\end{equation}
where
\begin{equation}\label{key}
\delta^{e}_{ij} = \begin{cases}
1 & \text{if the vertices } i \text{ and } j \text{ are connected by an edge} \\
0 & \text{otherwise}
\end{cases}
\end{equation}
The lattice is split into individual time-steps, which induces a corresponding splitting of the matrix $ K $. We let $ K_{(n)} $ be the $ q \times q $ submatrix of $ K $ corresponding to time-slice $ n $ and $ K_{(n,n+1)} $ the $ q \times q $ submatrix of $ K $ with rows corresponding to variables at time-slice $ n $ and columns corresponding to variables at time-slice $ n+1 $. The former is further separated into two parts $ K_{(n)} = K_{(n)}^{-} + K_{(n)}^{+} $ (for $ n = 1, ..., t-1 $) where $ K_{(n)}^{-} $ and $  K_{(n)}^{+} $ describe the boundary time-slice $ n $ of the two separated time steps: one between $ n-1 $ and $ n $ ($ - $), other between $ n $ and $ n+1 $ ($ + $). The action associated to the time-step between $ n $ and $ n+1 $ then takes the form
\begin{equation}\label{Sn+1ex}
S_{n+1} = \textstyle \frac{1}{2} \left( x_{n}^{T} K_{(n)}^{+} x_{n} + 2 x_{n}^{T} K_{(n,n+1)} x_{n+1} + x_{n+1}^{T} K_{(n+1)}^{-} x_{n+1} \right)
\end{equation}
where $ x_{n} \in \mathbb{R}^{q} $ is the $ q $-tuple of field values $ \varphi_{i} $ in vertices $ i $ belonging to time-slice $ n $. One identifies the matrices in \eqref{prepostmomentamatrix} as
\begin{equation}\label{LRK}
	\begin{aligned}
		L_{n}&= \textstyle K_{(n)}^{+} \qquad & R_{n}&=K_{(n,n+1)} \\
		\bar{L}_{n+1}&= -R_{n}^{T} = -K_{(n+1,n)} \qquad & \bar{R}_{n+1}&= - K_{(n+1)}^{-} \\
	\end{aligned}
\end{equation}
The sum of the contributions \eqref{Sn+1ex} gives the overall action \eqref{sfactioncondensed},
\begin{equation}\label{actionadd}
	S_{0t} = \sum_{n=0}^{t-1} S_{n+1}(x_{n},x_{n+1})
\end{equation}

\subsection{One-Step Examples}
We will now give a set of three basic examples of one-step fixed lattices endowed with massless quantum scalar field. The very same examples have been analyzed in \cite{Kaninsky2020a}, the reader is welcomed to compare the present quantum case to the classical one. We shall employ the general non-unitary model of section 5. For better readability, let us repeat the most important formulas accompanying this model. The evolution is characterized by the map \eqref{Un+1c} defined as
\begin{equation}\label{ev}
\mathbb{U}_{n+1} \vert \alpha_{n} \rangle_{\dot{\mathtt{c}}} = \sum_{\lambda_{n+1}} \mathcal{l}(\lambda_{n+1}) ~ \vert \alpha_{n+1} \rangle_{\ddot{\mathtt{c}}} ~ \delta^{s_{n}}( \iota_{n}) ~ f(\alpha_{n+1})
\end{equation}
with the conventions
\begin{equation}\label{key}
\alpha_{n} \equiv \begin{pmatrix}
\varrho_{n} \\
\iota_{n}
\end{pmatrix}, \qquad \alpha_{n+1} \equiv \begin{pmatrix}
\varrho_{n}\\
\lambda_{n+1}
\end{pmatrix}
\end{equation}
\begin{equation}\label{f}
f(\alpha_{n+1}) = e^{i\frac{1}{2} \alpha_{n+1}^{T} \bar{\Sigma}^{-1} V^{T} \bar{R}_{n+1} V \bar{\Sigma}^{-1} \alpha_{n+1}}
\end{equation}
It results in the canonical eigenvector basis propagator \eqref{prop}
\begin{equation}\label{pr}
_{\mathtt{c}} \langle \beta_{n+1} \vert \mathbb{U}_{n+1} \vert \gamma_{n} \rangle_{\mathtt{c}} = \mathcal{l}(V_{2}^{T} \beta_{n+1}) ~ \abs{\det \Sigma_{r} }^{1/2} ~ (2\pi)^{-q/2} ~ e^{iS_{n+1}(\gamma_{n},\beta_{n+1})}
\end{equation}
We will be interested not only in the evolution itself, but also in the effect of either of the introduced regularization schemes. First, we shall look at functional regularization, which is given by a special choice \eqref{el}, \eqref{kepsilon} of the extra term
\begin{equation}\label{l}
\mathcal{l}(\lambda_{n+1}) \equiv \pi^{-s_{n}/4} ~ \varepsilon^{s_{n}/4} ~ e^{-\frac{1}{2} \varepsilon \lambda_{n+1}^{T}\lambda_{n+1} }
\end{equation}
Upon evolving some initial state $ \vert \varphi_{n} \rangle $ to the next time-step, its norm is generally not conserved. According to \eqref{vphiUvepsisq}, \eqref{Cvarphi}, it changes to
\begin{equation}\label{Cvp}
C_{\vert \varphi_{n} \rangle} \equiv \langle \varphi_{n} \vert \overline{\mathbb{U}_{n+1}^{\varepsilon}} \mathbb{U}_{n+1}^{\varepsilon} \vert \varphi_{n} \rangle = \sum_{\varrho_{n}} \langle \varphi_{n} \vert J_{r_{n}} \varrho_{n} \rangle_{\dot{\mathtt{c}}} ~_{\dot{\mathtt{c}}} \langle J_{r_{n}} \varrho_{n} \vert \varphi_{n} \rangle
\end{equation}
In the second cylindrical regularization scheme, we lay $ \mathcal{l} = 1 $ and rely on the \textit{physical inner product} $ \langle ~ \rangle^{\mathrm{ph}} $ to dispose the divergence originating from the integration over free parameters. In the case of one-step evolution, \eqref{1spip} and \eqref{pipcanonical} instruct us on how it should be computed:
\begin{equation}\label{1pip}
\begin{aligned}
\langle \psi_{n+1} \vert \phi_{n+1} \rangle^{\mathrm{ph}} &= \sum_{\tau_{n+1}} ~ \langle \psi_{n+1} \vert \gamma_{n+1} \rangle_{\ddot{\mathtt{c}}} ~_{\ddot{\mathtt{c}}} \langle \gamma_{n+1} \vert \phi_{n+1} \rangle = \\ &= \abs{\det \Sigma_{r} }^{-1} \sum_{\tau_{n+1}} ~ \langle \psi_{n+1} \vert V \bar{\Sigma}^{-1} \gamma_{n+1} \rangle_{\mathtt{c}} ~_{\mathtt{c}}\langle V \bar{\Sigma}^{-1} \gamma_{n+1} \vert \phi_{n+1} \rangle
\end{aligned}
\end{equation}
where $ \gamma_{n+1} \equiv \begin{pmatrix}
\tau_{n+1} & \eta_{n+1}
\end{pmatrix}^{T} $ and $ \eta_{n+1} $ represents additional free parameters which are not integrated over. With this, the physical norm of an initial state $ \vert \varphi_{n} \rangle $ evolved one time-step is given by \eqref{phnorm}, \eqref{Cphnorm} and reads
\begin{equation}\label{Cvpcr}
C_{\vert \varphi_{n} \rangle} \equiv \langle \varphi_{n} \vert \overline{\mathbb{U}^{0}_{n+1}} \mathbb{U}^{0}_{n+1} \vert \varphi_{n} \rangle^{\mathrm{ph}} = \sum_{\varrho_{n}} \langle \varphi_{n} \vert J_{r_{n}} \varrho_{n} \rangle_{\dot{\mathtt{c}}} ~_{\dot{\mathtt{c}}} \langle J_{r_{n}} \varrho_{n} \vert \varphi_{n} \rangle
\end{equation}
which is identical to \eqref{Cvp}. Therefore, at least for one-step evolution, both regularization schemes produce the same (re)normalization. With the above formulas in mind, we can proceed to the examples.

\begin{example}\label{ex:1}
First we consider a time-step between time-slices 0 and 1 with exactly three vertices at each time-slice. The lattice is depicted in Fig. \ref{fig:lat1}.

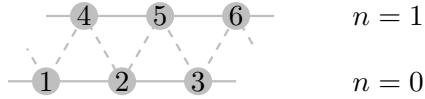
\begin{figure}[h!]
	\centering
	\begin{tikzpicture}[scale=1]
	\tikzset{
		vertex/.style={
			shape=circle,fill=lightgray!100,minimum size=3mm,inner sep=0.2mm, label={[fill=none,label distance=1mm]90:#1}
		},
		vertexwhite/.style={
			shape=circle,fill=white!100,minimum size=10mm,inner sep=0.2mm, label={[fill=none,label distance=1mm]90:#1}
		},
		edge/.style={
			draw,-,color=lightgray!100,line width=0.3mm
		},
		edget/.style={
			draw,dashed,color=lightgray!100,line width=0.3mm
		}
	}
	
	\coordinate (cia) at (-2,0);
	\coordinate (c1) at (-1,0);
	\coordinate (c2) at (0,0);
	\coordinate (c3) at (1,0);
	\coordinate (cib) at (2,0);
	\coordinate (cic) at (-1.5,0.866);
	\coordinate (c4) at (-0.5,0.866);
	\coordinate (c5) at (0.5,0.866);
	\coordinate (c6) at (1.5,0.866);
	\coordinate (cid) at (2.5,0.866);

	\draw[edge] (cia) -- (c1) -- (c2) -- (c3) -- (cib);
	\draw[edget] (cic) -- (c1) -- (c4) -- (c2) -- (c5) -- (c3) -- (c6) -- (cib);
	\draw[edge] (cic) -- (c4) -- (c5) -- (c6) -- (cid);
	
	\node[vertex] at (c1) {1};
	\node[vertex] at (c2) {2};
	\node[vertex] at (c3) {3};
	\node[vertex] at (c4) {4};
	\node[vertex] at (c5) {5};
	\node[vertex] at (c6) {6};
	\node[vertexwhite] at (cia) {};
	\node[vertexwhite] at (cib) {};
	\node[vertexwhite] at (cic) {};
	\node[vertexwhite] at (cid) {};
	\node[] at (3.5,0) {$ n = 0 $};
	\node[] at (3.5,0.866) {$ n = 1 $};

	\end{tikzpicture}
	\vspace{0 mm}
	\caption{Diagram of the time-step lattice of Example \ref{ex:1}. The fragments of edges on the right are meant to be connected to the fragments on the left, so that each time-slice is a closed loop. Timelike edges are drawn in dashed line.}
	\label{fig:lat1}
	\vspace{4 mm}
\end{figure}

We have $ q = 3 $, $ t = 1 $ and $ N = 6 $. The dynamical matrix \eqref{Kij} is
\begin{equation}\label{key}
		K = \begin{pmatrix}
			-3 & -1/2 & -1/2 & 2 & 0 & 2 \\
			& -3 & -1/2 & 2 & 2 & 0 \\
			& & -3 & 0 & 2 & 2 \\
			& & & -3 & -1/2 & -1/2 \\
			& & & & -3 & -1/2 \\
			& & & & & -3 \\
		\end{pmatrix}
\end{equation}
Because the matrix is symmetric, we only write the upper triangle. It is easy to read out the matrices of \eqref{LRK}. Since we consider a single time-slice, no splitting is needed. We get in particular
\begin{equation}\label{key}
		R_{0} = \begin{pmatrix}
			2 & 0 & 2 \\
			2 & 2 & 0 \\
			0 & 2 & 2
		\end{pmatrix}
\end{equation}
and the singular value decomposition $ R_{0} = U \Sigma V^{T} $ produces
\begin{equation}\label{key}
		U = \frac{1}{\sqrt{6}} \begin{pmatrix}
			\sqrt{2} & 0 & -2 \\
			\sqrt{2} & -\sqrt{3} & 1 \\
			\sqrt{2} & \sqrt{3} & 1
		\end{pmatrix}, ~ \Sigma = \begin{pmatrix}
			4 & 0 & 0 \\
			0 & 2 & 0 \\
			0 & 0 & 2
		\end{pmatrix}, ~ V = \frac{1}{\sqrt{6}} \begin{pmatrix}
			\sqrt{2} & -\sqrt{3} & -1 \\
			\sqrt{2} & 0 & 2 \\
			\sqrt{2} & \sqrt{3} & -1
		\end{pmatrix}
\end{equation}
Now we can easily plug into our formulas and describe the evolution of the quantum system. Since $ R_{0} $ is regular, the evolution will turn out unitary, without any need for regularization.

We can proceed to a description of the quantum system. Since we have $ s_{0} = 0 $, the one-step evolution map \eqref{ev} simplifies to
\begin{equation}\label{key}
\mathbb{U}_{1} \vert \alpha_{0} \rangle_{\dot{\mathtt{c}}} = \mathcal{l} ~ \vert \alpha_{1} \rangle_{\ddot{\mathtt{c}}} ~ f(\alpha_{1})
\end{equation}
with constant $ \mathcal{l} $. The evolution is regular. It is natural to put $ \mathcal{l} = 1 $ which also follows if we straightforwardly plug in the functional regularization prescription \eqref{l}. Next, the propagator \eqref{pr} amounts to
\begin{equation}\label{key}
\begin{aligned}
_{\mathtt{c}} \langle \beta_{1} \vert \mathbb{U}_{1} \vert \gamma_{0} \rangle_{\mathtt{c}} &= \abs{\det \Sigma }^{1/2} ~ (2\pi)^{-q/2} ~ e^{iS_{1}(\gamma_{0},\beta_{1})} =\\
&= 4 (2\pi)^{-3/2} ~ e^{iS_{1}(\gamma_{0},\beta_{1})}
\end{aligned}
\end{equation}
According to \eqref{sfactioncondensed}, the action equals
\begin{equation}\label{key}
S_{1}(\gamma_{0},\beta_{1}) = S(\gamma_{0},\beta_{1}) = \frac{1}{2} \begin{pmatrix}
\gamma_{0}^{T} & \beta_{1}^{T}
\end{pmatrix} K \begin{pmatrix}
\gamma_{0} \\ \beta_{1}
\end{pmatrix}
\end{equation}
If we start with a normalizable initial state $ \vert \varphi_{0} \rangle $ so that it has the (finite) norm $ \langle \varphi_{0} \vert \varphi_{0} \rangle $ and recall \eqref{Cvp}, we find that the norm of the final state will be
\begin{equation}\label{key}
\langle \varphi_{0} \vert \overline{\mathbb{U}_{1}} \mathbb{U}_{1} \vert \varphi_{0} \rangle = \langle \varphi_{0} \vert \varphi_{0} \rangle
\end{equation}
i.e., the normalization is conserved and the quantum evolution is indeed unitary.

\end{example}

\vspace{\baselineskip}

\begin{example}\label{ex:2}
The second example of one-step lattice shall have three vertices at time-slice 0 but only one vertex at time-slice 1, as depicted in Fig. \ref{fig:lat2}.
	
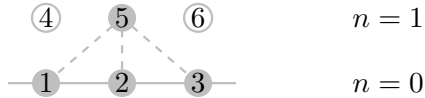
\begin{figure}[h!]
	\centering
	\begin{tikzpicture}[scale=1]
	\tikzset{
		vertex/.style={
			shape=circle,fill=lightgray!100,minimum size=3mm,inner sep=0.2mm, label={[fill=none,label distance=1mm]90:#1}
		},
		vertexwhite/.style={
			shape=circle,fill=white!100,minimum size=10mm,inner sep=0.2mm, label={[fill=none,label distance=1mm]90:#1}
		},
		vertexvirtual/.style={
			shape=circle,draw=lightgray!100,fill=white!100,line width=0.3mm,minimum size=2.7mm,inner sep=0.2mm, 	label={[fill=none,label distance=1mm]90:#1}
		},
		edge/.style={
			draw,-,color=lightgray!100,line width=0.3mm
		},
		edget/.style={
			draw,dashed,color=lightgray!100,line width=0.3mm
		},
	}
	
	\coordinate (cia) at (-2,0);
	\coordinate (c1) at (-1,0);
	\coordinate (c2) at (0,0);
	\coordinate (c3) at (1,0);
	\coordinate (cib) at (2,0);
	\coordinate (cic) at (-2,0.866);
	\coordinate (c4) at (-1,0.866);
	\coordinate (c5) at (0,0.866);
	\coordinate (c6) at (1,0.866);
	\coordinate (cid) at (2,0.866);

	\draw[edge] (cia) -- (c1) -- (c2) -- (c3) -- (cib);
	\draw[edget] (c1) -- (c5) -- (c2);
	\draw[edget] (c3) -- (c5);
	
	\node[vertex] at (c1) {1};
	\node[vertex] at (c2) {2};
	\node[vertex] at (c3) {3};
	\node[vertexvirtual] at (c4) {4};
	\node[vertex] at (c5) {5};
	\node[vertexwhite] at (cia) {};
	\node[vertexwhite] at (cib) {};
	\node[vertexwhite] at (cic) {};
	\node[vertexwhite] at (cid) {};
	\node[vertexvirtual] at (c6) {6};
	\node[] at (3.5,0) {$ n = 0 $};
	\node[] at (3.5,0.866) {$ n = 1 $};
	
	\end{tikzpicture}
	\vspace{0 mm}
	\caption{Diagram of the time-step lattice of Example \ref{ex:2}. It is made of three identical type 2-1 triangles. Vertices 4 and 6 are virtual. Dashed edges are timelike.}
	\label{fig:lat2}
	\vspace{4 mm}
\end{figure}
	
We have $ q = 3 $, $ t = 1 $ and $ N = 6 $ as before. We read out the dynamical matrix
\begin{equation}\label{key}
		K = \begin{pmatrix}
			-1 & -1/2 & -1/2 & 0 & 2 & 0 \\
			& -1 & -1/2 & 0 & 2 & 0 \\
			& & -1 & 0 & 2 & 0 \\
			& & & 0 & 0 & 0 \\
			& & & & -6 & 0 \\
			& & & & & 0 \\
		\end{pmatrix}
\end{equation}
including the upper right block
\begin{equation}\label{key}
	R_{0} = \begin{pmatrix}
			0 & 2 & 0 \\
			0 & 2 & 0 \\
			0 & 2 & 0
		\end{pmatrix}
\end{equation}
describing the interaction between time-slices. It may be decomposed into
\begin{equation}\label{key}
		U = \frac{1}{\sqrt{6}} \begin{pmatrix}
			\sqrt{2} & -\sqrt{3} & -1 \\
			\sqrt{2} & 0 & 2 \\
			\sqrt{2} & \sqrt{3} & -1
		\end{pmatrix}, ~ \Sigma = \begin{pmatrix}
			2 \sqrt{3} & 0 & 0 \\
			0 & 0 & 0 \\
			0 & 0 & 0
		\end{pmatrix}, ~ V = \begin{pmatrix}
			0 & 0 & 1 \\
			1 & 0 & 0 \\
			0 & 1 & 0
		\end{pmatrix}
\end{equation}

We see that $ r_{0} = 1 $ and $ s_{0} = 2 $ so the lattice is irregular. The one-step quantum evolution map \eqref{ev} then gives
\begin{equation}\label{ua2}
\mathbb{U}_{1} \vert \alpha_{0} \rangle_{\dot{\mathtt{c}}} = \sum_{\alpha_{12},\alpha_{13}} \mathcal{l}(\alpha_{12},\alpha_{13}) ~ \vert \alpha_{1} \rangle_{\ddot{\mathtt{c}}} ~ \delta(\alpha_{02}) \delta(\alpha_{03}) ~ f(\alpha_{1})
\end{equation}
Here one can see how the pre-constraint implemented by $ \delta(\alpha_{02}) \delta(\alpha_{03}) $ sets the final state to zero unless the adapted coordinates of the initial state satisfy $ \alpha_{02} = \alpha_{03} = 0 $. Meanwhile, the final state takes on various adapted coordinates $ \alpha_{12},\alpha_{13} $ weighted by the function $ \mathcal{l} $. Recall that these can in general also enter the function $ f $. That does not happen here due to the form of the matrix $ \bar{R}_{1} = -K^{-}_{(1)} $ whose only nonzero element is $ (\bar{R}_{1})_{22} = 6 $. Upon plugging into \eqref{f}, one finds that $ f $ contains only the adapted coordinate $ \alpha_{11} $. This is because the adapted coordinates $ \alpha_{12},\alpha_{13} $ happen to describe the field values at virtual vertices and thus bear no physical meaning.

The propagator \eqref{pr} amounts to
\begin{equation}\label{wf2}
_{\mathtt{c}} \langle \beta_{1} \vert \mathbb{U}_{1} \vert \gamma_{0} \rangle_{\mathtt{c}} = \mathcal{l}(\beta_{11}, \beta_{13}) ~ (2 \sqrt{3})^{1/2} ~ (2\pi)^{-3/2} ~ e^{iS(\gamma_{0},\beta_{1})}
\end{equation}
and the action given by \eqref{Sn+1ex} can be rewritten as
\begin{equation}\label{S2}
\begin{aligned}
S(\gamma_{0},\beta_{1}) & = \textstyle \frac{1}{2} \left( \gamma_{0}^{T} K_{(0)}^{+} \gamma_{0} + 2\gamma_{0}^{T} K_{(0,1)} \beta_{1} + \beta_{1}^{T} K_{(1)}^{-} \beta_{1} \right) = \\
&= \textstyle \frac{1}{2} \left( \gamma_{0}^{T} K_{(0)}^{+} \gamma_{0} + 4 \beta_{12}(\gamma_{01}+\gamma_{02}+\gamma_{03}) - 6\beta_{12}^{2} \right)
\end{aligned}
\end{equation}
First let us point out that the mixed term of \eqref{S2} hides the non-trivial pre-constraint $ \delta(\alpha_{02}) \delta(\alpha_{03}) $ from \eqref{ua2}. At the same time, it tells us that the field value $ \beta_{12} $ at vertex $ 5 $ is coupled to field values $ \gamma_{01}, \gamma_{02}, \gamma_{03} $ at all the three vertices of time-slice $ 0 $. It comes as no surprise that the field values $ \beta_{11}, \beta_{13} $ at virtual vertices are completely missing: if it had not been for the regularization term $ \mathcal{l}(V_{2}^{T} \beta_{1}) = \mathcal{l}(\beta_{11}, \beta_{13}) $ (the $ 3 \times 2 $ matrix $ V_{2} $ is formed by the last two columns of $ V $), the wavefunction \eqref{wf2} would be constant with respect to them.

If we want to compute inner products on $ \mathcal{F}_{1} $, e.g. in order to normalize the final state, we need to set up the regularization scheme. Assume that we start with a normalizable initial state $ \vert \varphi_{0} \rangle $. Under the functional regularization \eqref{l}, the wavefunction $ _{\mathtt{c}} \langle \beta_{1} \vert \mathbb{U}_{1}^{\varepsilon} \vert \varphi_{0} \rangle $ generated from \eqref{wf2} will become square integrable in the variables $ \beta_{11}, \beta_{12} $ and will therefore belong to $ \mathcal{F}_{1} $, at least as long as we keep $ \varepsilon > 0 $. Thanks to the normalization $ \sum_{\beta_{11}, \beta_{13}} \mathcal{l}(\beta_{11}, \beta_{13})^{2} = 1 $, our initial state $ \vert \varphi_{0} \rangle $ will evolve into $ \mathbb{U}^{\varepsilon}_{1} \vert \varphi_{0} \rangle $ with the norm
\begin{equation}\label{Cphi2}
C_{\vert \varphi_{n} \rangle} = \langle \varphi_{0} \vert \overline{\mathbb{U}^{\varepsilon}_{1}} \mathbb{U}^{\varepsilon}_{1} \vert \varphi_{0} \rangle = \sum_{\alpha_{01}} \langle \varphi_{0} \vert \alpha_{0} \rangle_{\dot{\mathtt{c}}} ~_{\dot{\mathtt{c}}} \langle \alpha_{0} \vert \varphi_{0} \rangle
\end{equation}
where it must be understood that $ \alpha_{0} = \begin{pmatrix}
\alpha_{01} & 0 & 0
\end{pmatrix}^{T} $. Since there is no normalizable state $ \vert \varphi_{0} \rangle $ which would fully satisfy the pre-constraint (the state would have to be sharp in the adapted coordinates $ \alpha_{02} = \alpha_{03} = 0 $ w.r.t. the basis $ \{ \dot{e}_{01}, \dot{e}_{02}, \dot{e}_{03} \} $ of $ \mathcal{Q}_{0} $), the norm \eqref{Cphi2} of any normalizable state will be smaller than $ \langle \varphi_{0} \vert \varphi_{0} \rangle $. This clearly demonstrates the non-unitary of the evolution. If $ C_{\vert \varphi_{n} \rangle} \neq 0 $, we can renormalize the final state so that it has the same norm as $ \vert \varphi_{0} \rangle $.

Eventually let us consider some state $ \vert \psi_{1} \rangle $ and look at its inner product with the final state $ \mathbb{U}^{\varepsilon}_{1} \vert \varphi_{0} \rangle $. It is given as
\begin{equation}\label{psiUepsphi2}
\langle \psi_{1} \vert \mathbb{U}^{\varepsilon}_{1} \vert \varphi_{0} \rangle = \sum_{\beta_{1}} \sum_{\gamma_{0}} \langle \psi_{1} \vert \beta_{1} \rangle_{\mathtt{c}} ~ _{\mathtt{c}} \langle \gamma_{0} \vert \varphi_{0} \rangle ~ \mathcal{l}(\beta_{11}, \beta_{13}) ~ (2\sqrt{3})^{1/2} ~ (2\pi)^{-3/2} ~ e^{iS(\gamma_{0},\beta_{1})}
\end{equation}
The reader may notice that the wavefunction $ _{\mathtt{c}} \langle \beta_{1} \vert \psi_{1} \rangle $ depends on the meaningless variables $ \beta_{11}, \beta_{13} $ describing field values at virtual vertices $ 4, 6 $. Indeed, there is no canonical way of choosing this wavefunction. Luckily, the particular choice does not matter, since the unphysical variables only enter into the integral $ \sum_{\beta_{11}, \beta_{13}} \langle \psi_{1} \vert \beta_{1} \rangle_{\mathtt{c}} ~ \mathcal{l}(\beta_{11}, \beta_{13}) $. For $ _{\mathtt{c}} \langle \beta_{1} \vert \psi_{1} \rangle $ from $ \mathcal{F}_{1} $, the integral is finite. In the limit $ \varepsilon \rightarrow 0 $, all free parameters are given the same weight so the result will be influenced only by the dependence of $ _{\mathtt{c}} \langle \beta_{1} \vert \psi_{1} \rangle $ on the physical $ \beta_{12} $. For a computation of the transition probability, we shall normalize $ \vert \psi_{1} \rangle $ so that $ \langle \psi_{1} \vert \psi_{1} \rangle = 1 $.

The cylindrical regularization is perhaps more natural here. It fixes the evolved eigenstates of the adapted coordinates as
\begin{equation}\label{ua2c}
\mathbb{U}_{1}^{0} \vert \alpha_{0} \rangle_{\dot{\mathtt{c}}} = \sum_{\alpha_{12},\alpha_{13}} \vert \alpha_{1} \rangle_{\ddot{\mathtt{c}}} ~ \delta(\alpha_{02}) \delta(\alpha_{03}) ~ f(\alpha_{1})
\end{equation}
Consequently, the wavefunction \eqref{wf2} under cylindrical regularization turns out constant w.r.t. $ \beta_{11}, \beta_{13} $. The evolved state $ \mathbb{U}_{1}^{0} \vert \varphi_{0} \rangle $ is therefore not normalizable; its wavefunction does not belong to $ \mathcal{F}_{1} $. The remedy is provided by the physical inner product \eqref{1pip}, taking the form
\begin{equation}\label{pip2}
\begin{aligned}
\langle \psi_{1} \vert \phi_{1} \rangle^{\mathrm{ph}} &= \sum_{\gamma_{11}} ~ \langle \psi_{1} \vert \gamma_{1} \rangle_{\ddot{\mathtt{c}}} ~_{\ddot{\mathtt{c}}} \langle \gamma_{1} \vert \phi_{1} \rangle = \\
&= (2\sqrt{3})^{-1} \sum_{\gamma_{11}} ~ \langle \psi_{1} \vert V \bar{\Sigma}^{-1} \gamma_{1} \rangle_{\mathtt{c}} ~_{\mathtt{c}}\langle V \bar{\Sigma}^{-1} \gamma_{1} \vert \phi_{1} \rangle
\end{aligned}
\end{equation}
We see that instead of integrating over all the (adapted) coordinates (as in the inner product on $ \mathcal{F}_{1} $), the physical inner product features only one integral over the adapted coordinate $ \gamma_{11} $ while $ \gamma_{12}, \gamma_{13} $ play the role of additional parameters. The matrix
\begin{equation}\label{key}
V \bar{\Sigma}^{-1} = \begin{pmatrix}
0 & 0 & 1 \\
1 & 0 & 0 \\
0 & 1 & 0
\end{pmatrix} \begin{pmatrix}
2 \sqrt{3} & 0 & 0 \\
0 & 1 & 0 \\
0 & 0 & 1
\end{pmatrix}^{-1} = \begin{pmatrix}
0 & 0 & 1 \\
(2\sqrt{3})^{-1} & 0 & 0 \\
0 & 1 & 0
\end{pmatrix}
\end{equation}
just transforms the adapted coordinates back to the canonical ones. In particular, note that $ (V \bar{\Sigma}^{-1} \gamma_{1})_{2} = (2\sqrt{3})^{-1} \gamma_{11} $ is the canonical coordinate corresponding to the field value at vertex $ 5 $. What we see in \eqref{pip2} is an integration over this field value. This justifies us calling the inner product `physical'---it disregards the field values at virtual vertices $ 4, 6 $ which have no physical meaning and integrates the wavefunctions only over the field value at vertex $ 5 $. With this, the evolved state wavefunction $  _{\mathtt{c}} \langle \beta_{1} \vert \mathbb{U}_{1}^{0} \vert \varphi_{0} \rangle $ is found to have a finite norm in the physical Hilbert space $ \mathcal{F}_{1}^{\mathrm{ph}} $ equal to
\begin{equation}\label{key}
\langle \varphi_{0} \vert \overline{\mathbb{U}^{0}_{1}} \mathbb{U}^{0}_{1} \vert \varphi_{0} \rangle^{\mathrm{ph}} = \sum_{\gamma_{11}} \langle \varphi_{0} \vert \gamma_{1} \rangle_{\dot{\mathtt{c}}} ~_{\dot{\mathtt{c}}} \langle \gamma_{1} \vert \varphi_{0} \rangle
\end{equation}
where it must be understood that $ \gamma_{1} = \begin{pmatrix} \gamma_{11} & 0 & 0 \end{pmatrix}^{T} $. This norm is equivalent to \eqref{Cphi2} and can be interpreted as the norm of $ _{\mathtt{c}} \langle \gamma_{0} \vert \varphi_{0} \rangle $ when limited to the pre-constraint surface. This makes sense, because $ \mathcal{F}_{1}^{\mathrm{ph}} $ is the space of square integrable functions on $ \dot{\mathcal{C}}^{+\mathcal{Q}}_{n+1} $ and the latter space is in fact isomorphic to the coordinate part $ \mathcal{C}_{0}^{- \mathcal{Q}} $ of the pre-constraint surface $ \mathcal{C}_{0}^{-} $.

The physical inner product of the final state $ \mathbb{U}^{0}_{1} \vert \varphi_{0} \rangle $ with another state $ \vert \psi_{1} \rangle $ can be written down as
\begin{equation}\label{psiU0phi2}
\langle \psi_{1} \vert \mathbb{U}^{0}_{1} \vert \varphi_{0} \rangle^{\mathrm{ph}} = (2\sqrt{3})^{-1/2} ~ (2\pi)^{-3/2} \sum_{\beta_{11}} \sum_{\gamma_{0}} ~ \langle \psi_{1} \vert V \bar{\Sigma}^{-1} \beta_{1} \rangle_{\mathtt{c}} ~_{\mathtt{c}}\langle \gamma_{0} \vert \varphi_{0} \rangle _{\mathtt{c}} ~ e^{iS(\gamma_{0},V \bar{\Sigma}^{-1} \beta_{1})}
\end{equation}
In comparison to the analogical expression \eqref{psiUepsphi2} subjected to functional regularization, we see that \eqref{psiU0phi2} leaves out the integrals over the free parameters $ \beta_{12}, \beta_{13} $. Instead of $ \vert \psi_{1} \rangle $, we can consider a state $ \vert \psi_{1}^{\mathrm{ph}} \rangle $ from $ \Phi_{1}^{\mathrm{ph} \times} $, whose wavefunction belongs to $ \mathcal{F}_{1}^{\mathrm{ph}} $ and therefore is square integrable w.r.t. the physical inner product. Then it can be normalized so that $ \langle \psi_{1}^{\mathrm{ph}} \vert \psi_{1}^{\mathrm{ph}} \rangle = 1 $ and used for a computation of the transition probability into this state without paying any attention to the redundant free parameters.

\end{example}

\vspace{\baselineskip}

\begin{example}\label{ex:3}
The last instance of a triangular one-step lattice is the time-reversed version of that from Example \ref{ex:2}. Its diagram is given in Fig. \ref{fig:lat3}.
	
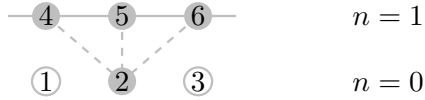
\begin{figure}[h!]
	\centering
	\begin{tikzpicture}[scale=1]
	\tikzset{
		vertex/.style={
			shape=circle,fill=lightgray!100,minimum size=3mm,inner sep=0.2mm, label={[fill=none,label distance=1mm]90:#1}
		},
		vertexwhite/.style={
			shape=circle,fill=white!100,minimum size=10mm,inner sep=0.2mm, label={[fill=none,label distance=1mm]90:#1}
		},
		vertexvirtual/.style={
			shape=circle,draw=lightgray!100,fill=white!100,line width=0.3mm,minimum size=2.7mm,inner sep=0.2mm, 	label={[fill=none,label distance=1mm]90:#1}
		},
		edge/.style={
			draw,-,color=lightgray!100,line width=0.3mm
		},
		edget/.style={
			draw,dashed,color=lightgray!100,line width=0.3mm
		},
	}
	
	\coordinate (cia) at (-2,0);
	\coordinate (c1) at (-1,0);
	\coordinate (c2) at (0,0);
	\coordinate (c3) at (1,0);
	\coordinate (cib) at (2,0);
	\coordinate (cic) at (-2,0.866);
	\coordinate (c4) at (-1,0.866);
	\coordinate (c5) at (0,0.866);
	\coordinate (c6) at (1,0.866);
	\coordinate (cid) at (2,0.866);
	
	\draw[edget] (c4) -- (c2) -- (c6);
	\draw[edget] (c5) -- (c2);
	\draw[edge] (cic) -- (c4) -- (c5) -- (c6) -- (cid);
	
	\node[vertexvirtual] at (c1) {1};
	\node[vertex] at (c2) {2};
	\node[vertexvirtual] at (c3) {3};
	\node[vertex] at (c4) {4};
	\node[vertex] at (c5) {5};
	\node[vertexwhite] at (cia) {};
	\node[vertexwhite] at (cib) {};
	\node[vertexwhite] at (cic) {};
	\node[vertexwhite] at (cid) {};
	\node[vertex] at (c6) {6};
	\node[] at (3.5,0) {$ n = 0 $};
	\node[] at (3.5,0.866) {$ n = 1 $};
	\end{tikzpicture}
	\vspace{0 mm}
	\caption{Diagram of the time-step lattice of Example \ref{ex:3}. It again contains three identical triangles, this time of type 1-2. Vertices 1 and 3 are virtual.}
	\label{fig:lat3}
	\vspace{4 mm}
\end{figure}

This time we have
\begin{equation}\label{key}
		K = \begin{pmatrix}
			0 & 0 & 0 & 0 & 0 & 0 \\
			& -6 & 0 & 2 & 2 & 2 \\
			& & 0 & 0 & 0 & 0 \\
			& & & -1 & -1/2 & -1/2 \\
			& & & & -1 & -1/2 \\
			& & & & & -1 \\
		\end{pmatrix}
\end{equation}
and thus
\begin{equation}\label{key}
		R_{0} = \begin{pmatrix}
			0 & 0 & 0 \\
			2 & 2 & 2 \\
			0 & 0 & 0
		\end{pmatrix}
\end{equation}
with the decomposition
\begin{equation}\label{key}
		U = \begin{pmatrix}
			0 & 0 & 1 \\
			1 & 0 & 0 \\
			0 & 1 & 0
		\end{pmatrix}, ~ \Sigma = \begin{pmatrix}
			2 \sqrt{3} & 0 & 0 \\
			0 & 0 & 0 \\
			0 & 0 & 0
		\end{pmatrix}, ~ V = \frac{1}{\sqrt{6}} \begin{pmatrix}
			\sqrt{2} & -\sqrt{3} & -1 \\
			\sqrt{2} & 0 & 2 \\
			\sqrt{2} & \sqrt{3} & -1
		\end{pmatrix}
\end{equation}

We again have $ r_{0} = 1 $ and $ s_{0} = 2 $. The one-step quantum evolution map given by \eqref{ev} formally looks the same as in the preceding example:
\begin{equation}\label{ua3}
\mathbb{U}_{1} \vert \alpha_{0} \rangle_{\dot{\mathtt{c}}} = \sum_{\alpha_{12},\alpha_{13}} \mathcal{l}(\alpha_{12},\alpha_{13}) ~ \vert \alpha_{1} \rangle_{\ddot{\mathtt{c}}} ~ \delta(\alpha_{02}) \delta(\alpha_{03}) ~ f(\alpha_{1})
\end{equation}
However, it differs in content. Note that the adapted coordinates $ \alpha_{0} $ and $ \alpha_{1} $ do not have the same meaning as before. Their defining symplectic transformation is related, but not equivalent to the one from Example \ref{ex:2}. For instance, one can show that the adapted coordinates $ \alpha_{02}, \alpha_{03} $ actually represent the canonical momenta at vertices $ 3, 1 $, respectively; see Example 6.3 in \cite{Kaninsky2020a}. The pre-constraint $ \delta(\alpha_{02}) \delta(\alpha_{03}) $ therefore does not influence any physical quantities, it only fixes the values of these momenta to zero. This behavior is typical. In the quantum-mechanical realm, we can understand it as a consequence of the uncertainty relations: the coordinates (field values) at virtual vertices are fully undetermined, so the corresponding momenta must be fixed. Meanwhile, the free adapted coordinates $ \alpha_{12}, \alpha_{13} $ stand for some linear combinations of the canonical coordinates describing field values at vertices $ 4, 5, 6 $. In contrast to the previous example, these are all real vertices, which makes both $ \alpha_{12}, \alpha_{13} $ physically meaningful. Indeed, they may be important e.g. for the measurement of the field at time-slice $ 1 $ or for the evolution to later time-slices.

Recall that according to \eqref{f}, the function $ f $ equals
\begin{equation}\label{key}
f(\alpha_{1}) = e^{i\frac{1}{2} \alpha_{1}^{T} \bar{\Sigma}^{-1} V^{T} \bar{R}_{1} V \bar{\Sigma}^{-1} \alpha_{1}}
\end{equation}
where
\begin{equation}\label{key}
\bar{R}_{1} = -K^{-}_{(1)} = \begin{pmatrix}
1 & 1/2 & 1/2 \\
1/2 & 1 & 1/2 \\
1/2 & 1/2 & 1
\end{pmatrix}
\end{equation}
It follows that $ f $ depends on all the three adapted coordinates $ \alpha_{11}, \alpha_{12}, \alpha_{13} $ including the latter two, which are the free parameters of the evolution. In view of the preceding, this is expected, since they correspond to (a linear combination of) field values at vertices $ 4, 5, 6 $ which enter the action.

The propagator \eqref{pr} takes the form
\begin{equation}\label{wf3}
_{\mathtt{c}} \langle \beta_{1} \vert \mathbb{U}_{1} \vert \gamma_{0} \rangle_{\mathtt{c}} = \mathcal{l}(V_{2}^{T} \beta_{1}) ~ (2\sqrt{3})^{1/2} ~ (2\pi)^{-3/2} ~ e^{iS(\gamma_{0},\beta_{1})}
\end{equation}
cf. \eqref{wf2}. Again, $ V_{2} $ is formed by the last two columns of $ V $, i.e.,
\begin{equation}\label{V2Tbeta}
V_{2}^{T} \beta_{1} = \frac{1}{\sqrt{6}} \begin{pmatrix}
-\sqrt{3} & 0 & \sqrt{3} \\
-1 & 2 & -1
\end{pmatrix} \begin{pmatrix}
\beta_{11} \\
\beta_{12} \\
\beta_{13}
\end{pmatrix} =  \frac{-1}{\sqrt{6}} \begin{pmatrix}
\sqrt{3} \left( \beta_{11} - \beta_{13} \right) \\
\beta_{11} - 2 \beta_{12} + \beta_{13}
\end{pmatrix}
\end{equation}
These happen to be the linear combinations of field values $ \beta_{11}, \beta_{12}, \beta_{13} $ at vertices $ 4, 5, 6 $ (respectively) which are not fixed by the classical equations of motion. With the help of \eqref{Sn+1ex}, we find the action to be
\begin{equation}\label{S3}
\begin{aligned}
S(\gamma_{0},\beta_{1}) & = \textstyle \frac{1}{2} \left( \gamma_{0}^{T} K_{(0)}^{+} \gamma_{0} + 2\gamma_{0}^{T} K_{(0,1)} \beta_{1} + \beta_{1}^{T} K_{(1)}^{-} \beta_{1} \right) = \\
&= \textstyle \frac{1}{2} \left( - 6 \gamma_{02}^{2} + 4 \gamma_{02} ( \beta_{11}+\beta_{12}+\beta_{13} ) + \beta_{1}^{T} K_{(1)}^{-} \beta_{1} \right)
\end{aligned}
\end{equation}
where we took advantage of the only nonzero element of $ K_{(0)}^{+} $ being $ (K_{(0)}^{+})_{22} = -6 $. In comparison to \eqref{S2}, we see that the field values at time-slice 0 and the field values at time-slice 1 has effectively switched their roles. The coordinates $ \gamma_{01}, \gamma_{03} $ are not present because they belong to virtual vertices. Consequently, the wavefunction \eqref{wf3} is constant w.r.t. $ \gamma_{01}, \gamma_{03} $. One can also check that the mixed term is independent of the two linear combinations from \eqref{V2Tbeta}; these only enter the last term.

Let us consider a normalizable initial state $ \vert \varphi_{0} \rangle $, employ the functional regularization and express the evolved state wavefunction:
\begin{equation}\label{betaUphi3}
\begin{aligned}
_{\mathtt{c}} \langle \beta_{1} \vert \mathbb{U}_{1}^{\varepsilon} \vert \varphi_{0} \rangle &= \sum_{\gamma_{0}}  ~_{\mathtt{c}} \langle \beta_{1} \vert \mathbb{U}_{1}^{\varepsilon} \vert \gamma_{0} \rangle_{\mathtt{c}} ~_{\mathtt{c}} \langle \gamma_{0} \vert \varphi_{0} \rangle = \\ &= \sum_{\gamma_{02}} ~ \mathcal{l}(V_{2}^{T} \beta_{1}) ~ (2\sqrt{3})^{1/2} ~ (2\pi)^{-3/2} ~ e^{iS(\gamma_{0},\beta_{1})} \sum_{\gamma_{01},\gamma_{03}} ~ _{\mathtt{c}} \langle \gamma_{0} \vert \varphi_{0} \rangle
\end{aligned}
\end{equation}
Note that one could have such $ \vert \varphi_{0} \rangle $ that the integral over $ \gamma_{01}, \gamma_{03} $ diverges. This does not need to concern us: wavefunctions of normalizable states simply \textit{can} diverge sometimes. Consider the momentum wavefunction $ _{\mathtt{m}} \langle \beta_{0} \vert \varphi_{0} \rangle $, for which
\begin{equation}\label{mom0}
_{\mathtt{m}} \langle 0 \vert \varphi_{0} \rangle = (2 \pi)^{-q/2} ~ \sum_{\gamma_{0}} ~_{\mathtt{c}} \langle \gamma_{0} \vert \varphi_{0} \rangle
\end{equation}
This example well illustrates the situation in \eqref{betaUphi3}. Recall that the momenta $ \alpha_{02}, \alpha_{03} $ at vertices $ 1,3 $ are fixed to zero by the pre-constraint. Upon transforming from momenta to the canonical coordinates, we get \eqref{betaUphi3} with the second integral analogical to \eqref{mom0}. Let us remark that one could prevent all divergences by considering only initial states with wavefunctions from $ \Phi_{0} \subset \mathcal{F}_{0} $. But it is not necessary. Since the variables $ \gamma_{01}, \gamma_{03} $ are unphysical, we are free to fix the wavefunction w.r.t. them arbitrarily. In particular, we can choose such $ \vert \varphi_{0} \rangle $ that the divergence is avoided.

The functional regularization \eqref{l} ensures that the wavefunction \eqref{betaUphi3} is square integrable in the two variables defined by the linear combinations of \eqref{V2Tbeta} and therefore belongs to $ \mathcal{F}_{1} $. Its norm is
\begin{equation}\label{Cphi3}
C_{\vert \varphi_{n} \rangle} = \langle \varphi_{0} \vert \overline{\mathbb{U}^{\varepsilon}_{1}} \mathbb{U}^{\varepsilon}_{1} \vert \varphi_{0} \rangle = \sum_{\alpha_{01}} \langle \varphi_{0} \vert \alpha_{0} \rangle_{\dot{\mathtt{c}}} ~_{\dot{\mathtt{c}}} \langle \alpha_{0} \vert \varphi_{0} \rangle
\end{equation}
where it must be understood that $ \alpha_{0} = \begin{pmatrix}
\alpha_{01} & 0 & 0
\end{pmatrix}^{T} $, i.e., the point described by $ \alpha_{0} $ belongs to the pre-constraint surface. We remark that the adapted coordinate $ \alpha_{01} $ in fact describes the field value at vertex $ 2 $. As in the previous example, the norm \eqref{Cphi2} associated to any normalizable state $  \vert \varphi_{0} \rangle $ will be smaller than $ \langle \varphi_{0} \vert \varphi_{0} \rangle $.

Having expressed the final state wavefunction, one can evaluate its inner product with another state $ \vert \psi_{1} \rangle $, e.g. to compute the transition probability. From \eqref{betaUphi3} we get
\begin{equation}\label{key}
\langle \psi_{1} \vert \mathbb{U}^{\varepsilon}_{1} \vert \varphi_{0} \rangle = \sum_{\beta_{1}} \sum_{\gamma_{0}} \langle \psi_{1} \vert \beta_{1} \rangle_{\mathtt{c}} ~ _{\mathtt{c}} \langle \gamma_{0} \vert \varphi_{0} \rangle ~ \mathcal{l}(V_{2}^{T} \beta_{1}) ~ (2\sqrt{3})^{1/2} ~ (2\pi)^{-3/2} ~ e^{iS(\gamma_{0},\beta_{1})}
\end{equation}
Observe that $ _{\mathtt{c}} \langle \beta_{1} \vert \psi_{1} \rangle $ may have a physically meaningful dependence on all the field values $ \beta_{11}, \beta_{12}, \beta_{13} $ at vertices $ 4, 5, 6 $. Because the linear combinations in \eqref{V2Tbeta} represent free parameters of the evolution, the product $ _{\mathtt{c}} \langle \beta_{1} \vert \psi_{1} \rangle ~ \mathcal{l}(V_{2}^{T} \beta_{1}) $ is the only part of the integrand dependent on them. If  $ _{\mathtt{c}} \langle \beta_{1} \vert \psi_{1} \rangle $ belongs to $ \mathcal{F}_{1} $, the integral over $ \beta_{1} $ will clearly converge. In the limit $ \varepsilon \rightarrow 0 $, all values of the free parameters are given the same weight and one gets a physically relevant inner product whose absolute value squared describes the transition probability. As we know, it can be renormalized upon dividing by \eqref{Cphi3}.

Finally let us shortly look at the cylindrical regularization. It provides the evolved eigenstates of the adapted coordinates in the form
\begin{equation}\label{ua3c}
\mathbb{U}_{1}^{0} \vert \alpha_{0} \rangle_{\dot{\mathtt{c}}} = \sum_{\alpha_{12},\alpha_{13}} \vert \alpha_{1} \rangle_{\ddot{\mathtt{c}}} ~ \delta(\alpha_{02}) \delta(\alpha_{03}) ~ f(\alpha_{1})
\end{equation}
Our discussion of \eqref{ua3} still applies here. Inevitably, one will find that the evolved state $ \mathbb{U}_{1}^{0} \vert \varphi_{0} \rangle $ is not normalizable; its wavefunction $ _{\mathtt{c}} \langle \beta_{1} \vert \mathbb{U}_{1}^{0} \vert \varphi_{0} \rangle $ does not belong to $ \mathcal{F}_{1} $. We need to apply the physical inner product \eqref{1pip}, which is of the form
\begin{equation}\label{pip3}
\begin{aligned}
\langle \psi_{1} \vert \phi_{1} \rangle^{\mathrm{ph}} &= \sum_{\gamma_{11}} ~ \langle \psi_{1} \vert \gamma_{1} \rangle_{\ddot{\mathtt{c}}} ~_{\ddot{\mathtt{c}}} \langle \gamma_{1} \vert \phi_{1} \rangle = \\
&= (2\sqrt{3})^{-1} \sum_{\gamma_{11}} ~ \langle \psi_{1} \vert V \bar{\Sigma}^{-1} \gamma_{1} \rangle_{\mathtt{c}} ~_{\mathtt{c}}\langle V \bar{\Sigma}^{-1} \gamma_{1} \vert \phi_{1} \rangle
\end{aligned}
\end{equation}
where the adapted coordinates $ \gamma_{12}, \gamma_{13} $---recall that they are given by the two rows of \eqref{V2Tbeta}---play the role of additional parameters. The matrix $ V \bar{\Sigma}^{-1} $ produces linear combinations of the adapted coordinates $ \gamma_{11}, \gamma_{12}, \gamma_{13} $ corresponding to the canonical ones. Observe that unlike in the preceding Example \ref{ex:2}, all the adapted coordinates are physical here, yet the integration over $ \gamma_{12}, \gamma_{13} $ is still leaved out. This happens because the physical inner product is respectful of the symplectic structure carried over from time-slice 0 whose configuration space has only one physical degree of freedom---the field value at vertex 2. Thanks to this, the evolved state wavefunction $ _{\mathtt{c}} \langle \beta_{1} \vert \mathbb{U}_{1}^{0} \vert \varphi_{0} \rangle $ is found to have a finite norm in $ \mathcal{F}_{1}^{\mathrm{ph}} $ equal to
\begin{equation}\label{key}
\langle \varphi_{0} \vert \overline{\mathbb{U}^{0}_{1}} \mathbb{U}^{0}_{1} \vert \varphi_{0} \rangle^{\mathrm{ph}} = \sum_{\gamma_{11}} \langle \varphi_{0} \vert \gamma_{1} \rangle_{\dot{\mathtt{c}}} ~_{\dot{\mathtt{c}}} \langle \gamma_{1} \vert \varphi_{0} \rangle
\end{equation}
where we assume $ \gamma_{1} = \begin{pmatrix} \gamma_{11} & 0 & 0 \end{pmatrix}^{T} $. This norm is again equivalent to \eqref{Cphi3} and can be interpreted as the norm of $ _{\mathtt{c}} \langle \gamma_{0} \vert \varphi_{0} \rangle $ when limited to the pre-constraint surface. The physical inner product of the final state $ \mathbb{U}_{1}^{0} \vert \varphi_{0} \rangle $ with some state $ \vert \psi_{1} \rangle $ shall take the form
\begin{equation}\label{psiU0phi3}
\langle \psi_{1} \vert \mathbb{U}^{0}_{1} \vert \varphi_{0} \rangle^{\mathrm{ph}} = (2\sqrt{3})^{-1/2} ~ (2\pi)^{-3/2} \sum_{\beta_{11}} \sum_{\gamma_{0}} ~ \langle \psi_{1} \vert V \bar{\Sigma}^{-1} \beta_{1} \rangle_{\mathtt{c}} ~_{\mathtt{c}}\langle \gamma_{0} \vert \varphi_{0} \rangle _{\mathtt{c}} ~ e^{iS(\gamma_{0},V \bar{\Sigma}^{-1} \beta_{1})}
\end{equation}
The wavefunction $ \langle \psi_{1} \vert V \bar{\Sigma}^{-1} \beta_{1} \rangle_{\mathtt{c}} $ contains important information about all the (adapted) coordinates $ \beta_{11}, \beta_{12}, \beta_{13} $. Therefore, unlike in Example \ref{ex:2}, the state $ \vert \psi_{1} \rangle $ cannot be replaced by $ \vert \psi_{1}^{\mathrm{ph}} \rangle $ whose wavefunction would only depend on $ \beta_{11} $---unless we parametrize it with $ \beta_{12}, \beta_{13} $, that is. The normalization $ \langle \psi_{1}^{\mathrm{ph}} \vert \psi_{1}^{\mathrm{ph}} \rangle = 1 $ appropriate for computing the transition probability then must be satisfied for all values of the additional, yet physical parameters $ \beta_{12}, \beta_{13} $.
\end{example}

With this we conclude our exposition of the one-step examples. The evolution across lattices with more time-steps could be treated in complete analogy, by stacking the one-step evolution maps one after another, as described in section \ref{sec:globevol}. This would only result in increasingly complicated expressions for the final states, their norms, inner products and ultimately the transition probabilities. In case of functional regularization, the difficulty of the computation would be amplified by the regularization terms and the necessary limiting procedure. In case of cylindrical regularization, we would encounter the need to analyze the structure of classical solutions on the whole studied time interval in order to fix the physical inner product.

\section{Conclusion}
Classical discrete dynamical systems have wide use in many branches of science including economics, demography or biology \cite{Galor2007}. In physics, one can also encounter their quantum counterparts. These are indispensable to a family of quantum gravity approaches based on discrete spacetime models, mainly the descendants of quantum Regge calculus \cite{Barrett2018}, but can be readily used for modeling other (often simpler) systems like quantum fields on a fixed spacetime lattice. In this work, the focus was kept on a special class of discrete dynamical systems whose configuration space is a vector space of finite dimension and whose equations of motion are linear. In the formalism of \textit{discrete canonical evolution} \cite{Dittrich2011, Dittrich2012, Dittrich2013}, one begins with an action additive w.r.t. the individual time-steps and easily obtains the equations of motion in the canonical picture. Thanks to the assumption of linearity, we could effectively study the arising dynamics with standard tools of linear algebra. When present, the irregularity of the system produces constraints and non-uniqueness of the evolution. These are not quite compatible with the standard canonical quantization procedure; one therefore needs to apply extra measures to correctly define the corresponding quantum system.

The present work introduces two fundamentally different models of quantum systems experiencing discrete linear evolution. In the first model described in section 4, one a priori limits one's consideration to the constraint surfaces and reduces all the non-uniqueness beforehand, thus effectively passing to an alternative system in which the evolution is fully \textit{regular}. In this reduced system, each point $ y_{n} $ in the representative time-slice phase space $ \dot{\mathcal{D}}_{n} $ for any $ n \in \{ 0, ..., t \} $ gives rise to a unique solution $ y \in \dot{\mathcal{S}}_{0t} $ and the symplectic form is conserved. Hence it is no problem to define the quantum theory canonically in the standard way. The problem lies elsewhere: in case that the irregularity describes real physical phenomena, the replacement of the original irregular system with a regular one will strip the model of its ability to describe these phenomena, and thus negatively influence its functionality. We further argue that even when the irregularity is rather technical and one could opt for removing it, the reduced system may suffer from some undesired properties, most importantly the back propagation of constraints and the implied breaking of causality. Last but not least, this approach produces a path integral which depends on the classical solutions and thus looses the usual advantage of supplying an off-shell formulation of the quantum dynamics. We find that much of the described behavior is actually shared by the approach adopted in the previous works \cite{Hoehn2014, Hoehn2014a} on the subject.

The second model given in section 5 aspires to overcome the above stated problematic features. It uses a different strategy: first quantize the system independently at each time-slice and then seek for an evolution mapping between the full neighboring quantum state spaces. The search for a suitable one-step evolution prescription is made harder by the absence of a Hamiltonian function in the canonical formalism, a paradox stemming from the discrete-time nature of the theory. Nevertheless, one can use the classical Hamiltonian evolution map to induce a corresponding map for the quantum case. We show how it can be done in a rather straightforward and constructive way. The procedure has two major ingredients: the knowledge of so-called \textit{adapted coordinates} introduced previously in \cite{Kaninsky2020a} and general transformation formulas for the eigenstates of coordinate and momentum observables supplied by \cite{Kaninsky2020}. Eventually, these had lead us to the desired quantum evolution map $ \mathbb{U}_{n+1}: \Phi^{\times}_{n} \rightarrow \Phi^{\times}_{n+1} $ between neighboring time-slice state spaces of our quantum system. The resulting propagator
\begin{equation}\label{prop2}
_{\mathtt{c}} \langle \beta_{n+1} \vert \mathbb{U}_{n+1} \vert \gamma_{n} \rangle_{\mathtt{c}} = \mathcal{l}(V_{2}^{T} \beta_{n+1}) ~ \abs{\det \Sigma_{r} }^{1/2} ~ (2\pi)^{-q/2} ~ e^{iS_{n+1}(\gamma_{n},\beta_{n+1})}
\end{equation}
agrees with the corresponding expressions in \cite{Hoehn2014, Hoehn2014a} up to differences related to the disparity of the approaches themselves. In particular, we find that the term $ \abs{\det \Sigma_{r} }^{1/2} $ is very well justified and coincidentally much simpler than the corresponding term given in \cite{Hoehn2014}. On the other hand, our propagator includes an additional regularization term $ \mathcal{l}(V_{2}^{T} \beta_{n+1}) $ which is not needed in \cite{Hoehn2014} due to the prior reduction of the Hilbert spaces. Next we show that the map $ \mathbb{U}_{n+1} $ is indeed \textit{not unitary} and requires regularization in order to obtain finite inner products. We offer two respective regularization schemes. The \textit{functional regularization} parametrized by $ \varepsilon > 0 $ shapes the evolved wavefunctions so that they are square integrable, and can be later removed via the limit $ \varepsilon \rightarrow 0 $. The \textit{cylindrical regularization} comes closer to the treatment of \cite{Hoehn2014, Hoehn2014a}, leaving the wavefunctions as they are and instead altering the inner product. Due to the non-unitarity of the evolution, the norms of the evolving states are not constant but only non-increasing. To compensate for this, one may renormalize the states, thus recovering the usual probabilistic interpretation. We point out that in certain situations it can be equally appropriate to leave the states (and probabilities) unrenormalized, reflecting the dissipative nature of the system. Eventually, we show that the one-step evolution maps can be composed to describe evolution across multiple time-slices, which naturally leads to the path integral formulation.

The present work may find various applications, but is primarily designed for modeling dynamics within Regge lattices. As demonstrated in the last section, it straightforwardly applies to quantum fields on a fixed triangulation. The featured example of scalar field is in itself interesting and (with small adjustments or generalizations) eligible for further research. It can be also modified e.g. to the case of fermion or gauge fields whose lattice actions are accessible in the literature \cite{Sorkin1975, McDonald2010a}.  The only limitations are that the configuration space has to be $ \mathbb{R}^{q} $ and the action must be quadratic. As of now, these requirements clearly disqualify gravity, whose case is complicated not only by nonlinear equations of motion but also by triangle inequalities and lattice diffeomorphisms \cite{Hoehn2014a}. Nevertheless, we believe that before addressing the problem of quantum gravity, it is profitable to understand simpler quantum systems and explore in some depth the discrete setup.

\section*{Acknowledgments}
This work was supported by Charles University Grant Agency [Project No. 906419].

\bibliographystyle{unsrt}
\renewcommand{\bibname}{Bibliography}
\bibliography{bibliography}

\end{document}